\documentclass[12pt]{article}
\usepackage[letterpaper,top=2.4cm,bottom=2.4cm,left=2.4cm,right=2.4cm,marginparwidth=.75cm]{geometry}

\usepackage{titling}

\usepackage{fancyhdr}
\usepackage{natbib} 
\usepackage{comment}
\usepackage{titling}
\usepackage{mathrsfs}
\usepackage{amsmath}
\usepackage{amsfonts}
\usepackage{bbm} 
\usepackage{cancel}
\usepackage{amsthm}
\usepackage{tabularx} 
\usepackage{geometry}
\usepackage{bm}
\usepackage{mathtools}

\usepackage{placeins}
\usepackage[font=small]{caption}
\usepackage{amsmath}
\usepackage{caption}
\usepackage{comment}
\usepackage{mathrsfs}
\usepackage{floatrow}
\usepackage{amsfonts}
\usepackage{bbm} 
\usepackage{cancel}
\usepackage{tabularx} 
\usepackage{bm}
\usepackage{mathtools}
\usepackage{xcolor,colortbl}
\usepackage{url} 
\usepackage{floatrow}
\newfloatcommand{capbtabbox}{table}[][\FBwidth]
\usepackage{wrapfig}
\usepackage{graphics}
\usepackage{subcaption}

\usepackage{times}
\usepackage{bm}
\usepackage{natbib}

\graphicspath{{./figures/}}

\usepackage[plain,noend]{algorithm2e}

\usepackage{url} 
\usepackage{wrapfig}
\usepackage{subcaption}
\usepackage{setspace}

\usepackage{hyperref}
 \usepackage{authblk}

\newtheorem{theorem}{Theorem}
\newtheorem{proposition}{Proposition}
\newtheorem{corollary}{Corollary}
\newtheorem{lemma}{Lemma}
\newtheorem{example}{Example}
\newtheorem{definition}{Definition}

\usepackage{xcolor,colortbl}
\definecolor{Gray}{gray}{0.85}
\newcolumntype{a}{>{\columncolor{Gray}}X}

\def \X {\mathbb{X}}

\def \e {\text{exp}}
\def \R {\mathbb{R}}

\def \i {\textbf{i}}
\def \j {\textbf{j}}

\def \X {\mathbb{X}}
\def \P {\mathbb{P}}
\def \p {\textbf{p}}
\def \q {\textbf{q}}
\def \Part {\mathcal{P}}
\def \E {\text{E}}

\def \e {\text{exp}}
\def \Corr {\mbox{\rm corr}}
\def \R {\mathbb{R}}

\def \i {\textbf{i}}
\def \j {\textbf{j}}

\def \d {\mathrm{d}}
\def \simiid {\overset{\text{i.i.d.}}{\sim}}

\allowdisplaybreaks

\usepackage[plain,noend]{algorithm2e}

\makeatletter
\renewcommand{\algocf@captiontext}[2]{#1\algocf@typo. \AlCapFnt{}#2} 
\def\@algocf@capt@plain{top}
\renewcommand{\algocf@makecaption}[2]{%
	\addtolength{\hsize}{\algomargin}%
	\sbox\@tempboxa{\algocf@captiontext{#1}{#2}}%
	\ifdim\wd\@tempboxa >\hsize
	\hskip .5\algomargin%
	\parbox[t]{\hsize}{\algocf@captiontext{#1}{#2}}
	\else%
	\global\@minipagefalse%
	\hbox to\hsize{\box\@tempboxa}
	\fi%
	\addtolength{\hsize}{-\algomargin}%
}
\makeatother

\title{{\huge \textbf{Nonparametric priors with\\ full-range borrowing of information}}}
\author{Filippo Ascolani}
\author{Beatrice Franzolini} 
\author{Antonio Lijoi}
\author{Igor Pr\"unster}
\affil{{ Bocconi Institute for Data Science and Analytics (BIDSA), \\Bocconi University Milano, Italy}}
\date{ }

\begin{document}

\maketitle

\vspace{-60pt}
	\begin{abstract} 
\noindent Modeling of the dependence structure across heterogeneous data is crucial for Bayesian inference since it directly impacts the borrowing of information. Despite the extensive advances over the last two decades, most available proposals allow only for non--negative correlations. We derive a new class of dependent nonparametric priors that can induce correlations of any sign, thus introducing a new and more flexible idea of borrowing of information. This is achieved thanks to a novel concept, which we term \textit{hyper-tie}, and represents a direct and simple measure of dependence. We investigate prior and posterior distributional properties of the model and develop algorithms to perform posterior inference. Illustrative examples on simulated and real data show that our proposal outperforms alternatives in terms of prediction and clustering.\\

\end{abstract}

	\section{Introduction}\label{s:introduction}
	Bayesian nonparametric methods are increasingly popular, mainly thanks to their flexibility and strong foundations. The most common assumption underlying Bayesian models is 
	exchangeability, which corresponds to invariance of the joint distribution of the observations with respect to finite permutations. However, real phenomena often present a level of heterogeneity that makes exchangeability  unrealistic: collected data may refer to different 
	features, populations, or, in general, may be collected under different experimental conditions. Such situations entail a significant level of heterogeneity 
	and opportunities for borrowing information, that can be exploited through the notion of partial exchangeability, which implies exchangeability within each experimental condition, but not across.
 Two sequences of observations ${X}=(X_i)_{i\geq1}$ and $Y=(Y_j)_{j\geq1}$, taking values in a space $\mathbb{X}$, are partially exchangeable if and only if, for all sample sizes $(n, m)$ and all permutations $(\pi_1, \pi_2)$,
	\[
	\bigl((X_{i})_{i=1}^{n},(Y_{j})_{j=1}^{m}\bigr)\overset{d}{=} \bigl((X_{\pi_{1}(i)})_{i=1}^{n},(Y_{\pi_{2}(j)})_{j=1}^{m}\bigr).
	\]
	 with $\overset{d}{=}$ denoting equality in distribution. 
	 From an inferential point of view, partial exchangeability entails that the order of the observations within each sample is non-informative, while the belonging to a specific sample is relevant and has to be taken into account. Moreover, by de Finetti's representation theorem \citep{DeFinetti1938} $X$ and $Y$ are partially exchangeable if and only if there exist random probabilities $(\tilde p_1, \tilde p_2)$ 
	 such that for any $i,j=1,\ldots,n$ 
	\begin{equation}\label{partial_model}
		\begin{aligned}
			(X_i,Y_j)\mid(\tilde p_1,\tilde p_2)  \overset{iid}{\sim} \tilde p_1\times\tilde p_2\qquad
			(\tilde p_1,\,\tilde p_2) \sim Q
		\end{aligned}
	\end{equation}
	with $Q$ 
	playing the role of the prior. 
	The dependence induced by $Q$ at the level of the observables defines the Bayesian learning mechanism and it connects to the notion of borrowing of information. This term was first coined by John Tukey \citep{brillinger2002john} and popularized with reference to Stein's paradox and empirical Bayes techniques in \cite{efron1977stein}. More generally, statisticians refer to borrowing of information when many samples contribute to inference related to just one sample. Imagine collecting the samples $(X_{i})_{i=1}^{n}$ and $(Y_{j})_{j=1}^{m}$, while being interested only in the parameter $\tilde p_1$ associated to $X$. The simplest approach could be to disregard the second sample $(Y_{j})_{j=1}^{m}$, with the drawback of losing potentially useful information. The typical borrowing instead consists in shrinking the estimates for different samples towards each other:  {shrinkage is justified by the fact that distributions of different, but related, populations are expected to be similar in terms of shape and/or location. However, many contexts may still require borrowing of information between $(X_{i})_{i=1}^{n}$ and $(Y_{j})_{j=1}^{m}$, but without necessarily resulting in shrinkage.
	Indeed, one's 
	available prior 
	information may imply that the responses in different groups have a negative association 
	and, thus, tend to be dissimilar in location, which makes shrinkage undesirable. Similarly, when 
	there is no pre-experimental knowledge on the 
	dependence between $X_i$ and $Y_j$, a flexible prior specification allowing also for negative association 
	would be more appropriate.
	A toy parametric example to further clarify that borrowing does not necessarily imply classic shrinkage is provided in Section S2.~of the Appendix. Some applied scenarios of borrowing of information not resulting in shrinkage are, for instance, the study of survival times and abundances of competitive species \citep{lee2020spatial}, the incorporation of retrospective data to study associations between biomarkers \citep{gong2021unpaired}, the association between dental caries and dental fluorosis \citep{lorenz2018inferring}, the analysis of stocks and bonds returns \citep[see][and Section~\ref{s:illustration1}]{bhardwaj2013business}, and the clustering of multivariate responses with missing entries (see Section~\ref{application}). In this paper we introduce a class of nonparametric priors that allows for a more general version of borrowing, which includes shrinkage as a special case. These can be used as core building blocks for models tailored to specific applications.}

	\noindent Starting from the pioneering works of \cite{cifarelli1978problemi} and 
	\cite{maceachern1999dependent,maceachern2000dependent}, Bayesian nonparametric contributions for non--exchangeable data have grown substantially, see 
	\cite{foti2013survey}, \cite{muller2015bayesian} and \cite{QuintanaDDP} for insightful reviews. The vast majority of nonparametric models for partially exchangeable data entails that the random probabilities in \eqref{partial_model} are such that
		\begin{equation}\label{species_sampling}
		\begin{cases}
			\tilde p_1\overset{a.s.}{=}\sum_{k\geq1}\bar{J}_k\delta_{\theta_k}\\[6pt]
			\tilde p_2\overset{a.s.}{=}\sum_{k\geq1}\bar{W}_k\delta_{\phi_k}
		\end{cases}
		\quad\quad  
		\theta_k \simiid P_0, \quad \phi_k \simiid P_0
	\end{equation}
	where the random weights $\left( (\bar{J}_k), (\bar{W}_k )\right)$ and the atoms $\left( (\theta_k), (\phi_k)\right)$ are independent and $\theta_k\perp \phi_h$ for $k\neq h$.
	In this paper we focus on this class of models and, for ease of exposition, take $\tilde p_1$ and $\tilde p_2$ with the same marginal distribution. 
	
	A first prominent strategy for defining $Q$ is to explicitly assign the distribution of the weights and the atoms in \eqref{species_sampling} so to create dependence between $\tilde p_1$ and $\tilde p_2$: 
	this approach has led to dependent Dirichlet processes \citep{maceachern1999dependent,maceachern2000dependent,QuintanaDDP}, dependent stick-breaking processes, kernel stick-breaking processes \citep{DunsonStickBreaking}, probit stick-breaking processes \citep{DunsonProbitStickBreaking} and others. Despite their flexibility and the availability of posterior sampling schemes, the derivation of analytical results is very difficult for these models; it is often not clear how the dependence of the series reflects at the level of the observables and therefore such methods may lack transparency.
	
	A second popular strategy, analytically more tractable, relies on completely random measures (CRMs) either working directly on the law of multi-dimensional vectors of CRMs \citep{epifani2010nonparametric,griffin2017compound, RivaPalacio2021} or combining conditionally independent CRMs, using additive structures \citep{muller2004method,griffin2013comparing,lijoi2014class,lijoi2014bayesian,lijoi2014dependent}, nested structures \citep{rodriguez2008nested,camerlenghi2019latent}, or hierarchical structures \citep{teh2006hierarchical,camerlenghi2019distribution}. CRMs are then suitably transformed to obtain the  random probabilities in \eqref{species_sampling}.
	
	Dependent random probabilities clearly induce dependence across groups of observations. The simplest and most intuitive way to quantify the dependence structure is through correlations. Therefore, when considering correlations among observables, we will implicitly assume real-valued $X_i$'s and $Y_j$'s, namely $\mathbb{X}=\mathbb{R}$. All other results and concepts are valid for general spaces $\mathbb{X}$. 
	A first result in this direction shows that, 
	regardless of the specific dependent model, observations in different groups cannot be more correlated (in absolute sense) than the ones in the same group.
	 
	\begin{proposition}\label{corr_between_across} 
		Suppose $X$ and $Y$ are	partially exchangeable sequences, 
		such that $\tilde p_1$ and $\tilde p_2$ in \eqref{partial_model} have the same marginal distribution. Then
		\[
		-\Corr (X_{i},X_{i'}) \leq \Corr (X_{i},Y_{j}) \leq \Corr (X_{i},X_{i'}),
		\]
		for any $i, i'$ and $j$.
	\end{proposition}
	 
	Due to exchangeability within each group, the upper bound in Proposition~\ref{corr_between_across} 
	is always non--negative and it can be shown that, for
	all the models 
	as in \eqref{species_sampling}, the correlation between observations in the same sample, $\Corr (X_{i},X_{i'})$, is 
	determined by the probability of a tie. 
	As for the correlation across samples $\Corr (X_{i},Y_{j})$, we show that 
	a similar result holds true, 
	with \emph{hyper-ties}, the new notion we introduce, replacing ties.   
	
	Moreover, note that for most models based on CRMs, which allow for the computation of the correlation, $\Corr (X_{i},Y_{j})$ turns out to be positive: this happens in particular when the interaction between two or more groups is of interest. Therefore, the literature available to date within the partially exchangeable setting is focused on models that attain a limited range of possible values of the correlation, when it can be evaluated.
	Here we aim to overcome this limitation and introduce a novel class of priors which yield a wider range of correlation values among the observables, including those with negative sign. 
	The next result shows that the sign of the correlation is only determined by the dependence structure between the atoms. 
	 
	\begin{proposition}\label{positive_corr}
	Suppose $X$ and $Y$ are 
	partially exchangeable sequences, such that the underlying $\tilde p_1$ and $\tilde p_2$ are as in \eqref{species_sampling}. Moreover, for any $k$ and $k'$, let $\Corr(\theta_k, \phi_{k'}) \geq 0$. Then $\Corr \,(X_{i},Y_{j}) \geq 0$, for any $i$ and $j$.
	\end{proposition}
	 
	For instance, hierarchical processes \citep{teh2006hierarchical,camerlenghi2019distribution}, which represent one of the most popular dependent models, induce dependence by the sharing of atoms across groups. However, by Proposition \ref{positive_corr}, this means that achieving negative correlation is impossible.
	Hence, a flexible joint distribution for the sequence of atoms must be specified. This task is accomplished by our proposal, termed normalized CRMs with Full-Range Borrowing of Information (n-FuRBI), that allows to attain any possible value for the correlation specified in Proposition \ref{corr_between_across}. Moreover, it encompasses many previous 
	constructions as special cases. We will show that it nicely combines the flexibility of the random series construction with the analytical tractability featured by CRMs.
	Our proposal allows to consider any interesting choice of borrowing of information: independence, classical shrinkage, but also repulsion of estimates for different samples, generating what we term \textit{full--range borrowing of information}. Note that the repulsive behaviour of n-FuRBI is different from the one featured by the
	priors introduced in \cite{petralia2012repulsive} and \cite{QuintanaRepulsive}, that induce repulsion among the atoms of a single random probability measure.
	
    The appendix includes all the analytical derivations and proofs, the simulation algorithms for the implementation of the proposed class of models, additional examples and numerical studies. In the following we use the prefix S to indicate sections of the Appendix. The code to allow full replication of the numerical results is available at \url{https://github.com/beatricefranzolini/FuRBI}.  
	
	\section{General results on dependent processes}\label{general_dep}
	
The vast majority of dependent processes introduced in the literature are almost surely discrete and therefore admit a series representation as in \eqref{species_sampling}. A key preliminary step leading to the definition of hyper-tie and n-FuRBI priors is the observation that the random probabilities in \eqref{species_sampling} can be embedded into
		\begin{equation}\label{series_normCRV}
		\begin{cases}
			p_1\overset{a.s.}{=}\sum_{k\geq1}\bar{J}_k\delta_{(\theta_k,\phi_k)}\\
			p_2\overset{a.s.}{=}\sum_{k\geq1}\bar{W}_k\delta_{(\theta_k,\phi_k)}
		\end{cases} \quad (\theta_k, \phi_k) \overset{i.i.d.}{\sim} G_0,
	\end{equation}
with $G_0$ a probability distribution on $\X \times \X$, whose marginals equal $P_0$. While $p_1$ and $p_2$ share the same atoms, the weights and the atoms are independent and 
the pair of random probability measures $\tilde p_1$ and $\tilde p_2$ in \eqref{species_sampling}  are obtained as the projections over different coordinates of $p_1$ and $p_2$, namely $\tilde p_1 (\cdot) = p_1(\cdot \times \X)$ and $\tilde p_2 (\cdot) = p_2(\X \times \cdot)$. The structure of popular models is recovered by letting either $G_0=P_0^2$, which corresponds to independence, or $G_0(\d \theta,\d \phi)=P_0(\d \theta)\delta_{\{\theta\}}(\d \phi)$, that is  $\theta_k=\phi_k$ for any $k$ as happens for, e.g., hierarchical processes \citep[see][]{camerlenghi2019distribution}.
Almost sure discreteness implies that a sample from the random probability measure $\tilde p_1$ (or $\tilde p_2$) 
will display ties with positive probability. The probability of a tie, i.e. a coincidence of any two observations $i$ and $j$ in the same sample, is 
	\begin{equation}
	 \label{eq:beta}
	\beta :=\P(X_i=X_j)=\sum_{k \geq 1} \E(\bar{J}^2_k) = \sum_{k \geq 1} \E(\bar{W}^2_k)=\P(Y_i=Y_j)
	\end{equation}
with $(\bar{J}_k)_{k\ge 1}$ and $(\bar{W}_k)_{k\ge 1}$ equal in distribution since we are assuming, for simplicity, that $\tilde p_1$ and $\tilde p_2$ are equal in distribution. 
When considering jointly the two samples, the concept of tie can be replaced by the one of \emph{hyper-tie}, that is two observations in different samples
 coinciding with components having the same label. According to \eqref{partial_model}, its probability is
	\begin{equation}
		\label{eq:gamma}
	\gamma:= \sum_{k\ge 1}\P(X_i=\theta_k,\,Y_j=\phi_k)=\sum_{k \geq 1} \E(\bar{J}_k\bar{W}_k).
	\end{equation}
Sampling 
from components with the same label is equivalent to sampling the same atom at the level of the underlying $(p_1, p_2)$ in \eqref{series_normCRV}. 
Clearly, when the atoms are shared between $\tilde p_1$ and $\tilde p_2$, i.e. $G_0(\d \theta,\d \phi)=P_0(\d \theta)\delta_{\{\theta\}}(\d \phi)$, a hyper-tie corresponds to an actual tie between observations in different samples.

The next result shows the relationship between $\beta$ and $\gamma$, the probabilities of a tie and hyper-tie, respectively: in particular, the probability of a tie is always larger and equality is attained if and only if the probability masses of $p_1$ and $p_2$ are almost surely equal.  
  
\begin{proposition}\label{beta_gamma}
		Let $(\tilde p_1,\tilde p_2)$ be as in \eqref{species_sampling} and $\beta, \gamma$ as in \eqref{eq:beta} and \eqref{eq:gamma}, respectively. Then $0 \leq \gamma \leq \beta$ and $\beta = \gamma$ if and only if $\bar{W}_k \overset{a.s.}{=} \bar{J}_k$ for any $k$.
\end{proposition}
 
	Hyper-ties play a crucial role in determining the dependence between observables across groups, as the ties do for the dependence between observables within groups, as shown by the next proposition.

	\begin{proposition}\label{corr}
		Consider model \eqref{partial_model} with $(\tilde p_1, \tilde p_2)$ as in \eqref{species_sampling}. 
		Then, for any $i\neq i'$ and any $j \neq j'$
		\[
		\Corr(X_i, X_{i'}) = \Corr(Y_j, Y_{j'}) = \beta \qquad \qquad \Corr(X_i, Y_j) = \gamma \, \rho_0 
		\]
		with $\rho_0$ the correlation between two random variables jointly sampled from $G_0$.
	\end{proposition}
	 
Thus, while the correlation between observations in the same sample equals the probability of a tie, the correlation between observations from different samples is determined by the probability of a hyper-tie, corrected by the correlation between atoms. Clearly a suitable choice of the joint distribution of the atoms makes the latter negative. 
Thus, by choosing $G_0$ appropriately, for instance as a bivariate normal, it is easy to tune the correlation according to the available prior knowledge. The following Corollary shows the values that can be attained, once the marginal law is specified.
 
	\begin{corollary}\label{values_corr}
		Consider model \eqref{partial_model} with  $(\tilde p_1, \tilde p_2)$ as in \eqref{species_sampling}. 
		If the marginal distribution of $\tilde p_1$ and $\tilde p_2$ is fixed, then
		$\Corr(X_i, Y_j) \in [-\beta, \beta]$
		and the extreme values are attained if and only if the jumps are equal and $\rho_0 = \pm 1$.
	\end{corollary}
	 
	Unsurprisingly, with equal weights and jumps, which corresponds to full exchangeability, one achieves the extreme case of $\Corr(X_i, Y_j) =\beta$.
	Null correlation, instead, is attained when atoms are uncorrelated or when the probability of hyper-ties is zero. Lastly, maximum negative correlation $\Corr(X_i, Y_j) = -\beta$ 
	is attained with equal weights and 	negatively correlated atoms 
	and can be thought of as the opposite case with respect to exchangeability, at least in terms of correlation. Ties and hyper-ties play a similar role also in the predictive structure, as the next result shows.
	 
	\begin{proposition}\label{first_pair}
		Consider model \eqref{partial_model} with $(\tilde p_1, \tilde p_2)$ as in \eqref{species_sampling}. 
		Then 
		\[
		\P \left(X_1 \in A, X_2 \in B \right) = \beta P_0(A \cap B) + \left(1-\beta \right)P_0(A) P_0(B).
		\]
		and
		\[
		\P \left(X_1 \in A, Y_1 \in B \right) = \gamma G_0(A \times B) + \left(1-\gamma \right)P_0(A) P_0(B).
		\]
	\end{proposition}
	 
	The result is indeed quite intuitive. If $X_1$ and $Y_1$ form a hyper-tie (with probability $\gamma$) they come from the same pair of atoms and need to be sampled jointly; otherwise they refer to different atoms and are sampled independently. The same happens inside each group, where $X_1$ and $X_2$ are equal with probability $\beta$.
	 
	\begin{example}
		\label{HDPexample}
		{\rm
		The hierarchical Dirichlet process \citep{teh2006hierarchical} is characterized by the hierarchical representation
		$\tilde p_i \mid \tilde p_0 \overset{\text{i.i.d.}}{\sim} \text{DP}(\theta, \tilde p_0)$, with $\tilde p_0 \sim  \text{DP}(\theta_0, P_0)$, where $P_0$ is a diffuse measure and DP$(\alpha, H)$ denotes the law of a Dirichlet process with concentration parameter $\alpha > 0$ and baseline distribution $H$. Since the $\tilde p_i$'s share the atoms, an hyper-tie corresponds to an actual tie between observations in different samples, so that with simple computations we get 
		\[
		\beta =  \Corr(X_i, X_j) = 1-\frac{\theta\theta_0}{(1+\theta)(1+\theta_0)},
	\qquad
		 \gamma = \Corr(X_i, Y_j) = \frac{1}{1+\theta_0}.
		\]
		Thus, the correlation among the observables is forced to be positive, with $\theta_0$ tuning the dependence; see Example $1$ in \cite{camerlenghi2019distribution} for more details.
	}
	\end{example}
	 
	Given the above results and considerations, it should be clear that $\gamma$ defined in \eqref{eq:gamma}  is crucial for tuning the level of dependence. However, closed form expressions of $\gamma$ are available only for a few cases 
	and, in fact, we are facing a trade--off: on the one hand we have dependent processes based on the stick-breaking representation, that allow for high flexibility while sacrificing the availability of analytical results; on the other hand we have constructions based on CRMs, for which an extensive theory has been developed, though they are not as effective for 
	tuning the dependence, since all the existing instances produce non-negative correlation across samples. In the following we combine the best of both approaches through n-FuRBI: they are flexible processes that can attain any value for the correlation between the observables, while at the same time a posterior representation can be derived. Their construction is based on CRMs and completely random vectors, reviewed in the next section. 
	
		\section{
		Some basics on completely random measures}\label{s:crm}
	As shown in \cite{lijoi2010models}, many Bayesian nonparametric models can be obtained as suitable transformations of CRMs; among others, these include the Dirichlet process, the Pitman-Yor process and the neutral-to-the-right priors. 
	The extension of CRMs to the bivariate setting is provided by \textit{completely random vectors} $\mu=(\mu_1,\mu_2)$, whose components take values in the space of boundedly finite measures on $\mathbb{X}$ and are such that, for every collection of pairwise disjoint sets $(A_i)_{i\geq1}^{n}$, the random vectors 
		$(\mu_1(A_1), \mu_2(A_1)), \ldots,(\mu_1(A_n),\mu_2(A_n))$
		are mutually independent. 
	We focus on the case of no fixed atoms and no deterministic component, so that the marginal CRMs $\mu_1$ and $\mu_2$ are almost surely discrete and can be written as sum of {$\X$--valued} random atoms with random weights, i.e.
	\[
	\mu_1 \overset{a.s.}{=}\sum_{i\geq1}J_i\delta_{\kappa_i}, \quad \mu_2 \overset{a.s.}{=} \sum_{i\geq1}W_i\delta_{\kappa_i}.
	\]
	{In the following section it will be convenient to use the reparametrization $\kappa_i=(\theta_i,\phi_i) \in \X=\X_1 \times \X_2$.}
	Such completely random vectors are characterized by the L\'evy-Khintchine representation
	\begin{equation}\label{Levy_bivariate}
		\mathbb{E}\left\{e^{-\mu_1(f_1)-\mu_2(f_2)}\right\} = 
		\exp\left[-\int\limits_{\R^2_+ \times \X} \{1-e^{-s_1\,f_1(x)-s_2\,f_2(x)}\}\, v(\d s_1, \d s_2, \d x)\right]
	\end{equation}
where $\mu_i(f_i)=\int_{\mathbb{X}}f_i(x)\mu_i(\d x)$ for $\mathbb{R}^+$-valued $f_i$ and
 $v(\d s_1,\d s_2,\d x)$ is the joint L\'evy intensity. We shall focus on the homogeneous case, in which jumps $(J_j)_{j\geq1}$ and locations $(X_j)_{j\geq1}$ are independent. In terms of L\'evy intensity it reads
$v(\d s_1,\d s_2,\d x)= \rho(\d s_1, \d s_2) \alpha(\d x)$
for some finite measure $\alpha$ on $\X$ and measure $\rho$. 
Moreover, in the sequel we will also need the joint and marginal Laplace exponents given by
	\begin{align*}
	\psi_b(\lambda_1, \lambda_2) :=& \int\limits_{\mathbb{R}^2_+\times\mathbb{X}}(1-e^{-\lambda_1 s_1-\lambda_2 s_2}) \rho(\d s_1, \d s_2) \alpha(\d x), \quad \lambda_1 > 0, \lambda_2 > 0.\\
	\psi(\lambda) :=& \int\limits_{\mathbb{R}_+\times\mathbb{X}}(1-e^{-\lambda s}) \rho(\d s) \alpha(\d x) \quad \lambda > 0,
	\end{align*}
 
	For an exhaustive account on CRMs, we refer to \cite{kingman1967completely,kingmanpoisson}.
	Completely random vectors and CRMs are often normalized to obtain random probability measures, as introduced in \cite{regazzini2003distributional}, i.e.
$p(\cdot) = {\mu(\cdot)}/{\mu(\X)}$.
Notice that in principle any random measure $\mu$ such that $\mathbb{P}(0 < \mu(\X) < \infty)=1$ can be normalized in order to define a random probability measure. However, the strength of completely random vectors and measures lies in their L\'evy--Khintchine representations and unique correspondence with the associated L\'evy intensity, which 
	allow a high degree of analytical tractability. CRMs and the corresponding normalized probabilities have been extensively studied to model exchangeable data \citep[see, for instance,][]{james2006conjugacy,james2009posterior,james2010posterior,lijoi2010models,favaro2016stick,camerlenghi2018bayesian}. Similarly, a completely random vector can be used to model dependence between two groups. For more details on completely random vectors and an interesting account of their dependence structure, we refer to \cite{catalano2021measuring,catalano2023}. Since the two measures in the vector share all the 
	atoms, by virtue of Proposition~\ref{positive_corr} 
	the induced 
	model 
	yields non--negative correlation between samples. The issue is addressed in the next section, by means of a novel class of random probability measures that leverage the dependence structure specifed for the atoms.
	
	\section{Full-range borrowing of information nonparametric prior}\label{s:prior}
	\subsection{Definition and first properties}
	In this section we introduce n-FuRBI and 
	for simplicity we still consider only the case of two samples with the same a priori marginal distribution.
	 
	\begin{definition}\label{def:aCRM}
		{\rm 
		Consider a completely random vector $(\mu_1,\mu_2)$ on $\X^2$ with L\'evy intensity
		\[
		v(\d s_1, \d s_2, \d x_1, \d x_2) = \rho(\d s_1, \d s_2)\: \alpha(\d x_1, \d x_2),
		\]
		where $\alpha(\d x_1, \d x_2) = \theta G_0(\d x_1, \d x_2)$, where $\theta =\alpha(\X^2)\in (0,+\infty)$, and $G_0$ is a non-atomic probability measure on $\X^2$ such that $G_0(\cdot \times \X) = G_0(\X \times \cdot) = P_0(\cdot)$. Then $\tilde \mu_1$ and $\tilde \mu_2$  defined as
		\begin{equation*}
			\tilde\mu_1(\cdot)=\mu_1(\mathbb{X}\times\cdot) \qquad \tilde\mu_2(\cdot)=\mu_2(\cdot\times\mathbb{X})
		\end{equation*}
		are CRMs with \textit{Full-Range Borrowing of Information} (FuRBI CRMs) 
		and underlying L\'evy intensity $v$. The normalized versions $\tilde p_j(\cdot) = \tilde \mu_j(\cdot)/\tilde\mu_j(\mathbb{X})$ for $j=1,2$ are said \textit{normalized CRMs with Full-Range Borrowing of Information} (n-FuRBI).
	}
	\end{definition}
	 
	Essentially, first a pair of random measures endowed with the same locations is constructed on the product space $\X^2$; 
	as a second step, the coordinates of each pair of atoms are split. Thus, the n-FuRBI admit a representation as in \eqref{species_sampling} and \eqref{series_normCRV}. In general FuRBI CRMs are not completely random vectors, because the joint sampling of the atoms forbids the independence 
	of the vector evaluated on pairwise disjoint sets. 
	However, the representation in terms of a completely random vector in the product space is useful to characterize the joint law of the FuRBI CRMs, as shown in the following proposition.
	 
	\begin{proposition}
		Let $(\tilde \mu_1, \tilde \mu_2)$ be a vector of FuRBI CRMs. Then
		\begin{itemize}
			\item[{\rm (i)}] 
			$\tilde \mu_1$ and $\tilde \mu_2$ are CRMs with intensity $\rho(\d s)\theta P_0(\d x)$, where $\rho(\d s) = \int _{\R_+}\rho(\d s_1, \d s)$.
			\item[{\rm (ii)}] For any $A$ and $B$, the following equality holds
			\begin{multline*}
				\begin{aligned}
					\E \bigl[ \mathrm{e}^{ -\lambda_1\tilde{\mu}_1(A) -\lambda_2\tilde{\mu}_2(B)} \bigr] =
					\exp\{-G_0(A \times B^c)\psi(\lambda_1)-G_0(A^c \times B)\psi(\lambda_2)\}\\
					\times\:\exp\{ -G_0(A \times B){\psi_b}(\lambda_1, \lambda_2) \},
				\end{aligned}
			\end{multline*}
where $\psi$ denotes the common marginal Laplace exponent and $\psi_b$ the joint Laplace exponent of $(\mu_1, \mu_2)$.
			\item[{\rm (iii)}] The joint law of $(\tilde \mu_1,\tilde \mu_2)$ is characterized by the joint L\'evy intensity of  $(\mu_1, \mu_2)$.
		\end{itemize}
	\end{proposition}
	 
	The next proposition shows that the $\beta$ and $\gamma$ associated to any couple of n-FuRBI can be computed through their Laplace exponents.
	 
	\begin{proposition}\label{prob_tie}
		Consider $(\tilde p_1, \tilde p_2)$ n-FuRBI. Then the probability of a tie and of a hyper-tie are respectively
		\[
		\beta = -\int_{\R_+}u\left\{\frac{\d^2}{\d u^2} \psi(u) \right\}e^{-\psi(u)}\, \d u,
		\quad  \gamma = -\int_{\R^2_+}\left\{\frac{\partial^2}{\partial u_1\partial u_2} \psi_b(u_1, u_2) \right\}e^{-\psi_b(u_1, u_2)}\, \d u_1 \d u_2.
		\]
	\end{proposition}
	 
	Thus, the crucial value of $\gamma$ can be obtained by computing, analytically or numerically, a bivariate integral. The two results above show a recurrent trait of our approach: interesting quantities will be usually rewritten in terms of the original completely random vector, in order to exploit its analytical tractability. We conclude this section with two examples of FuRBI CRMs, that also show how some existing constructions can be obtained as special cases.
	 
	\begin{example}[FuRBI CRMs with equal jumps]{\rm 
		Let
		$\rho(\d s_1)\delta_{s_1}(\d s_2) \, \theta\, G_0(\d x_1,\d x_2)$
		be the underlying L\'evy intensity. The series representation of the corresponding FuRBI CRMs is
		\[
		\tilde \mu_1(\cdot) \overset{a.s.}{=} \sum_{k\geq1}W_k\delta_{\theta_k}\qquad\qquad 
		\tilde \mu_2(\cdot) \overset{a.s.}{=} \sum_{k\geq1}W_k\delta_{\phi_k}\qquad \qquad \text{with}\enskip (\theta_k,\phi_k) \overset{i.i.d}{\sim}G_0.
		\]
		Therefore, $\gamma=\beta$, so that a tie and a hyper--tie are observed with the same 
		probability. 
		}
	\end{example}

	\begin{example}[Extended Compound FuRBI CRMs]\label{compound}
		{\rm Consider the 
		L\'evy intensity
		\[
		v(\d s_1, \d s_2, \d x_1, \d x_2) =\int z^{-2}h(s_1/z,\,s_2/z) \, \d s_1 \d s_2v^*(\d z)\, \theta\, G_0(\d x_1,\d x_2),
		\]
		where $h$ is some density 
		and $v^*$ is a L\'evy intensity that satisfies
		\[
		\int z^{-2}\int \min \{1, ||{s}|| \}h(s_1/z,\,s_2/z) \, \d s_1 \d s_2 v^*(\d z) < \infty, \quad ||{s}|| = \sqrt{s_1^2+s_2^2}.
		\]
		The series representation of the corresponding FuRBI CRMs is 
		\[
		\tilde \mu_1(\cdot) \overset{a.s.}{=} \sum_{k\geq1}m_{1,k}W_k\delta_{\theta_k}\qquad\qquad 
		\tilde \mu_2(\cdot) \overset{a.s.}{=} \sum_{k\geq1}m_{2,k}W_k\delta_{\phi_k}
		\]
		where $(\theta_k,\phi_k) \overset{i.i.d}{\sim}G_0$ and $(m_{1,k},m_{2,k})\overset{iid}{\sim}h$. When $G_0$ is degenerate on the main diagonal, one retrieves the class of compound random measures introduced by \cite{griffin2017compound}. }
	\end{example}

	\subsection{Correlation structure between n-FuRBI}
	In order to analyze the dependence between the marginal n-FuRBI priors $\tilde p_1$ and $\tilde p_2$, it is useful to compute the correlation of the random probability measures evaluated on the same set $A$. In all the existing CRM-based models such a correlation does not depend on the specific set considered and, hence, it is often used as a global measure of dependence. The next proposition provides the covariance structure between two n-FuRBI.
	 
	\begin{proposition}\label{t:correlation}
		Let $\tilde p_1$ and $\tilde p_2$ be n-FuRBI. Then for any $A, B$, such that $0\leq P_0(A)\leq1$ and $0\leq P_0(B)\leq1$, we have
		$\mbox{\rm cov}(\tilde p_1 (A) , \tilde p_2 (B) ) = \gamma \left[G_0 (A \times B) - P_0(A)P_0(B) \right]$
		and
		\[
		\Corr(\tilde p_1 (A) , \tilde p_2 (B) ) = \frac{\gamma}{\beta} \frac{G_0 (A \times B) - P_0(A)P_0(B)}{\sqrt{P_0(A)(1-P_0(A))P_0(B)(1-P_0(B))}}.
		\]
	\end{proposition}
	 
	By setting $A = B$, from the previous results one immediately deduces that
	$\mbox{\rm cov}(\tilde p_1 (A) , \tilde p_2 (A) ) = \gamma \left[G_0 (A \times A) - P_0(A)^2 \right]$ and
	\[
	\Corr(\tilde p_1 (A) , \tilde p_2 (A) ) = \frac{\gamma}{\beta} \frac{G_0 (A \times A) - P_0(A)^2}{P_0(A)(1-P_0(A))}.
	\]
	Unlike what usually happens with existing models, here the correlation can be negative, when $A$ is such that $G_0 (A \times A) < P_0(A)^2$, that is when $G_0$ exhibits a repulsive behaviour between the coordinates in $\X^2$. Moreover, the correlation depends on the specific set on which the two measures are evaluated and, therefore, it has to be interpreted as a local measure of dependence. See Section S3.~for an illustration of this phenomenon on sets of the form $(-\infty, x)$. 
	 
	\begin{example}[n-FuRBI with equal jumps]
		{\rm 
		In this case, Proposition~\ref{beta_gamma} 
		entails $\beta=\gamma$. Therefore
		\begin{equation*}
			\Corr \left( \tilde{p}_1(A),\tilde{p}_2(A) \right) = \frac{G_0(A \times A) - P_0(A)^2}{P_0(A) (1-P_0(A))}.
		\end{equation*}
		Moreover, still by virtue of Proposition \ref{beta_gamma}, for a given $G_0$ this is the highest possible correlation in absolute value.
	}
	\end{example}
	 

Proposition~\ref{corr} then provides the correlation between the observables, which is even more important from a modeling perspective. 
 
	\begin{example}[Gamma n-FuRBI with equal jumps]\label{example_equal_jumps}
	{\rm 
		If the common marginal is the law of a Dirichlet process, 
	then
		$\Corr(X_i, Y_j) = \rho_0/(1+\theta)$.
		Choosing appropriately $\rho_0$ and $\theta$ the entire range $(-1,1)$ becomes available.
	}
	\end{example}
	 Note that hyper-ties allow to perform a more general type of borrowing, compared to ties, even when the correlation is positive. While ties are a useful construction to model multiple samples that share certain values/latent parameters, hyper-ties can borrow information even when the two samples have no common values/latent parameter. This aspect will play a crucial role in the data-analyses of Sections~\ref{s:illustration1} and \ref{application}; for these the assumption of common values would be highly unrealistic.
	
	\section{Inference}\label{s:posterior}
	\subsection{Posterior Characterization}
	Having provided an exhaustive description of the a priori properties of n-FuRBI, the following key step is to provide a tractable posterior characterization. Conjugacy is out of question here: even in the exchangeable context it is a property 
	characterizing the Dirichlet process \citep[see][]{james2006conjugacy}. Nevertheless, conditional on a set of suitable latent variables, the posterior distribution of the original completely random vector $(\mu_1, \mu_2)$ turns out to be again a completely random vector leading to a neat posterior characterization and viable methods for sampling. 
	
	Consider a sample of $n$ observations $(X_{i})_{i=1}^{n}$ from $\tilde{p}_1$ with unique values $\underline X_{n}^* = \left( X^*_{1}, \dots, X^*_{k} \right)$ and associated multiplicities $(n_1, \dots, n_k)$; analogously, consider $m$ observations $(Y_{j})_{j=1}^{m}$ from $\tilde{p}_2$ with unique values $\underline{Y}_{m}^* = \left( Y^*_{1}, \dots, Y^*_{c} \right)$ and multiplicities $(m_1, \dots, m_c)$. 
	While it is immediate to check for ties, 
	hyper-ties cannot be 
	identified from the data. To this end, we define a latent random element ${p}$ encoding the hyper-ties, such that ${p} = \left\{ (i_l, j_l) \right\}_l$, where $(i,j)$, with $1 \leq i \leq k$ and $1 \leq j \leq c$, denotes a hyper-tie between $X_i^*$ and $Y_j^*$. Moreover $(i, 0)$, with $1 \leq i \leq k$, denotes that $X_i^*$ does not form a hyper-tie with any value in $\underline{Y}_{m}^*$ and $(0, j)$, with $1 \leq j \leq c$, denotes that $Y_j^*$ does not form an hyper-tie with any value in $\underline{X}_{n}^*$.
	
Therefore, if $(i,j) \in {p}$ with $i \neq 0$ and $j \neq 0$, it means that $X_i^*$ and $Y_j^*$ come from the same pair of atoms in representation \eqref{series_normCRV}. Instead, $(i,0) \in {p}$ implies that $X_i^*$ is the only value associated to a specific pair, and similarly for $Y_j^*$ if $(0,j) \in {p}$. Since we are working with unique values, it is clear that each $X_i^*$ and $Y_j^*$ can form at most one hyper-tie, i.e. it is associated to a unique member of ${p}$. This justifies the following formal definition.
	 
\begin{definition}{\rm 
		We say that ${p} = \left\{ (i_l, j_l) \right\}_l$ is a \textit{compatible hyper-ties structure} for $(X_{i})_{i=1}^{n}$ and $(Y_{j})_{j=1}^{m}$ if, firstly, for any $1 \leq i \leq k$, there exists exactly one $i_l$ such that $i_l = i$, thus each element of $\underline{X}_{n}^*$ forms at most one hyper-tie; secondly, for any $1 \leq j \leq c$, there exists exactly one $j_l$ such that $j_l = j$, thus each element of $\underline{Y}_{m}^*$ forms at most one hyper-tie; lastly, for any $l$, if $i_l = 0$ then $j_l \neq 0$, thus at least one coordinate refers to an element of $\underline{X}_{n}^*$ or $\underline{Y}_{m}^*$.
			}
	\end{definition}
	 
	As a simple example, suppose that $\underline{X}_{n}$ and $\underline{Y}_{m}$ contain respectively $2$ and $1$ unique values. Then $k = 2$, $c = 1$ and the support of ${p}$ is
	\[
	\Part = \biggl\{ \{(1,1), (2,0) \}, \{(1,0), (2,1) \}, \{(1,0), (2,0), (0,1) \}  \biggr\}.
	\]
	Once the latent structure ${p}$ is 	identified, its elements can be conveniently partitioned into the set $\Delta_{{p}} = \left\{ (i,j) \in {p} \mid i \neq 0 \text{ and } j \neq 0 \right\}$, which includes all the hyper-ties, and the sets
	$\Delta^1_{{p}} = \left\{ (i,j) \in {p} \mid j = 0 \right\}$ and $\Delta^2_{{p}} = \left\{ (i,j) \in {p} \mid i = 0 \right\}$.
	If $X_i^*$ and $Y_j^*$ form a hyper-tie, it means that $(X_i^*, Y_j^*)$ is an actual atom in representation \eqref{series_normCRV}. Instead, if $X_i^*$ does not form a hyper-tie, we have a partial knowledge of the original pair: the unknown second coordinate can be sampled from $P_{X_i^*}(\cdot)$, that is the conditional distribution given $X_i^*$, induced by the joint measure $G_0$, which will henceforth be assumed to be non--atomic. A similar argument applies if $Y_j^*$ does not form a hyper-tie.
	
	In order to simplify notation, we set
	$g_{i,j} = g_0(X_i^*, Y_j^*)$, $g_{i,0} = p_0(X_i^*)$, and $g_{0,j} = p_0(Y_j^*)$,
	where $g_0$ and $p_0$ are the density functions of $G_0$ and $P_0$ respectively, that we assume 
	exist with respect to suitable dominating measures. Finally, we consider the following integrals
	\[
	\tau_{n,m}(\underline{u}) = \int_{\R_+^2}e^{-u_1s_1-u_2s_2}s_1^ns_2^m \, \rho (\d s_1, \d s_2), \quad \underline{u} = (u_1, u_2),
	\]
	where often $n$ and $m$ will be equal to $n_i$ and $m_j$, with $1 \leq i \leq k$ and $1 \leq j \leq c$. For consistency, we set $n_0 = m_0 = 0$. 
	
	The key result of the section relies on a latent structure that is identified by random variables whose conditional distributions, given  $(X_{i})_{i=1}^{n}$ and $(Y_{j})_{j=1}^{m}$, are available. Indeed, these random variables are given by ${p}$, whose probability mass function is proportional to
		\[
		\left(\prod_{(i,j) \in p}g_{i,j}\right)\int_{\R_+^2}u_1^{n-1}u_2^{m-1}\prod_{(i,j) \in {p}}\tau_{n_i, m_j}(\underline{u})\,\mathrm{e}^{-\psi_b(\underline{u})}\, \d \underline{u},
		\]
	the vector $(U_1, U_2)$, whose density on $\R_+^2$ is proportional to
	$u_1^{n-1}u_2^{m-1}\prod_{(i,j) \in {p}}\tau_{n_i, m_j}(\underline{u})e^{-\psi_b({u})}$,
	the variables $\{Z_i^x\}_i$, whose distribution is $P_{X_i^*}(\cdot)$, for any $i = 1, \dots k$, and
	$\{Z_j^y\}_j$, whose distribution is $P_{Y_j^*}(\cdot)$, for any $j = 1, \dots, c$. 	We are now ready to state the key posterior characterization.
	
	\bigskip
	 
	\begin{theorem}\label{main}
		Let$(X_{i})_{i=1}^{n}$ and $(Y_{j})_{j=1}^{m}$ be from model \eqref{partial_model}, with $Q$ being the law of a n-FuRBI. Then, the distribution of $(\mu_1, \mu_2)$ conditional on $(X_{i})_{i=1}^{n}$, $(Y_{j})_{j=1}^{m}$ and the set of latent variables $( p, U_1, U_2,$ $ \{Z_i^x\}_i, \{Z_j^y\}_j)$ is
		\[
		(\hat{\mu}_1, \hat{\mu}_2) + \sum_{(i,j) \in \Delta_{{p}}}{J}_{i,j}\delta_{\left( X_i^*, Y_j^* \right)}+ \sum_{(i,j) \in \Delta^1_{{p}}}{J}_{i,0}\delta_{\left( X_i^*, Z_i^x \right)}+\sum_{(i,j) \in \Delta^2_{{p}}}{J}_{0,j}\delta_{\left( Z_j^y, Y_j^* \right)},
		\]
		where $(\hat{\mu}_1, \hat{\mu}_2)$ is a completely random vector with intensity $e^{-U_1s_1-U_2s_2}\rho(\d s_1, \d s_2)G_0(\d x)$ and ${J}_{i,j} = (J^1_{i,j}, J^2_{i,j})$, with $i=0,\ldots,k$ e $j=0,\ldots,c$, are jumps with density proportional to
		\[
		s_1^{n_i}s_2^{m_j}e^{-U_1s_1-U_2s_2}\rho(\d s_1, \d s_2).
		\]
		Moreover $(\hat{\mu}_1, \hat{\mu}_2)$ and ${J}_{i,j}$ are independent.
	\end{theorem}
	 
	Conditional on the latent variables, the structure is 
	quite intuitive: the posterior is the law of a completely random vector with modified intensity and fixed locations, given by the pairs formed by the hyper-ties. 
	This is somehow reminiscent of the posterior structures of exchangeable models \citep{james2009posterior, lijoi2010models}, with the key novelty played by the new notion of hyper-ties, in addition to the identification of a suitable latent structure.   
	
	The distribution of the latent variables 
	admits a nice interpretation. For instance, the mass function of the latent structure ${p}$ is the product of two terms: the probability of observing the number of hyper-ties identified by ${p}$ times the likelihood that exactly those pairs are formed, through the density function $g_0$.
	Thus, thanks to the homogeneity of the original completely random vector, we observe a separate effect for jumps and locations on this hidden clustering structure. The next corollary shows how the posterior distribution of the normalized measures can be deduced from Theorem~\ref{main}. The statement focuses on $p_1$, though an  analogous representation holds also for $p_2$.
	
	\bigskip
	 
	\begin{corollary}\label{post_p}
		Under the same assumptions of {\rm Theorem~\ref{main}}, conditional on $(X_{i})_{i=1}^{n}$, $(Y_{j})_{j=1}^{m}$ and the 
		latent variables $( p, U_1, U_2, \{Z_i^x\}_i, \{Z_j^y\}_j)$, the random probability measure $p_1$ in \eqref{series_normCRV} equals in distribution 
		\[
		\begin{aligned}
			w_1\frac{\hat{\mu}_1}{T_1} &+w_2\frac{\sum_{(i,j) \in \Delta_{{p}}}J^1_{i,j}\delta_{\left( X_i^*, Y_j^* \right)}}{\sum_{(i,j) \in \Delta_{{p}}}J^1_{i,j}}+ w_3\frac{\sum_{(i,j) \in \Delta^1_{{p}}}J^1_{i,0}\delta_{\left( X_i^*, Z_i^x \right)}}{\sum_{(i,j) \in \Delta^1_{{p}}}J^1_{i,0}}+w_4\frac{\sum_{(i,j) \in \Delta^2_{{p}}}J^1_{0,j}\delta_{\left( Z_j^y, Y_j^* \right)}}{\sum_{(i,j) \in \Delta^2_{{p}}}J^1_{0,j}},
		\end{aligned}
		\]
		where $T_1 = \hat{\mu}_1(\X \times \X)$, while
		\[
		w_1 \propto T_1, \quad w_2 \propto \sum_{(i,j) \in \Delta_{{p}}}J^1_{i,j}, \quad w_3 \propto \sum_{(i,j) \in \Delta^1_{{p}}}J^1_{i,0}, \quad w_4 \propto \sum_{(i,j) \in \Delta^2_{{p}}}J^1_{0,j},
		\]
		with the constraint $\sum_{i = 1}^4w_i = 1$. 
	\end{corollary}

	\subsection{Predictive structure}
	Prediction of new observations arises naturally within the Bayesian framework, since it coincides with the estimate of the distribution under a square loss function. Moreover, it has the merit of providing intuition on how the model behaves and learns
	and it can be used to develop marginal algorithms that avoid the direct sampling of $\tilde p_1$ and $\tilde p_2$, which are infinite-dimensional objects. In Proposition \ref{first_pair} we saw how to sample the first pair of observations. The next result tackles the general case.
	 
	\begin{theorem}\label{predictive}
		Consider samples $(X_{i})_{i=1}^{n}$ and $(Y_{j})_{j=1}^{m}$ from model \eqref{partial_model}, with the same setting of Theorem~\ref{main}. Then there exist probability weights $\xi_0$, $\{\xi_i^x\}$ and $\{\xi_j^y\}$ such that
		\[
		\P \bigl( X_{n+1} \in C \mid (X_{i})_{i=1}^{n} , \,(Y_{j})_{j=1}^{m} \bigr) = \xi_0P_0(C) + \sum_{i= 1}^k\xi_i^x \delta_{X_i^{*}}(C)+  \sum_{j = 1}^c\xi_j^yP_{Y_j^*}\left(C\right).
		\]
		Analogously, there exist probability weights $\eta_0$, $\{\eta_i^x\}$ and $\{\eta_j^y\}$ such that for any $C \in \mathcal{X}$ 
		\[
		\P \bigl( Y_{m+1} \in C \mid (X_{i})_{i=1}^{n}, \,(Y_{j})_{j=1}^{m} \bigr) = \eta_0P_0(C) + \sum_{j = 1}^c\eta_j^y \delta_{Y_j^{*}}(C)+  \sum_{i = 1}^k\eta_i^xP_{X_i^*}\left(C\right).
		\]
	\end{theorem}
Explicit formulae for the weights are available in the proof of Theorem \ref{predictive}, in Section $S1$. In specific cases they can be computed in closed form, conditional to the latent variables: see e.g. example $S1$ in Section $S4$ for the Inverse Gaussian case with equal jumps.
	 
	Hence, the marginal predictive distributions have a quite intuitive form: they are linear combinations of the centering distribution $P_0$, a weighted version of the empirical distribution and a last term that depends on the other sample. The crucial differences with respect to prediction rules arising in the exchangeable case  \citep{lijoi2010models,DeBlasi2015a} is the addition of the last term, which clearly shows how posterior inference changes when incorporating heterogeneous information and performing borrowing of information.
	 
	\begin{example}[n-FuRBI with equal atoms]{\rm 
		If the joint distribution $G_0$ is degenerate such that the atoms are completely shared between $\tilde p_1$ and $\tilde p_2$, then $P_{Z}(\cdot) = \delta_{Z}(\cdot)$. Therefore, the last term in Theorem~\ref{predictive} becomes a weighted version of the empirical distribution relative to the other sample.}
	\end{example}
	{Algorithms for posterior inference and prediction are derived in Section S4.}
	
		\section{Numerical Illustrations and Real Data Analyses}\label{s:illustration}
	\subsection{Bayesian mixture models}\label{bayesian_mixture}
	Discrete Bayesian models, as the one specified in \eqref{partial_model}, are usually not employed directly on the data, but as a building block in hierarchical mixture models: in this setting ${X}$ and ${Y}$ are hidden values that describes the clustering structure within the data. Such models have been introduced by 
	\cite{lo1984class} for the Dirichlet processes and gained popularity thanks also to the availability of sampling methods for posterior inference \citep{escobar1995bayesian,ishwaran2001gibbs,neal2000markov}.
	Suppose $\{f(\cdot \mid x) \, : \, x \in \X\}$ is a family of probability density kernels on a 
	space $\mathbb{W}$.  
	Then the model can be formulated in the context of \eqref{partial_model} as
	\begin{equation*}\label{hier_model}
		\begin{aligned}
			W_i \mid &X_i \overset{\text{ind}}{\sim} f(\cdot \mid X_i)\\
			&X_i \mid \tilde p_1 \simiid \tilde p_1
		\end{aligned},
		\quad
		\begin{aligned}
			V_j \mid &Y_j \overset{\text{ind}}{\sim} f(\cdot \mid Y_j)\\
			&Y_j \mid \tilde p_2 \simiid \tilde p_2
		\end{aligned},
		\quad (\tilde p_1, \tilde p_2) \sim \text{n-FuRBI}.
	\end{equation*}
	where $(W_i)_{i=1}^{n}$ and $(V_j)_{j=1}^{m}$ are the observable samples and are assumed to be conditionally independent, given $(X_{i})_{i=1}^{n}$ and $(Y_{j})_{j=1}^{m}$. Integrating out the latent variables $(X_{i})_{i=1}^{n}$ and $(Y_{j})_{j=1}^{m}$, the data are random draws from suitable countable mixtures, i.e.
	\[
	W_i \mid \tilde p_1 \overset{iid}{\sim} \int f(\cdot \mid x) \, \tilde p_1(\d x), \quad V_j\mid \tilde p_2 \overset{iid}{\sim} \int f(\cdot \mid y) \, \tilde p_2(\d y).
	\]
	\begin{example}[Gaussian mixtures]\label{gaussian_example}
	{\rm 	
		We assume $f(\cdot \mid x) := N(\cdot \mid x, \sigma^2)$, with $\sigma^2$ positive known constant, to be the normal density. Thus, the latent parameter is the mean, i.e. $\X = \R$.  In this case
		$\text{cov}(X_i, Y_j) = \text{cov}(W_i, V_j)$,
		so that the joint behavior of the latent means is reflected on the observations: this shows the importance of the correlation structure given by Proposition \ref{corr} also for hierarchical models. 
		Alternatively, the latent parameters could specify both the mean and the variance, with $\X = \R \times \R_+$.
	}
	\end{example}
	The goal is then to draw samples from the posterior distribution given $(W_{i})_{i=1}^{n}$ and $(V_{j})_{j=1}^{m}$: however this requires to integrate out all the possible partitions of the $n+m$ latent variables. As detailed in Section S4., it is possible to devise a Gibbs sampler for drawing from the posterior distribution of $(X_{i})_{i=1}^{n}$ and $(Y_{j})_{j=1}^{m}$.
	
	Once a posterior sample $(X_{i})_{i=1}^{n}$ and $(Y_{j})_{j=1}^{m}$ is generated, relevant quantities of interest can be approximated by exploiting the conditional independence 
	of $(W_{i})_{i=1}^{n}$ and $(V_{j})_{j=1}^{m}$, given the latent variables.

	\subsection{Simulation study for density estimation}\label{simulations_density}
	We consider a simple application with simulated data, in order to understand how inference changes when taking into account heterogeneous sources of information. Assume the following generating mechanism:
$ W_i \simiid N(\cdot \mid 10,1)$, for $i = 1, \dots, 20$, and
		$V_j \simiid N(\cdot \mid -10, 1)$, for $j = 1, \dots, 100$.
	Supposing only the phenomenon associated to the first sample is of interest, hierarchical mixtures are considered to make prediction on the unknown density of $W_i$. The kernel considered is the one specified in Example \ref{gaussian_example}, with known $\sigma^2 = 1$ and latent mean $\mu$. Four different approaches for modelling dependence between $(W_i)_{i\geq1}$ and $(V_i)_{i\geq1}$ are devised:
the exchangeable approach, according to which 
		sequences ${W}$ and ${V}$ are supposed to form one exchangeable sequence, inducing the highest positive correlation between $W_i$ and $V_j$; the independent approach, according to which the sample $(V_i)_{i\geq1}$ is disregarded entirely, that is $(W_i)_{i\geq1}$ and $(V_i)_{i\geq1}$ are treated independently; the hierarchical approach, where we use a hierarchical Dirichlet process (see Example \ref{HDPexample}) that corresponds to a classical borrowing of information; the FuRBI approach, where the underlying random probability measures $\tilde p_1$ and $\tilde p_2$ are n-FuRBI with equal weights and the distribution on the atoms is
			$G_0(\cdot \mid \rho_0) = N_2\left(\cdot \mid \underline{0}, 1, \rho_0\right)$ with  $\rho_0 \sim \mbox{Unif}([-1,1])$,
		where $N_2(\cdot \mid \underline{m}, \sigma_0^2, \rho_0)$ denotes the bivariate normal distribution with mean vector $\underline{m}$, common variance $\sigma_0^2$ and correlation $\rho_0$. It can be proven that under this specification $\Corr(W_i, V_j) = 0$, so that a priori ${W}$ and ${V}$ are marginally uncorrelated. The prior specification is purposely simple, especially regarding the base measure and the concentration parameter, in order to single out the effect of the borrowing between the two groups as much as possible.
		
	For the first two cases and the n-FuRBI, the marginal distribution is given by a Dirichlet process with $\theta = 1$ and $P_0(\cdot) = N(\cdot \mid 0, 1)$; instead for the hierarchical process the concentration parameters are fixed in order to match the expected number of different clusters with the other methods, for a fair comparison.
	As highlighted in Example \ref{example_equal_jumps}, n-FuRBI with equal jumps lead to the most general setting in terms of achievable correlation between samples; moreover, choosing the marginal processes to derive from a Gamma process, we can achieve any value in the interval $(-1,1)$, tuning appropriately the concentration parameter $\theta$.
	
	\begin{figure}[h!]
		\centering
		\captionsetup{width=\linewidth}
		\includegraphics[width=.49\textwidth]{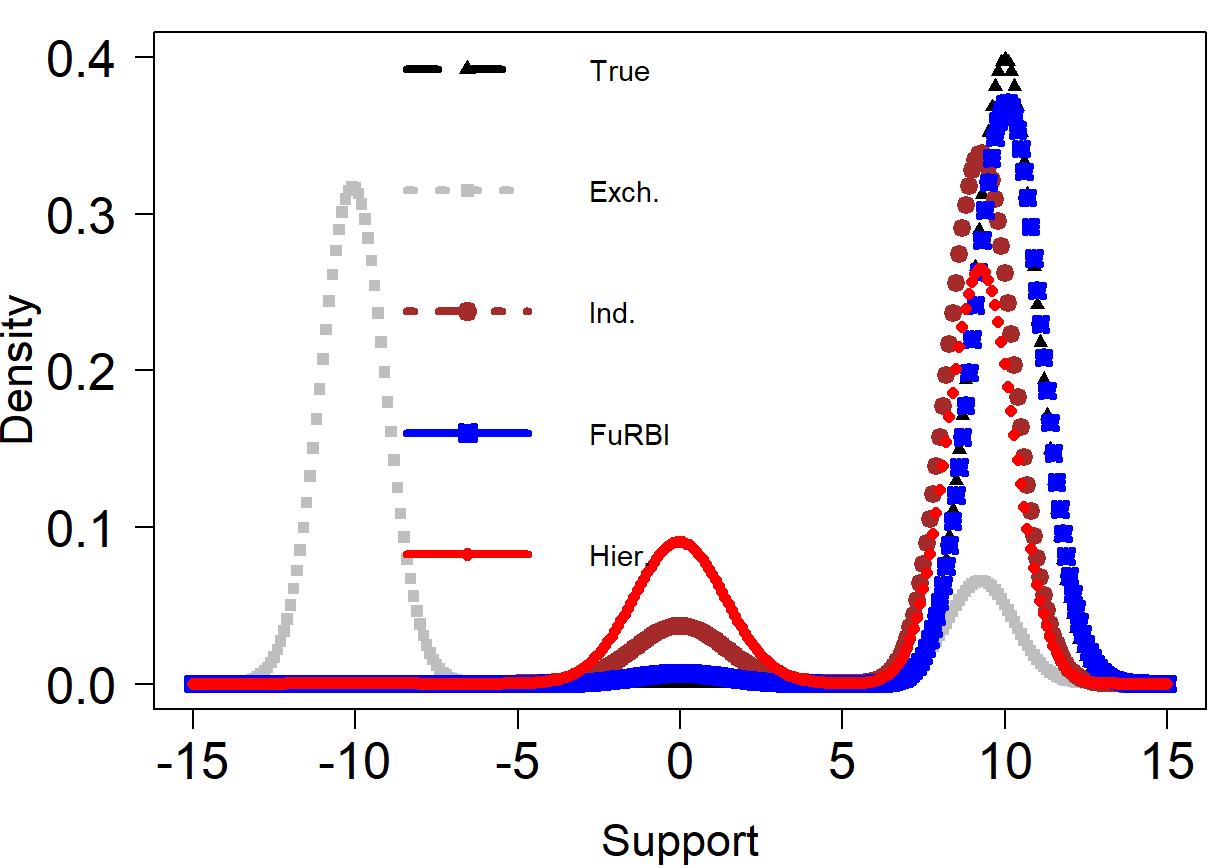}
		\includegraphics[width=.49\textwidth]{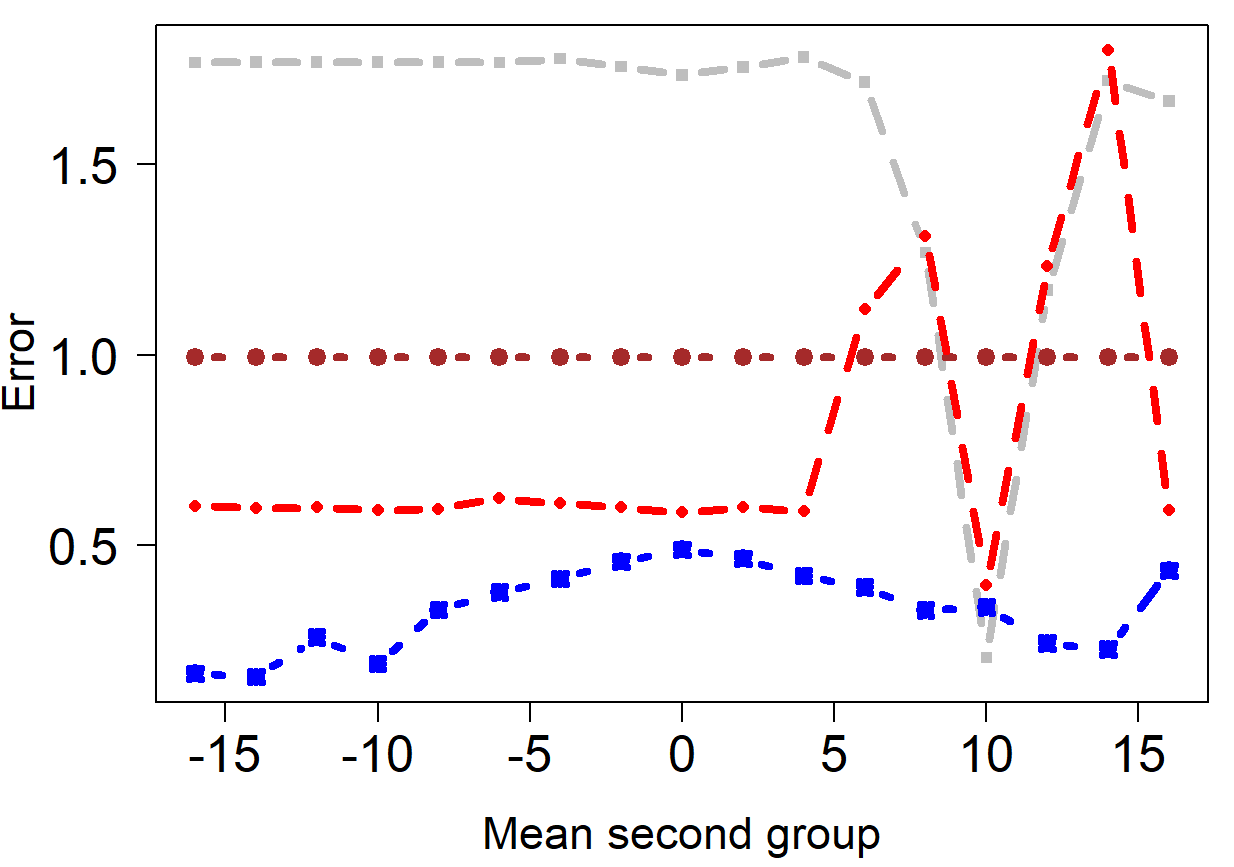}
		\caption{ Left: mean posterior densities for the case with opposite true means. Right: mean integrated error (computed on a grid and as the median over $50$ different samples) for the four estimates, varying the true mean of $V$.}
		\label{fig1}
	\end{figure}
	
	The left panel of Figure \ref{fig1} shows the performances of the four methods, after the application of the blocked Gibbs sampler provided in the supporting material: the mean posterior density (computed pointwise) is depicted. The exchangeable approach behaves very badly, as expected, because the two samples clearly have a different distribution. The independent choice leads to a reasonable estimate, even if it still overestimates the probability mass around the prior mean (because of the small sample size of the first sample). The hierarchical estimate is quite good, but our proposal, instead, fits almost perfectly the target density and seems to exploit the opposite behaviour of the two phenomena: this is clearly highlighted by the posterior distribution of $\rho_0$, whose approximated mean is close to $-0.9$.
	
	One may wonder whether these superior performances follow from the precise specification above, with opposite true means. Therefore, we have repeated the experiment by keeping the same generating mechanism for $W$, but with the true mean of ${V}$ ranging in the set $\{-16, -14, \dots, 14, 16\}$: the mean integrated absolute error (computed on a grid and as the median over $50$ different samples) is depicted in the right panel of Figure \ref{fig1}. It is apparent that the FuRBI approach almost always yields the smallest error, regardless of the true value. Its performance is close to the exchangeable case only when the two true means are equal, that is when exchangeability actually holds; analogously, the n-FuRBI priors yield the highest error when the mean of $V$ corresponds to the prior mean, i.e., when the other group provides less additional information. The hierarchical process captures the right dependence when the two means coincide, but can be misled when they are close; finally, when the second sample is very far from the first one it performs better than the independent model, probably thanks to the different inner clustering structure. The results are also summarized in Table \ref{table 1}. Thus, n-FuRBI seem to be always capable of combining heterogeneous information in the right way; in particular, at least in this example, they recognize the most useful type of borrowing of information. In Section S5.1 similar experiments are conducted, using different data generating distributions: they show that the conclusions hold even when the data display significantly different features, as multimodality or heavy tails.
		\begin{figure}\label{fig_cor}
			\vspace{-1\baselineskip}
\floatbox[{\capbeside\thisfloatsetup{capbesideposition={right,center},capbesidewidth=6cm}}]{figure}[\FBwidth]
{\caption{Posterior median of the correlation (obtained through $100$ simulation studies) between the three unknown means. Black with triangular shapes: correlation between the first and second components. Red with square shapes: correlation between the first and third components. Green with circular shapes: correlation between the second and third components.}\label{fig_corr}}
{\includegraphics[width=7cm]{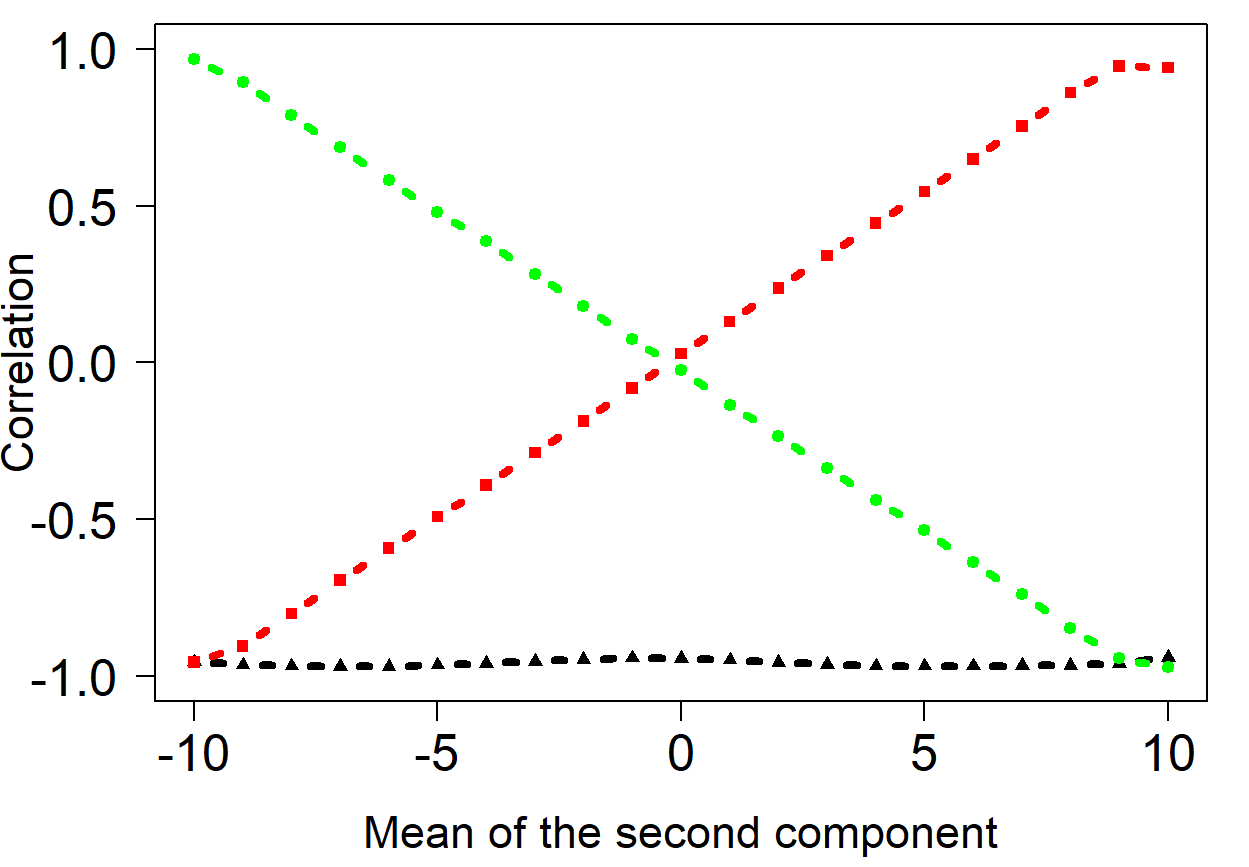}}
	\end{figure}

	\begin{table}
		\caption{\label{table 1}Mean integrated absolute error associated to the four methods for some values of the mean of ${V}$. The values in bold are the smallest ones for each row.}
		\centering
		\fbox{%
			\begin{tabular}{ccccc}\hline
				\text{Mean of } {V}  &\text{Exch.} &\text{Ind.} & \text{FuRBI} &\text{Hier.}  \\\hline
				-16   & 1.769 & 0.995 & \textbf{0.163} & 0.604 \\\hline
				-10   & 1.769 &0.995 & \textbf{0.189} & 0.592 \\\hline
				0   & 1.737 & 0.995 & \textbf{0.489} & 0.587 \\\hline
				10   & \textbf{0.205} & 0.995 & 0.338 & 0.397 \\\hline
				16   & 1.666 & 0.995 & \textbf{0.435} & 0.592 \\\hline
		\end{tabular}}
	\end{table}
Finally, we consider a similar application with three groups, in order to see whether n-FuRBI are able to discern more complex types of dependence. We assume to observe
$W_{1,i} \simiid N(\cdot \mid 10,1)$, $W_{2,i} \simiid N(\cdot \mid -10, 1)$, and $W_{3,i} \simiid N(\cdot \mid x, 1)$, 	where $i=1,\ldots,20$ and $x \in \{-10, -9, \dots, 10 \}$. Then, for each value of $x$ we apply the same n-FuRBI with the same weights described above, but where the atoms are distributed according to
	\begin{equation*}\label{proposed_g0_2}
		\begin{aligned}
		G_0\left(\cdot \right)&= N_3\left(\cdot \biggl \lvert\underline{0}, 1, \begin{bmatrix} 1& \rho_{12}& \rho_{13}\\ \rho_{12}&1&\rho_{23} \\ \rho_{13}& \rho_{23}& 1\end{bmatrix}\right),\\
		 \end{aligned}
	\end{equation*}
	where $N_3(\cdot \mid \mu_0, \sigma^2, \Psi)$ denotes a multivariate normal distribution with mean $\mu_0$, all the variances equal to $\sigma^2$ and correlation matrix $\Psi$ and $ \rho_{12 }, \rho_{13}, \rho_{23} \simiid \mbox{Unif}([-1,1])$. The posterior medians of $\rho_{12}, \rho_{13}$ and $\rho_{23}$ are depicted in Figure \ref{fig_corr}, for any value of $x$. 
	The results are in line with our intuition: the correlation between the first and second component is always close to $-1$ (indeed they have opposite behaviour relative to the prior), while $\rho_{13}$ and $\rho_{23}$ vary linearly with $x$, being positive when the means have the same sign.

 \subsection{Predicting stocks and bonds returns}\label{s:illustration1} 

	\begin{figure}
		\includegraphics[width=0.65\linewidth]{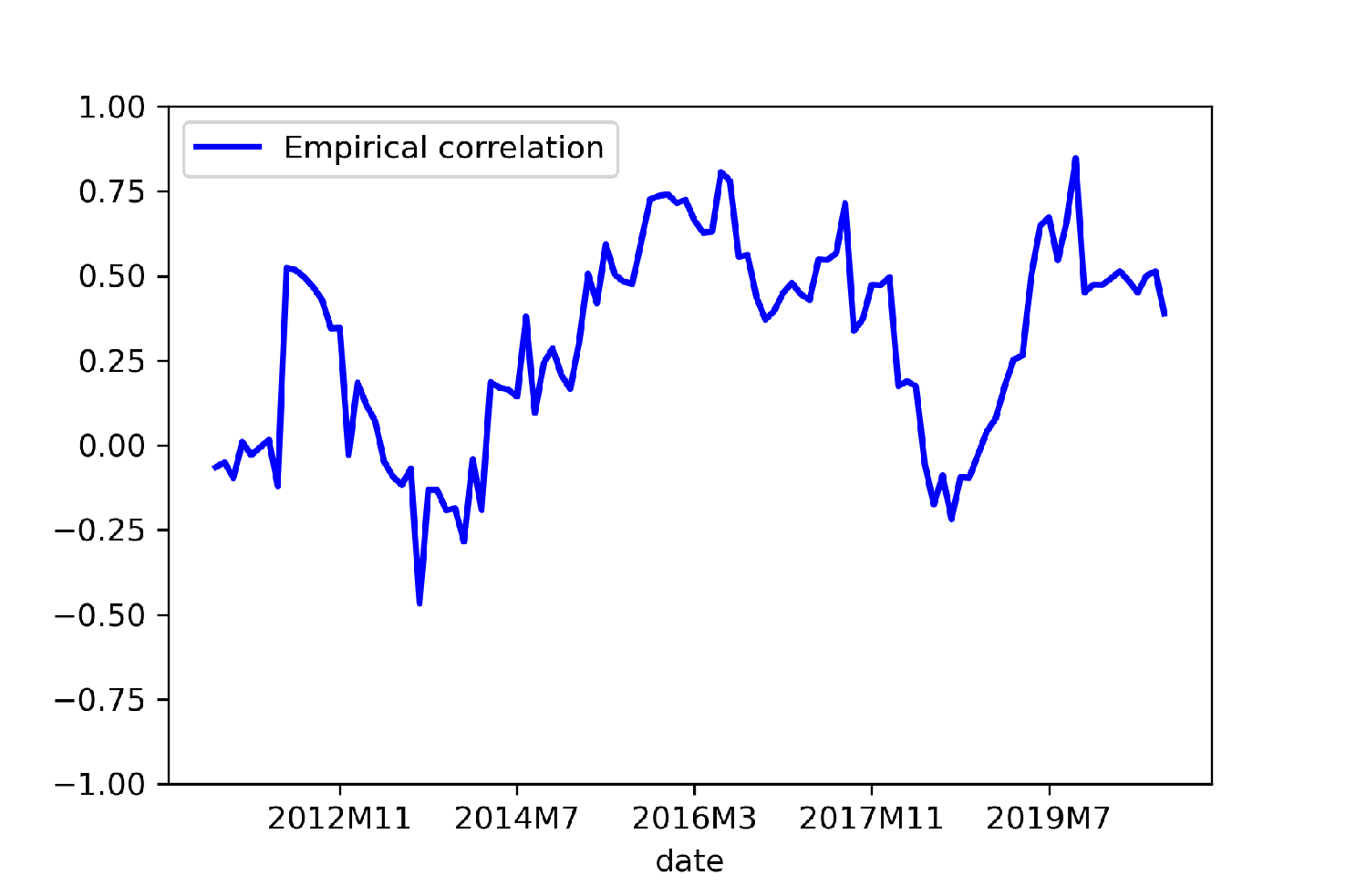}
		\caption{\label{fig:corr}{Empirical correlation between average stock return and average commodity return computed on a moving window of 12 months using data from March 2011 to January 2021.}}
	\end{figure}
Findings from the previous section and Section S5.1 suggest that n-FuRBI may be used to enhance density estimates and prediction in multi-sample data. Here, the performance 
	is showcased on a real dataset of stock and bond returns. We collected monthly returns of January 2021 for a sample of 49 stocks portfolios from the Kenneth R. French's Data Library (data available at \url{http://mba.tuck.dartmouth.edu/pages/faculty/ken.french/data_library.html}) and for a sample of 55 commodities from the Primary Commodity Prices Database of the International Monetary Fund (data available at \url{https://www.imf.org/en/Research/commodity-prices}).

{We employ a Bayesian mixture model and assume that stock and bonds returns, denoted by $W_i$ and $V_j$, respectively, are sampled from mixtures of normals where the mixing distributions act on mean and variance of the kernel, i.e.,
\begin{equation*}
	\begin{aligned}
		W_i\mid \tilde p_1 \overset{iid}{\sim} \int N(\cdot \mid x,\sigma^2_w) \, \tilde p_1(\d x, \d \sigma^2_w) \quad&\quad\quad
		V_j\mid \tilde p_2 \overset{iid}{\sim} \int N(\cdot \mid y,\sigma^2_v) \, \tilde p_2(\d y,\d \sigma^2_v).
	\end{aligned}
\end{equation*}}

Stocks and commodities exhibit correlation that largely varies over time ranging from positive to negative values \citep[see, for instance,][and Figure~\ref{fig:corr}]{bhardwaj2013business}. 
As a consequence, commodities returns contain useful information to make inference over the distribution of stocks portfolios, and viceversa. Thus, borrowing of information represents a natural strategy to improve inference. However, returns may differ even largely in value between the two sets of financial instruments, especially in periods of negative correlation. For instance, in our dataset, 53\% of the observed stocks returns are negative, while only 16\% of the bonds returns have negative sign. As such, classical nonparametric borrowing, consisting in sharing of mixture components, is not appropriate and, as shown in the following, possibly harmful. We instead make use of n-FuRBI models as prior distribution, i.e.,
\begin{equation*}
	\begin{aligned}
	(\tilde p_1,\tilde p_2) \mid \theta, z, G_0 &\sim \text{n-FuRBI}(\theta, \rho , G_0)\\
	\theta &\sim \text{Gamma}(\alpha,\beta)
	\end{aligned}
\end{equation*}
The base measure $G_0$ is chosen so that marginal distributions are given by normalized CRMs with conjugate Normal-InverseGamma base measure, i.e. 
\begin{equation*}
	\begin{aligned}
		G_0(\d_x, \d_y, \d \sigma^2_w, \d \sigma^2_v \mid \rho_0) = &\,N_2(\d x, \d y \mid {m}, \Sigma(\lambda_1, \lambda_2, \sigma^2_w, \sigma^2_v \rho_0))\\
		&\times \text{InvGamma}(\d \sigma^2_w \mid \alpha_1, \beta_1) \times \text{InvGamma}(\d \sigma^2_v \mid \alpha_2, \beta_2)
	\end{aligned}
\end{equation*}
with 
\begin{equation*}
	m = (m_1, m_2)' \qquad\text{and}\qquad
	\Sigma = \begin{bmatrix}
		\frac{\sigma^2_w}{\lambda_1} & 
		\rho_0\,\frac{\sigma_w}{\lambda_1^{1/2}} \frac{\sigma_v}{\lambda_2^{1/2}} \\
		\rho_0\,\frac{\sigma_w}{\lambda_1^{1/2}} \frac{\sigma_v}{\lambda_2^{1/2}}  & \frac{\sigma^2_v}{\lambda_2}
	\end{bmatrix}
\end{equation*}
and we use the following joint underlying L\'evy intensity
$v(\d s_1, \d s_2, \d x_1, \d x_2) =\{z\,[\rho(\d s_1)\delta_0(\d s_2) + \rho(\d s_2)\delta_0(\d s_1) ] +(1-z)\,\rho(\d s_1)\delta_{s_1}(\d s_2)\}\, \theta\, G_0(\d x_1,\d x_2)$,
with $z \sim \text{Unif}([0,\,1])$.
We term the resulting n-FuRBI \emph{additive n-FuRBI}, since the series representation of the corresponding FuRBI CRMs is
\begin{equation*}
	\tilde \mu_1(\cdot) \overset{a.s.}{=} \sum_{k\geq1}W_k\delta_{\theta_{0,k}}+\sum_{k\geq1}J_k\delta_{\theta_{1,k}}
	\qquad\qquad
	\tilde \mu_2(\cdot) \overset{a.s.}{=} \sum_{k\geq1}W_k\delta_{\phi_{0,k}}+\sum_{k\geq1}V_k\delta_{\phi_{2,k}},
\end{equation*}
where $(\theta_{0,k},\phi_{0,k})\overset{i.i.d}{\sim}G_0 $, $\theta_{1,k} \overset{i.i.d}{\sim}P_0 $ and $\phi_{2,k} \overset{i.i.d}{\sim}P_0$. When $G_0$ is degenerate on the main diagonal (i.e. $\rho_0=1$), one retrieves  GM-dependent completely random measures \citep{lijoi2014bayesian,lijoi2014dependent,lijoi2014class}. 
In order to obtain two Dirichlet processes marginally we set $\rho(s)=s^{-1}e^{-s}$, so that 
$\beta = 1/(1+\theta)$ and 
$\gamma = (1-z)\,{}_3 F_2(\theta-\theta\,z+2,\,1,\,1;\,\theta+2,\,\theta+2;\,1){\theta}/{(1+\theta)^2}$,
where ${}_3 F_2$ is the generalized hypergeometric function.

\begin{figure}[!t]
	\centering
	\begin{subfigure}[b]{0.32\textwidth}
		\centering
		\includegraphics[width=\textwidth]{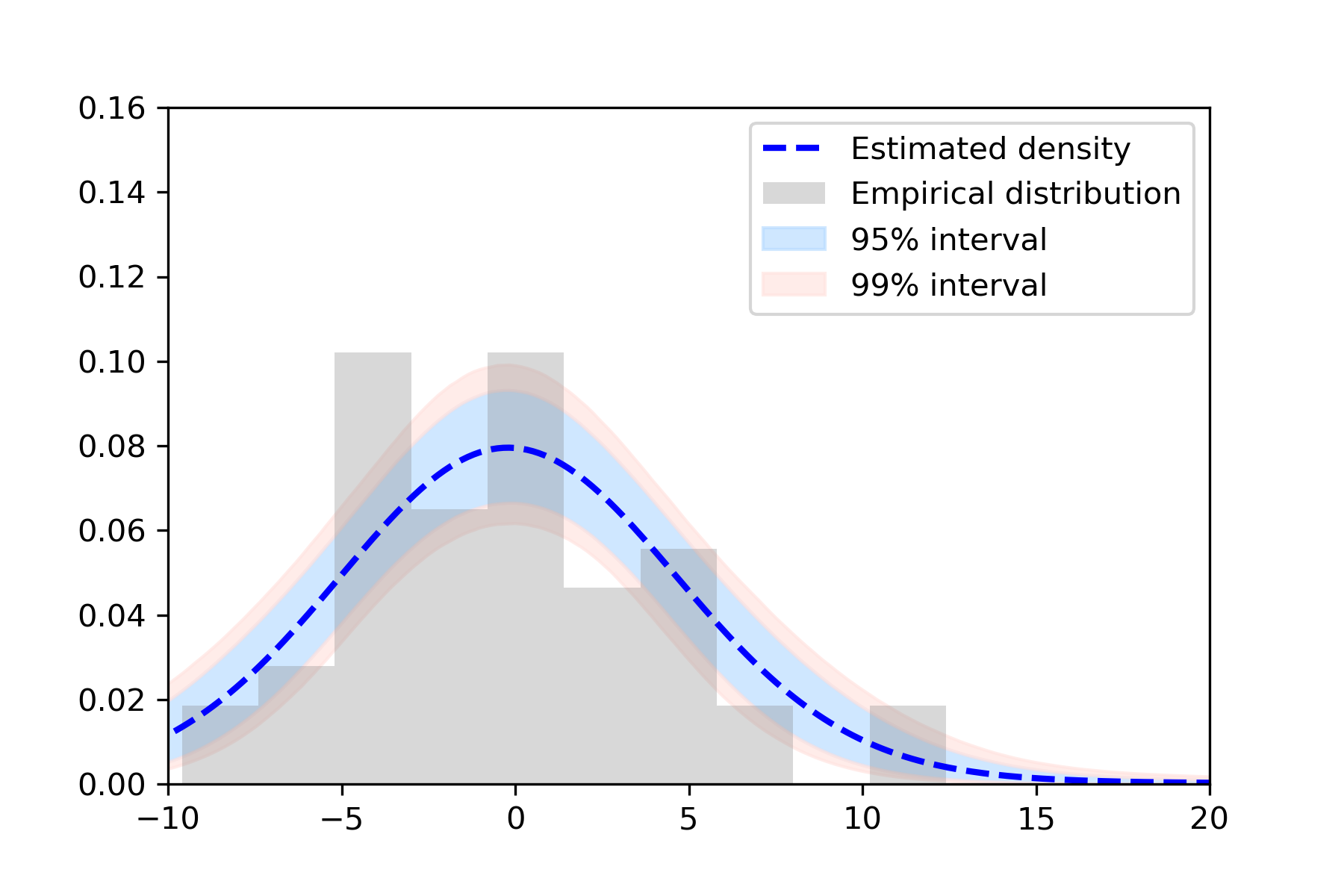}
		\caption{FuRBI with $\rho_0 \in [-1,1]$}
	\end{subfigure}
	\hfill
	\begin{subfigure}[b]{0.32\textwidth}
		\centering
		\includegraphics[width=\textwidth]{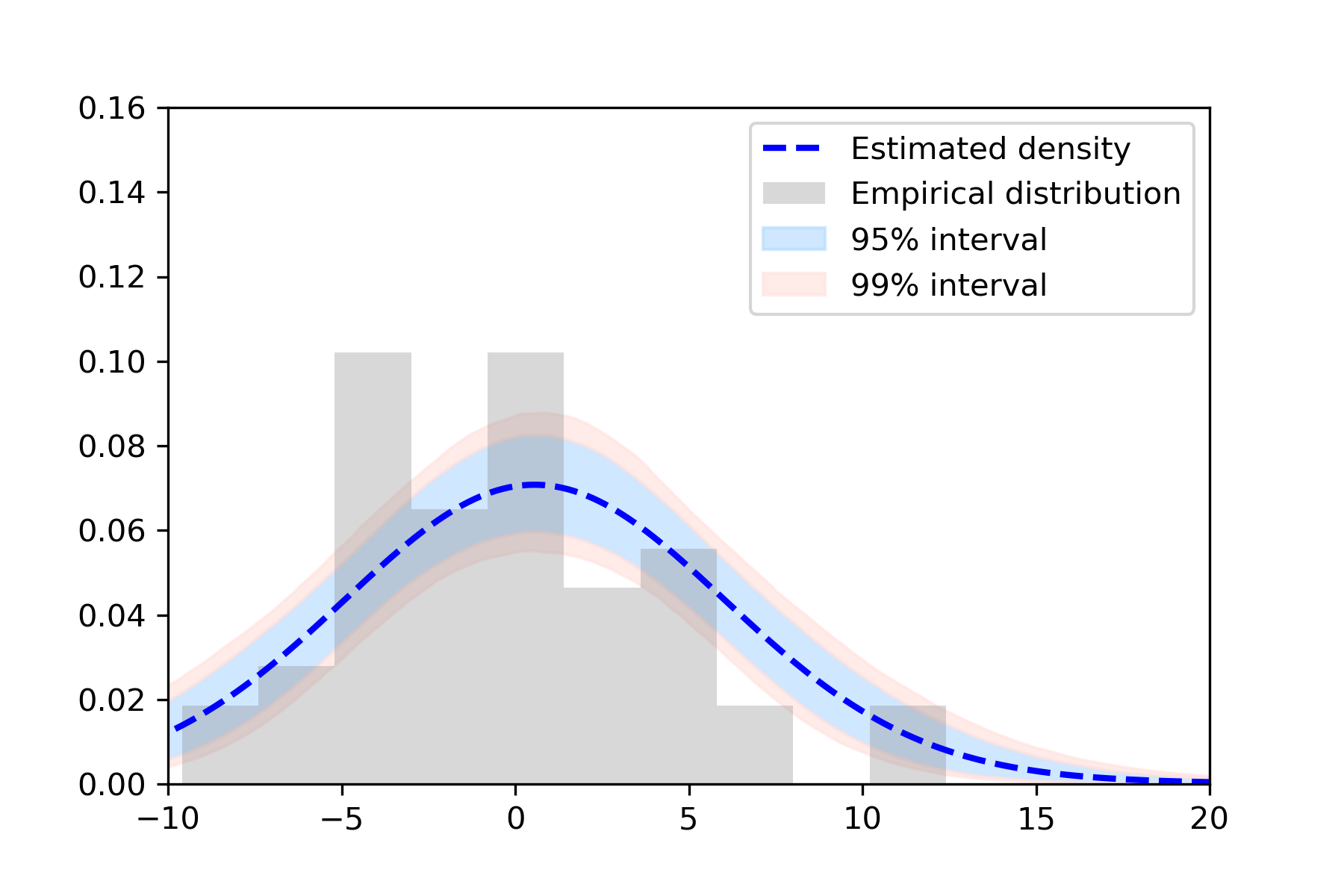}
		\caption{FuRBI with $\rho_0 = -0.95$}
	\end{subfigure}
	\hfill
	\begin{subfigure}[b]{0.32\textwidth}
		\centering
		\includegraphics[width=\textwidth]{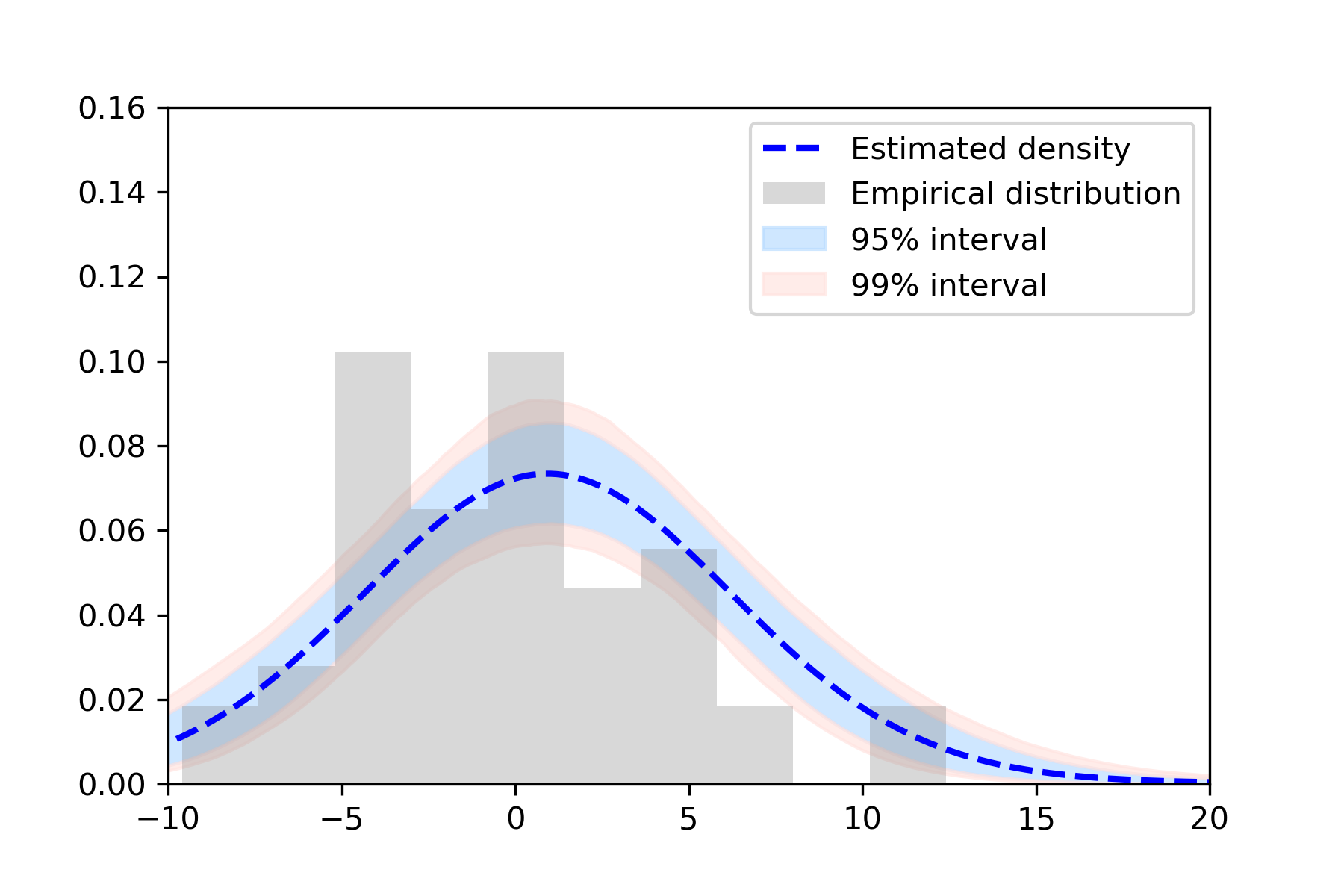}
		\caption{FuRBI with $\rho_0 = 0.95$}
	\end{subfigure}
	\begin{subfigure}[b]{0.32\textwidth}
		\centering
		\includegraphics[width=\textwidth]{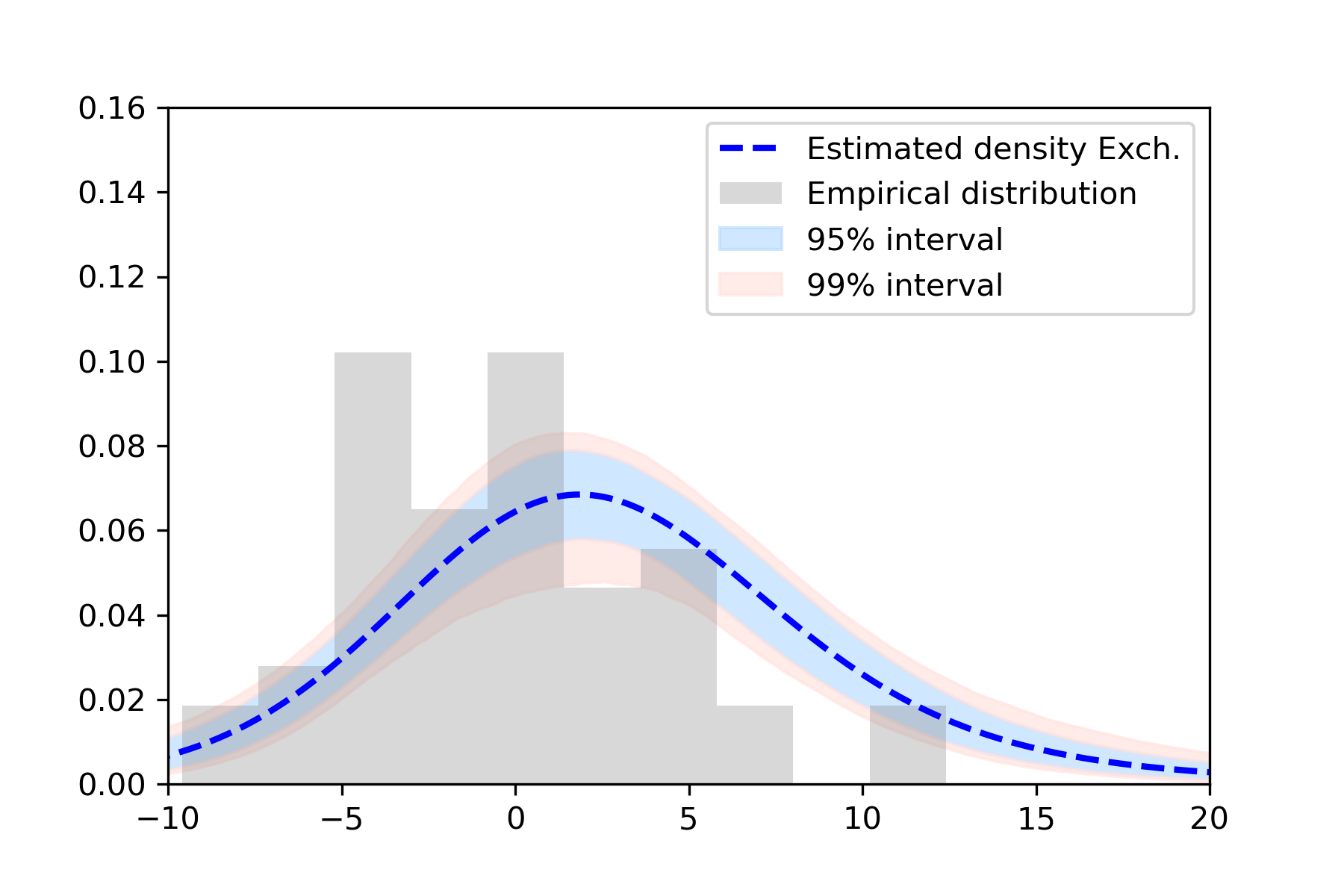}
		\caption{Exchangeable model}
	\end{subfigure}
	\hfill
	\begin{subfigure}[b]{0.32\textwidth}
		\centering
		\includegraphics[width=\textwidth]{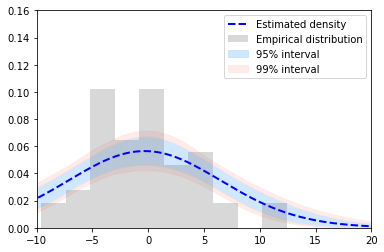}
		\caption{GM-dependent model}
	\end{subfigure}
	\hfill
	\begin{subfigure}[b]{0.32\textwidth}
		\centering
		\includegraphics[width=\textwidth]{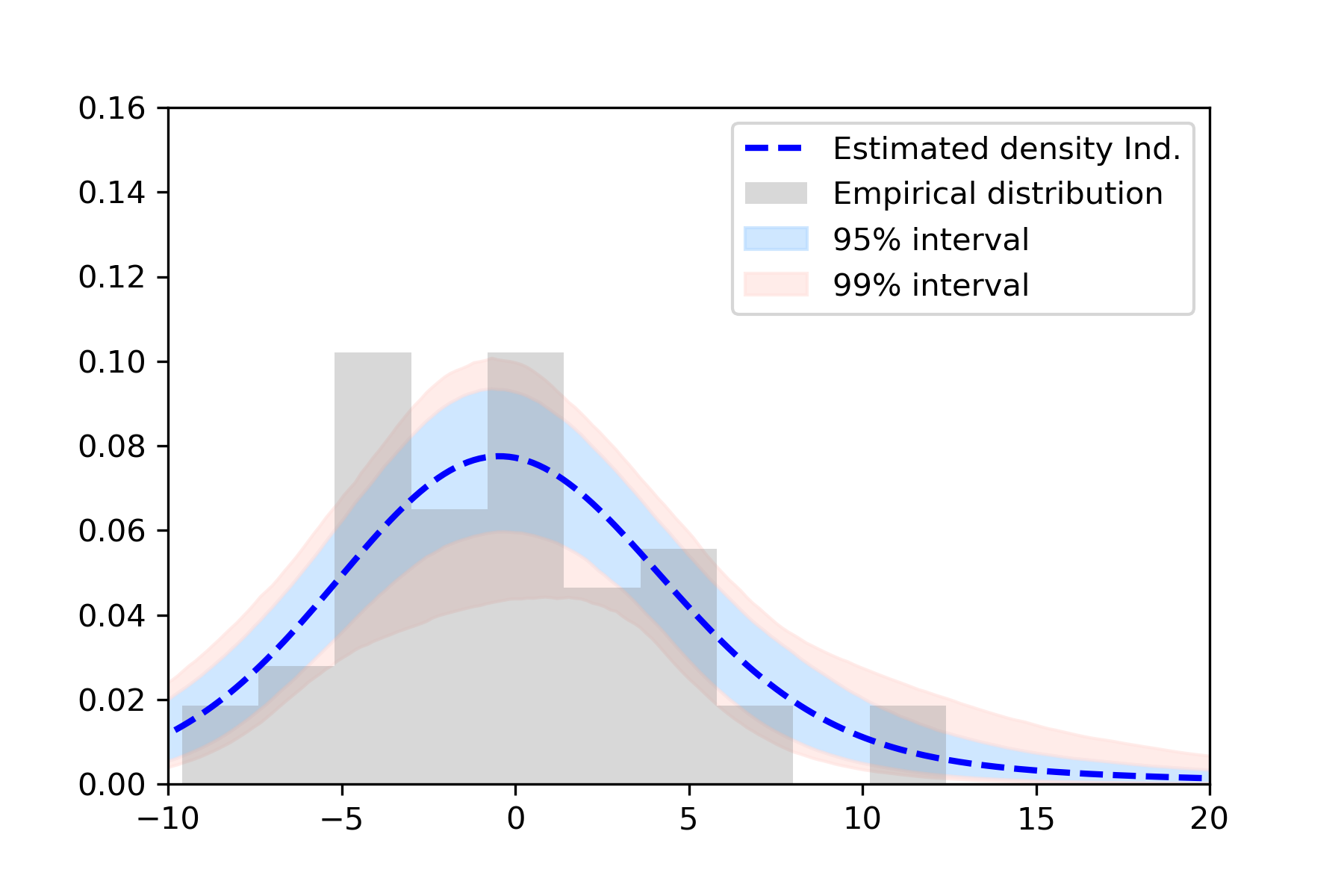}
		\caption{Independent model}
	\end{subfigure}
	\caption{\label{fig:densest} {Posterior density estimates for stocks returns.}}
\end{figure}
As for the hyperparameters of the model, we set the a priori expectations $m_1$ and $m_2$ in the two groups equal to the empirical averages of the two groups in December 2020, i.e., the month preceding the data collection, leading to $m_1 = 5.8591$ and $m_2 = 3.9731$. In the following, we say that a financial instrument is \emph{outperforming} if its observed return is higher than its a priori expected value. In order to assign $\rho_0$, we use the results of Propositions~\ref{corr} and \ref{first_pair}. The elicited $\rho_0$ should reflect our prior opinion about the correlation, which means that it should induce a learning mechanism agreeing with the following principle: under positive/negative correlation, conditioning on the event of outperforming commodities, the prior probability of outperforming/underperforming stocks should increase. Prior opinion about the correlation can be formulated working with financial experts and, thanks to n-FuRBI, incorporated through an informative prior on the parameter $\rho_0$.  Here, we consider three scenarios: in the first and second, we derive inferential results under a prior opinion of negative and positive correlation, respectively, while in the third scenario we assume that no information on the correlation is available. The three scenarios are obtained with, respectively, $\rho_0=0.95$, $\rho_0=-0.95$, and using a uniform prior on $\rho_0$.  After standardizing the data, we set the remaining hyperparameters in a weakly informative way, i.e. $\lambda_1 = \lambda_2 = 1$, $\alpha_1=\alpha_2=2$, and $\beta_1=\beta_2 = 4$. Sensitivity analysis, carried out in Section S6.2, shows that results are robust with respect to different choices for $\lambda_j$, $\alpha_j$ and $\beta_j$ for $j=1,2$. We perform $50,000$ iterations of the marginal algorithm (Section S4.1) and discard the first $10,000$ as burn--in. Section S8.~contains results about convergence diagnostic, mixing performance, and computational times of the algorithm.

Finally, we compare our approach with three alternative models: the independent model and the exchangeable model, described in the previous section, and the GM-dependent model from \cite{lijoi2014dependent}, which performs classical borrowing based on ties and shares the same additive structure of additive n-FuRBI.

\begin{table}
	\begin{tabular}{lcc}\hline
		& ALCPO & MLCPO\\\hline\hline
		FuRBI $\rho_0 \in [-1,1]$ & \textbf{-1.2347} & \textbf{-0.9627}\\
		FuRBI $\rho_0 = -0.95$ & -1.2925 & -1.0115\\
		FuRBI $\rho_0 = 0.95$  & -1.2896 & -1.0149 \\
		Exch & -1.5024   & -1.1521 \\
		GM-dep & -1.4864 & -1.1557 \\
		Ind &  -1.3495	& -1.1017\\
	\end{tabular}
	\caption{\label{table 2}{ALCPO and MLCPO under the three models. Best performance is highlighted in bold.}}
\end{table}

Figure~\ref{fig:densest} displays the posterior density estimates for stocks returns. The analogous figure for bonds returns can be found in Section S6.1. Models employing additive n-FURBI produce density estimates that 
better resemble the empirical distribution. The best performance is attained by placing a (non-informative) prior over the correlation $\rho_0$, which leads to a posterior skewed towards negative values but still quite dispersed (see Figure S7) reflecting the direction and intensity of the borrowing of information. 
The FuRBI models with fixed $\rho_0$ perform worse compared to full-borrowing; nonetheless, thanks to their flexibility, they still produce better results than other competitors. The GM-dependent and the exchangeable models yield the worst density estimates in terms of resemblance of the histogram, as expected. Indeed, the type of borrowing they perform differ from the one allowed by FuRBIs (even when $\rho_0=0.95$), as it is based on ties, which are not appropriate for the specific problem at hand. Lastly, we note that the independent model appears to provide a reasonable density estimation, but presents significantly higher uncertainty. 
 
While Figure~\ref{fig:densest} provides insight on the model performance, 
	an important \textit{caveat} is in order: a too close resemblance of the empirical distribution may indicate overfitting.
 
To evaluate the predictive performance, we resort to the conditional predictive ordinates (CPOs) statistics \citep[see, e.g.][]{gelfand1992model,barrios2013modeling}. Essentially, for each value $i$, we train the model without the $i$-th observation and compute the predictive density at the observed point. For the first sample it reads
$\text{CPO}^w_i=\tilde f (w_i\mid {w}^{-i},{v})$, for $i=1,\ldots,n$
and analogously for the second sample we have 
$\text{CPO}^v_j=\tilde f (v_j\mid {w},{v}^{-j})$, for $j=1,\ldots,m$,
where ${w}$ and ${v}$ denote the vectors of observed returns for, respectively, stocks and commodities.

Table~\ref{table 2} displays the average logarithmic CPO (ALCPO) and the median logarithmic CPO (MLCPO) in the overall sample. 
Higher values correspond to a better performance, and the n-FuRBI exhibits the best performance.

\subsection{Clustering of multivariate data with missing entries}\label{application}

 \begin{figure}[b!]
    \centering
    \begin{subfigure}[b]{0.49\textwidth}
    \includegraphics[width=\textwidth]{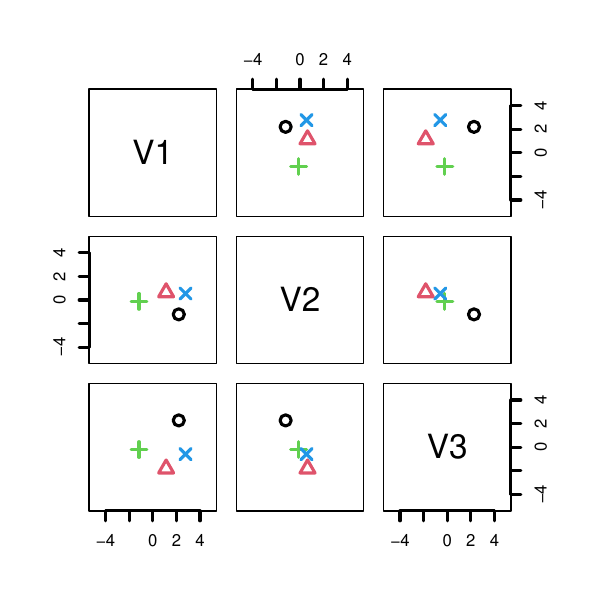}
    \end{subfigure}
    \hfill
    \begin{subfigure}[b]{0.49\textwidth}
    \includegraphics[width=\textwidth]{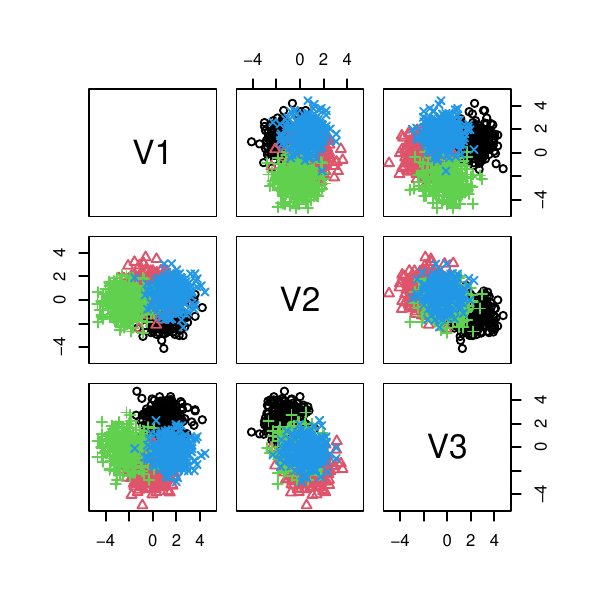}
    \end{subfigure}
    \vspace{-0.5\baselineskip}
    \caption{\label{fig:simdatana}{Simulated data: left panel shows true clusters locations, right panel shows complete simulated data for $n=1000$ before applying the missingness mechanisms.}} 
    \label{fig:my_label}
\end{figure}

We now show how to leverage on our methodology to perform borrowing of information and clustering with multivariate data affected by missing entries. The n-FuRBI priors are very well suited for this problem: indeed, incomplete observations can be interpreted as projections of latent complete observations and, in particular, hyper-ties between incomplete observations can be thought of as actual ties between complete observations.

We consider a $P$-variate ($P>1$) dataset with missing entries and divide the dataset into distinct samples based on the missing entries: denote by $(\underline{W}^{(j_1,\ldots,j_l)}_i, i=1,\ldots,n_{(j_1,\ldots,j_l)})$ the sample where $l$ outcomes with labels $(j_1,\ldots,j_l)$ are missing. The dimension of the vector $\underline{W}^{(j_1,\ldots,j_l)}_i$ is therefore $P_{j_1,\ldots,j_l} = P - l$.  Denote by $ \tilde q_{j_1,\ldots,j_l}$ the corresponding unknown distribution, i.e., 
\[
\underline{W}^{(x)}_i \mid \tilde q_x \overset{iid}{\sim} \tilde q_x \qquad \text{for } i = 1,\ldots,n_x  \text{ and } x\in I,
\]
where $I$ is the index set of all the possible combinations of missing variables identifying different samples, which are at most $2^P-1$. Independent analyses for each sample should clearly be avoided and classical nonparametric borrowing cannot even be specified because the support spaces of different samples differ one from the other.   

To perform clustering, we assume that each $\tilde q_x$ is a mixture of multivariate normal kernels with diagonal covariance matrix and mixing measure $\tilde p_x$ on locations, i.e.
\[
\underline{W}^{(x)}_i \mid \tilde p_x, \underline{\sigma}^2 \overset{iid}{\sim} \int N_{P_x}(\cdot \mid \underline{\mu}_x,\underline{\sigma}^2_x) \, \tilde p_x(\d \underline{\mu}_x) ,
\]
where $\underline{\sigma}^2 = \left(\sigma^2_1, \dots, \sigma^2_P \right)$, $\underline{\sigma}^2_x$ is the restriction of $\underline{\sigma}^2$ to all the elements besides $x$ and $N_K(\cdot \mid \underline{\mu}, \underline{\tau}^2)$ denotes the $K$-variate normal distribution with mean vector $\underline{\mu}$ and diagonal covariance matrix given by $\underline{\tau}^2$. Independence of the kernel (implied by the diagonal covariance matrix) is a common assumption in clustering models for multivariate responses \citep[see, for instance,][]{gao2020clusterings, franzolini2022bayesian}: in this way we are forcing the clustering structure to encode all the dependence across responses. The $\tilde p_x$ are distributed as
\[
(\tilde p_x, x\in I) \sim \text{additive n-FuRBI},
\]
described in Section $6.3$. The atoms of $(\tilde p_x, x\in I)$ are costrained so that an hyper-tie can be interpreted as an actual tie between complete observations: moreover the choice of dependent weights allows to recover group-specific features, if the missingness mechanism is informative. Section S7.1 provides a discussion of this and contains the details about the choice of the hyperparameters.

\begin{table}[!tb]
\begin{tabular}{c|c|c|c|c|c|c|c}
simul&missing& \% of missing& n-FuRBI & n-FuRBI & n-FuRBI & mice +  &mice +\\
number& mechanism &  entries& $z=0.2$&$z = 0.5$&$z=0.8$& k-means& DPM\\
\hline
n.1&MCAR& 16.1\%&\textbf{0.7883}&0.7882&0.7881&0.7408&0.7734\\
n.2&MNAR& 16.7\%&0.7703&0.7704&\textbf{0.7706}&0.6323&0.7617\\
n.3&MCAR& 35.9\%&\textbf{0.7292}&0.7285&0.7283&0.6786&0.7165\\
n.4&MNAR& 34\%&0.7304&0.7301&\textbf{0.7432}&0.6391&0.7328
\end{tabular}
\caption{\label{tab:ri}{Rand indexes for $5$ competing methods: 3 n-FuRBI models with varying parameter $z$, mice+k-means and mice+DPM. For n-FuRBI and mice+DPM the posterior expected value is computed averaging over the Rand indexes of all clustering configurations visited by the MCMC chain after burn-in.}}
\end{table}

\begin{table}[!tb]
\begin{tabular}{c|c|c|c|c|c|c|c}
simul&missing& \% of missing& n-FuRBI & n-FuRBI & n-FuRBI &mice + & mice +\\
number&  mechanism&   entries & $z=0.2$&$z = 0.5$&$z=0.8$& k-means& DPM\\
\hline
n.1&MCAR& 16.1\%&4.24&\textbf{4.19}&4.22&3&5.48\\
n.2&MNAR& 16.7\%&\textbf{4.59}&3.29&3.37&2&5.36\\
n.3&MCAR& 35.9\%&4.38&\textbf{4.18}&4.20&3&7.01\\
n.4&MNAR& 34.0\%&4.28&\textbf{4.17}&4.59&2&5.85\\
\end{tabular}
\caption{\label{tab:K} {Estimated number of clusters for $5$ competing methods. The posterior mean is used for n-FuRBI and mice+DPM, while the number of clusters is selected by maximizing the average silhouette for mice+k-means. The true number of clusters is equal to $4$.}}
\end{table}

First, we conduct a simulation study where data for $n=1,000$ items, $P=3$ responses, and $K=4$ clusters are simulated from a mixture of Gaussian distributions. Figure~\ref{fig:simdatana} shows the locations of the true clusters and the complete simulated data before deleting entries. Then, different missingness	mechanisms are applied to determine the entries to be treated as missing. Missing completely at random (MCAR) scenarios are obtained by sampling missing entries uniformly, while, in missing non at random (MNAR) scenarios the probability of being missing depends on the true cluster allocation. Different combinations of missing variables define different samples: the number of samples ranges from 3 to 6 among simulation scenarios. The detailed distributions of missing values are provided in Section S7.2. Different values of the hyperparameter $z$ of the L\'evy intensity are considered. Our results are compared with those obtained with two alternative approaches, called ``mice + k-means" and ``mice + DPM", which follow a two-steps procedure: first one imputes missing data by chained equations as implemented in the \texttt{R} package \texttt{mice} \citep{mice}, then, the clustering structure is estimated with, respectively, k-means and a Dirichlet process mixture. Note that the number of clusters for k-means is chosen to maximize the average silhouette. For each run of the n-FuRBI model, we perform $25,000$ iterations of the MCMC chain and discard the first half as burn-in. Section S8.~contains results about convergence diagnostics, mixing performance, and computational times of the algorithm. 
Tables~\ref{tab:ri} and \ref{tab:K} summarize the performance of the models. The n-FuRBI priors outperform the alternatives in all scenarios considered, in term of estimating both the number of clusters and the  clustering configuration, measured by Rand indexes between the estimated configuration and the true clustering structure. Moreover, the posterior distribution of n-FuRBI models reflects uncertainty both about the estimated clustering configuration and about the imputation mechanism, which is instead ignored by two-step procedures.

\begin{figure}[tb!]
	\centering
	\includegraphics[width=0.75\textwidth]{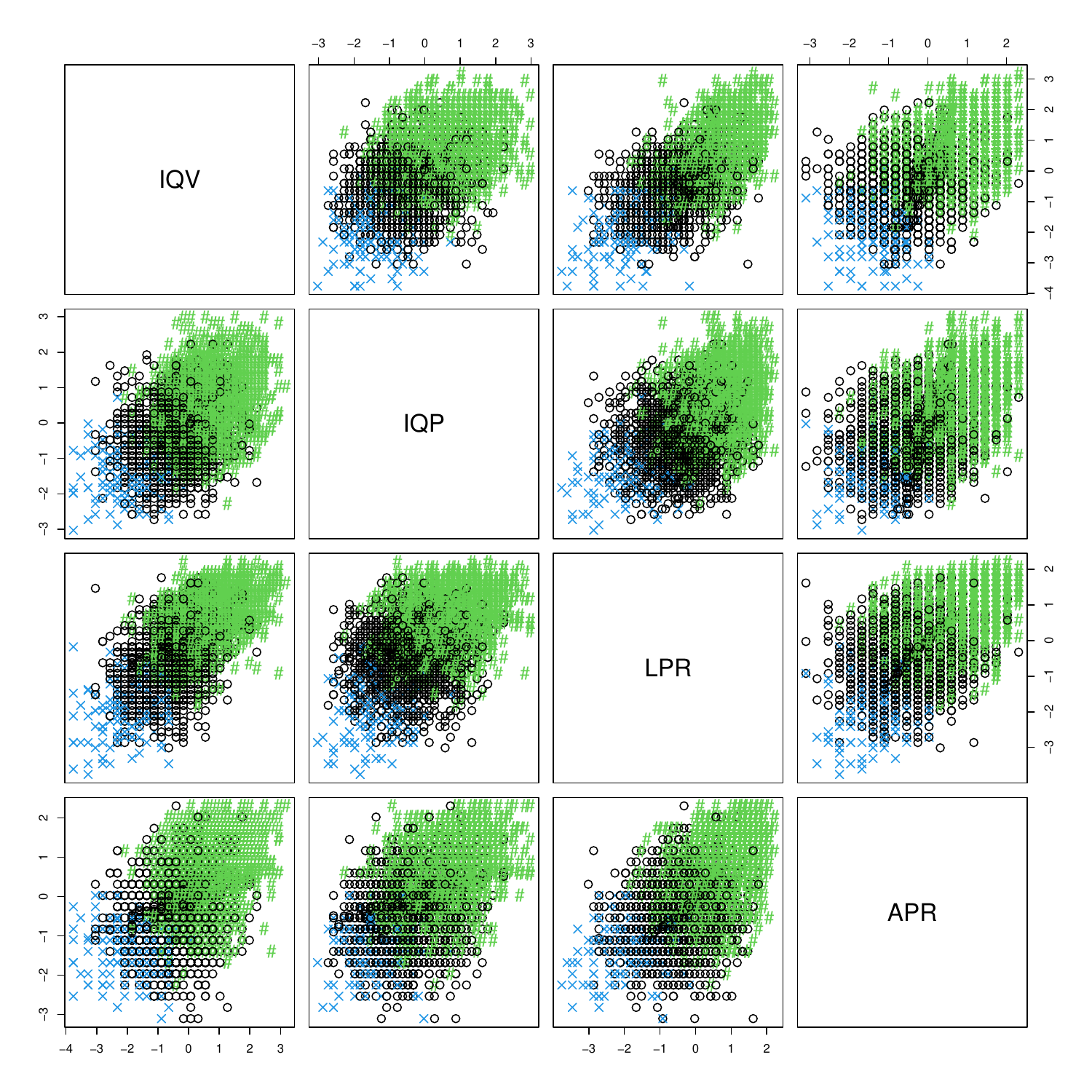}
	\caption{\label{fig:bradsma}{Scatter plots of the four scores (after standardization) for the  \texttt{brandsma} dataset. Coordinates of missing data are set equal to their respective posterior median. Different colors and symbols denote the three estimated clusters obtained minimizing the variation of information loss with respect to the posterior distribution.}}
\end{figure}

Finally, we apply the model also on the \texttt{brandsma} dataset (\cite{tom2012multilevel}), which refers to grade 8 students (age about 11 years) in elementary schools in the Netherlands \citep[see,][]{brandsma1989effects}. The goal is to cluster $n=4,106$ pupils, based on their IQ verbal score (IQV), IQ performance score (IQP), language score (LRP), and arithmetic score (APR). The number of subjects presenting missing entries is $339$ out of $4,106$ (i.e., $8.26\%$). As before, different combinations of missing variables define different samples: the number of samples is 7 in the \texttt{brandsma} dataset. In this real data analysis, the final clustering configuration provides a lower-dimensional description of the data rather than an estimate of ideal true clusters. Data are standardized before running the model, so that the sample means and variances are equal to $0$ and $1$. Figure~\ref{fig:bradsma} shows the estimated clustering configuration obtained minimizing the variation of information loss with respect to the posterior distribution. The model identifies three clusters, which show as major tendency that groups of students performing above/below average for one of the four scores tends to perform above/below average also for the other scores. In particular, a first cluster includes 53\% of the subjects, which have lower performances: indeed cluster averages of the standardized scores are IQV$=-0.371$, IQP$=-0.398$, LRP$=-0.387$, and APR$=-0.435$. Instead, the second cluster, including 44\% of the subjects, retains the best students: the cluster averages of the standardized scores are IQV$=0.609$, IQP$=0.595$, LRP$=0.629$, and APR$=0.642$. Finally, the students with the worst scores are allocated to a third cluster whose averages are IQV$=-2.01$, IQP$=-1.43$, LRP$=-1.90$, and APR$=-1.34$.

\section{Conclusion}\label{s:conclusion}
Hyper-ties play a crucial role in driving the  
Bayesian learning mechanism and the borrowing of information across samples. However existing nonparametric priors either do not allow an explicit evaluation of the probability of a hyper-tie or, when they do, often only non-negative correlation is induced. On the contrary, n-FuRBIs allow for analytical tractability and may induce either positive or negative correlation between the random probabilities as well as across samples resulting in a novel and flexible idea of borrowing of strength. They are immediately applicable to model multi-sample data through mixture models, as shown in Section~\ref{s:illustration1}. Morever, n-FuRBIs  also allow for a variety of interesting extensions, since they can be seen as an effective building block to model non-trivial dependencies in more complex data analyses. Future work will further explore these applications.

\section*{Acknowledgements}
F. Ascolani, A. Lijoi and I. Pr\"unster are partially supported by MIUR, PRIN Project P2022H5WZ9. Part of this work was carried out while B.\ 
Franzolini was a Research Fellow at the Agency for Science, Technology and Research, in Singapore, Republic of Singapore. B.\ Franzolini is supported by PNRR - PE1 FAIR - CUP B43C22000800006.

	\bibliographystyle{chicago}
	\bibliography{biblio}
	

\newpage
\newgeometry{top=20mm, bottom=20mm, left=20mm, right=20mm,
	foot=10mm, marginparsep=0mm}

\fancypagestyle{mypagestyle}{
	\fancyhf{}
	\fancyfoot[C]{A--\thepage}  
	\renewcommand{\headrulewidth}{0pt}
	\renewcommand{\footrulewidth}{0pt}
}
\makeatletter
\let\ps@plain\ps@mypagestyle
\makeatother

\pagestyle{mypagestyle}
\setcounter{page}{1}

\begin{center}
	{{\LARGE \textbf{Appendix} }}
\end{center}

\normalsize
\setcounter{section}{0}
\renewcommand\thesection{S\arabic{section}}
\setcounter{figure}{0}
\renewcommand{\thefigure}{S\arabic{figure}}

\section{Proofs}
\subsection{Proofs of Section 1}

\bigskip

\begin{proof}[Proof of Proposition 1]
	Consider two partially exchangeable sequences $X$ and $Y$ whose elements take value in $\mathbb{R}$. By de Finetti's representation theorem, there exist two random probability measures $\tilde p_1$ and $\tilde p_2$ such that
	\[
	\left( X_i, Y_j \right) \mid \tilde p_1, \tilde p_2 \overset{\text{iid}}{\sim} \tilde p_1 \times \tilde p_2.
	\]
	Note that $
	\mbox{cov}(X_i,Y_j)=\mathbb{E}\{\mbox{cov}(X_i,Y_j\mid\tilde p_1,\tilde p_2)\} + \mbox{cov}\{\mathbb{E}(X_i\mid\tilde p_1),\mathbb{E}(Y_j\mid\tilde p_2)\}$,
	where the first term equals 0, so that
	\[
	\mbox{cov}(X_i,Y_j)=\mbox{cov}\left(\int x\,\tilde p_1(dx),\int x\,\tilde p_2(dx)\right),
	\]
	and analogously
	\[
	\mbox{cov}(X_i,X_{i'})=\mbox{cov}\left(\int x\,\tilde p_1(dx),\int x\,\tilde p_1(dx)\right)=\mbox{var}\left(\int x\,\tilde p_1(dx)\right).
	\]
	Lastly assume that $\tilde p_1\overset{d}{=}\tilde p_2$, where $\overset{d}{=}$ indicates equality in distribution. By the Cauchy-Schwartz inequality 
	\[
	- \mbox{var}\left(\int x\,\tilde p_1(dx)\right)\leq \mbox{cov}\left(\int x\,\tilde p_1(dx),\int x\,\tilde p_2(dx)\right) \leq \mbox{var}\left(\int x\,\tilde p_1(dx)\right),
	\] 
	which, in terms of the observables, can be equivalently rewritten as 
	\[
	-\mbox{cov}(X_i,X_{i'})\leq \mbox{cov}(X_i,Y_j)\leq \mbox{cov}(X_i,X_{i'}).
	\]
\end{proof}

\bigskip

\begin{proof}[Proof of Proposition 2]
	By definition of covariance we have 
	\[
	\text{cov}(X_i, Y_j) = \text{cov}\left(\sum_{j \geq 1}J_j\theta_j, \sum_{k \geq 1}W_k\phi_k \right) = \sum_{j \geq 1}\sum_{k \geq 1}\text{cov}\left(J_j\theta_j, W_k\phi_k \right).
	\]
	For arbitrary $j$ and $k$ we have
	\[
	E\left(J_jW_k\theta_j\phi_k \right) = E(J_jW_k)E(\theta_j \phi_k) \geq E(J_jW_k)E(\theta_j)E( \phi_k),
	\]
	since cov$(\theta_j, \phi_k) \geq 0$. Denoting $c = E(\theta_j) = E( \phi_k)$, we get
	\[
	\text{cov}\left(J_j\theta_j, W_k\phi_k \right) \geq c^2\text{cov}(J_j, W_k). 
	\]
	Finally, since $\tilde p_1$ and $\tilde p_2$ are random probability measures it holds
	\[
	\text{cov}(X_i, Y_j) \geq c^2\text{cov}\left(\sum_{j \geq 1}J_j, \sum_{k \geq 1}W_k\right) = 0,
	\]
	which completes the proof.
\end{proof}
\subsection{Proofs of Section 2}

\bigskip

\begin{proof}[Proof of Proposition 3]
	Recall that 
	\[
	\beta :=  \sum_{k \geq 1}\E(\bar{J}^2_k) = \sum_{k \geq 1} \E(\bar{W}^2_k)\qquad 
	\gamma:= \sum_{k \geq 1} \E(\bar{J}_k\bar{W}_k).
	\]
	Since
	\[
	\E(\bar{J}_k\bar{W}_k) \leq \sqrt{\E(\bar{J}^2_k)\E(\bar{W}^2_k)}=\E(\bar{J}^2_k)
	\]
	it follows that $\gamma\leq\beta$. Moreover, the equality holds if and only if $\bar J_k\overset{a.s}{=}a_k+\bar W_k$, for any $k$, with $a_k\in \mathbb{R}$. However the equality of marginal distributions implies $a_k=0$.
\end{proof}

\bigskip

\begin{proof}[of Proposition 4]
	Recall that
	\[
	\text{cov}(X_i, Y_j) = \text{cov}\left(\sum_{k \geq 1}\bar J_k\theta_k, \sum_{h \geq 1}\bar W_h\phi_h\right) = \sum_{k \geq 1}\sum_{h \geq 1}\text{cov}\left(\bar J_k\theta_k, \bar W_h\phi_h \right).
	\]
	and for arbitrary $k$ and $h$, we have 
	\begin{equation*}
		\begin{aligned}
			\mathbb{E}(\bar J_k \bar W_{h}\theta_k\phi_{h})=&
			\E(\bar J_k \bar W_{h})\E(\theta_k\phi_{h})\\
			=&\E(\bar J_k \bar W_{h}) 
			\left[\E(\theta_k\phi_{k})\mathbbm{1}_{\{k=h\}} + \E(\theta_k)\E(\phi_{h})\mathbbm{1}_{\{k\neq h\}}\right],
		\end{aligned}
	\end{equation*}
	while 
	\[
	\mathbb{E}(\bar J_k \theta_k) = \E(\bar J_k ) \E(\theta_k).
	\]
	Thus, setting $c = E(\theta_k) = E( \phi_h)$, we have
	\begin{equation*}
		\text{cov}(X_i, Y_j)
		=
		\sum_{k \geq 1} \E(\bar J_k \bar W_{h}) \E(\theta_k\phi_{k}) 
		- c^2\sum_{k \geq 1} \E(\bar J_k) \E(\bar W_{k}) +
		+c^2\sum_{k \geq 1}\sum_{h \neq k} \text{cov}\left(\bar J_k, \bar W_{h}\right) 
	\end{equation*}
	where 
	\begin{equation*}
		\begin{aligned}
			\sum_{k \geq 1}\sum_{h \neq k}\text{cov}\left(\bar J_k, \bar W_{h}\right) =& 
			\text{cov}\left(\sum_{k \geq 1}\bar J_k, \sum_{h \geq 1 }\bar W_{h}\right)-\sum_{k\geq1}\text{cov} \left(\bar J_k \bar W_k \right)\\
			=&  - \sum_{k\geq1} \E(\bar J_k \bar W_{h})  + \sum_{k\geq1} \E(\bar J_k) \E(\bar W_{k})
		\end{aligned}
	\end{equation*}
	Putting everything together we obtain
	\begin{equation*}
		\begin{aligned}
			\text{cov}(X_i, Y_j) = \sum_{k \geq 1} \E[\bar J_k \bar W_{k}]  \text{cov}(\theta_k,\phi_k).
		\end{aligned}
	\end{equation*}
	Moreover 
	\[
	\mbox{var}(X_i) = \mbox{var}(Y_j) = \int\int x\, G_0(\d x, \d y) = \text{\mbox{var}}(\theta_k)
	\]
	Thus, $\mbox{corr}(X_i,Y_j)=\gamma\,\rho_0$ proving the second statement in Proposition 4.
	Finally, applying the same procedure marginally, we get 
	\[
	\text{cov}(X_i, X_i') = \sum_{k \geq 1} \E(\bar J_k ^2)\,\text{\mbox{var}}(\theta_k) 
	\]
	which proves the first statement in Proposition 4.
\end{proof}

\bigskip

\begin{proof}[Proof of Corollary 1]
	The result immediately follows 
	from Propositions 3 and 4.
\end{proof}

\bigskip

\begin{proof}[Proof of Proposition 5]
	Let $\beta$ be the probability of a tie. By definition 
	we get
	\[
	\begin{aligned}
		\P \left(X_1 \in A, X_2 \in B \right) =&\, \P (X_1 \in A, X_2 \in B \mid X_1=X_2) \beta + \\
		&+ 
		\P (X_1 \in A, X_2 \in B \mid X_1\neq X_2) (1-\beta),
	\end{aligned}
	\]
	which, by independence of the atoms, equals
	\[
	\begin{aligned}
		\P \left(X_1 \in A, X_2 \in B \right) =&\, \P (X_1 \in A\in B) \beta + \\
		&+ 
		\P (X_1 \in A) \P( X_2 \in B) (1-\beta).
	\end{aligned}
	\]
	Analogously, we have 
	\[
	\begin{aligned}
		\P \left(X_1 \in A, Y_1 \in B \right) =&\, \P (X_1 \in A, Y_1 \in B \mid X_1 \text{ and } Y_1 \text{ form an hyper-tie }) \gamma  +\\
		&+ 
		\P (X_1 \in A, Y_1 \in B \mid X_1 \text{ and } Y_1 \text{ do not form an hyper-tie })(1-\gamma),
	\end{aligned}
	\]
	where $\gamma$ is the probability of a hyper-tie, which equals
	\[
	\begin{aligned}
		\P \left(X_1 \in A, Y_1 \in B \right) =&\, \P ( (X_1,Y_1) \in A\times B \mid X_1 \text{ and } Y_1 \text{ form an hyper-tie } ) \gamma + \\
		&+ 
		\P (X_1 \in A) \P( Y_1 \in B) (1-\gamma).
	\end{aligned}
	\]
\end{proof}

\subsection{Proofs of Section 4}

\begin{proof}[Proof of Proposition 6]
	The first point follows from the L\'evy-Khintchine representation of the Laplace functional of a CRV. As 
	for (ii), one has
	\begin{equation*}
		\begin{aligned}
			\E \left( \e \{ -\lambda_1\tilde{\mu}_1(A) -\lambda_2\tilde{\mu}_2(B) \} \right)  =&\, \E \left( \e \{ -\lambda_1\mu_1(A \times \X) -\lambda_2\mu_2(\X \times B) \} \right)\\
			=&\,\E \bigl( \e \{ -\lambda_1\mu_1(A \times B^c) - \lambda_1\mu_1(A \times B)+\\
			&\,-\lambda_2\mu_2(A^c \times B) -\lambda_2\mu_2(A \times B)\} \bigr).
		\end{aligned}
	\end{equation*}
	By independence of evaluations on disjoint sets, 
	$\mu_1(C)$ and $\mu_2(D)$ are independent if $C \cap D = \emptyset$, so that the right hand side reads
	\begin{equation*}
		\begin{aligned}
			\E \bigl( \e \{ -\lambda_1\tilde{\mu}_1(A) -\lambda_2\tilde{\mu}_2(B) \} \bigr)  =&  \E \left( \e \{ -\lambda_1\mu_1(A \times B^c) \} \right) \E \left( \e \{ -\lambda_2\mu_2(A^c \times B) \} \right)\times \\
			&\times\E \left( \e \{ -\lambda_1\mu_1(A \times B)- \lambda_2\mu_2(A \times B) \} \right).
		\end{aligned}
	\end{equation*}
	The result follows upon 
	upon using the expressions of the marginal and joint Laplace exponents of $\mu_1$ and $\mu_2$. Since from the joint L\'evy intensity it is possible to recover the joint Laplace exponent, 
	(iii) is also proved.
\end{proof}

In order to prove Proposition 7, we 
show that
\[
\P \left( X \in A, Y \in B\right) = P_0(A)P_0(B)\left( 1-\delta \right) + G_0(A \times B)\delta,
\]
where
\[
\delta := -\int_{\R_+^2} \left\{ \frac{\partial^2}{\partial u_1 \partial u_2} \psi_b(u_1, u_2) \right\}e^{- \psi_b(u_1, u_2)}\, \d u_1 \d u_2.
\]
is the probability of a pseudo-tie. 
We start with three Lemmas.

\begin{lemma}\label{tech1}
	If $\psi_b$ is the joint Laplace exponent of a CRV, then
	\begin{equation*}
		\begin{aligned}
			\int_{\R_+^2}& \left\{ \frac{\partial}{\partial u_1}\psi_b(u_1, u_2)\right\}\left\{\frac{\partial}{\partial u_2}\psi_b(u_1, u_2)\right\}e^{- \psi_b(u_1, u_2)}\, \d u_1 \d u_2 = 1-\delta.
		\end{aligned}
	\end{equation*}
\end{lemma}
\begin{proof}[Proof of Lemma S2.1]
	Integrating by parts
	\begin{equation*}
		\begin{aligned}
			\int_0^\infty& \left\{ \frac{\partial}{\partial u_1}\psi_b(u_1, u_2)\right\}\left\{\frac{\partial}{\partial u_2}\psi_b(u_1, u_2)\right\}e^{- \psi_b(u_1, u_2)}\, \d u_1\\
			&= -\int_0^\infty \left\{ \frac{\partial}{\partial u_2}\psi_b(u_1, u_2)\right\}\left\{\frac{\partial}{\partial u_1}e^{-\psi_b(u_1, u_2)}\right\}\, \d u_1\\
			&=  \biggl[\left[-\left\{ \frac{\partial}{\partial u_2}\psi_b(u_1, u_2)\right\}e^{- \psi_b(u_1, u_2)} \right]_0^\infty +\int_0^\infty \left\{ \frac{\partial^2}{\partial u_1 \partial u_2} \psi_b(u_1, u_2) \right\}e^{- \psi_b(u_1, u_2)}\, \d u_1 \biggr]\\
			& =  \biggl[\left\{ \frac{\partial}{\partial u_2}\psi_b(0, u_2)\right\}e^{- \psi_b(0, u_2)}+\int_0^\infty \left\{ \frac{\partial^2}{\partial u_1 \partial u_2} \psi_b(u_1, u_2) \right\}e^{- \psi_b(u_1, u_2)}\, \d u_1 \biggr].
		\end{aligned}
	\end{equation*}
	Note that $\int_0^\infty \left\{ \frac{\d}{\d u_2}\psi_b(0, u_2)\right\}e^{- \psi_b(0, u_2)} \, \d u_2 = 1$, by the fundamental theorem of calculus. Thus the result follows immediately.
\end{proof}

\begin{lemma}\label{tech2}
	We have
	\begin{equation*}
		\begin{aligned}
			\int_{\R_+^2}&\E \left(e^{-u_1\mu_1(\X  \times \X)-u_2\mu_2(\X \times \X)}\mu_1(C)\mu_2(C)\right)\, \d u_1 \d u_2 =  G_0(C)^2\left( 1-\delta \right) + G_0(C)\delta.
		\end{aligned}
	\end{equation*}
\end{lemma}
\begin{proof}[Proof of Lemma S2.2]
	By independence of evaluations on disjoint sets it follows that
	\begin{equation*}
		\begin{aligned}
			&\int_{\R_+^2}\E \left(e^{-u_1\mu_1(\X  \times \X)-u_2\mu_2(\X \times \X)}\mu_1(C)\mu_2(C)\right)\, \d u_1 \d u_2\\
			&= \int_{\R_+^2}\E \left(e^{-u_1\mu_1(C)-u_2\mu_2(C)-u_1\mu_1(C^c)-u_2\mu_2(C^c)}
			\mu_1(C)\mu_2(C)\right)\, \d u_1 \d u_2 \\
			&=\int_{\R_+^2}\E \left(e^{-u_1\mu_1(C)-u_2\mu_2(C)}\mu_1(C)\mu_2(C)\right\}\E \left( e^{-u_1\mu_1(C^c)-u_2\mu_2(C^c)} \right) \, \d u_1 \d u_2 \\
			& = \int_{\R_+^2}\E \left(\frac{\partial }{\partial u_1}\frac{\partial }{\partial u_2}e^{-u_1\mu_1(C)-u_2\mu_2(C)}\right)\E \left( e^{-u_1\mu_1(C^c)-u_2\mu_2(C^c)} \right) \, \d u_1 \d u_2 \\
			& = \int_{\R_+^2}\frac{\partial }{\partial u_1}\frac{\partial }{\partial u_2}\left[ \E \left(e^{-u_1\mu_1(C)-u_2\mu_2(C)}\right) \right]\E \left( e^{-u_1\mu_1(C^c)-u_2\mu_2(C^c)} \right) \, \d u_1 \d u_2\\
			&= \int_{\R_+^2}\frac{\partial }{\partial u_1}\frac{\partial }{\partial u_2} \left\{e^{-G_0(C)\psi_b(u_1, u_2)} \right\} e^{-G_0(C^c)\psi_b(u_1, u_2)}\, \d u_1 \d u_2\\
			& = \int_{\R_+^2}\frac{\partial }{\partial u_1} \left\{-G_0(C)\frac{\partial}{\partial u_2}\psi_b(u_1, u_2)e^{-G_0(C)\psi_b(u_1, u_2)} \right\}e^{-G_0(C^c)\psi_b(u_1, u_2)}\, \d u_1 \d u_2.
		\end{aligned}
		\]
		Performing the derivative with respect to $u_1$, the latter expression can be written as follows
		\[
		\begin{aligned}
			& = \int_{\R_+^2} \left\{G_0(C)^2\frac{\partial}{\partial u_1}\psi_b(u_1, u_2)\frac{\partial}{\partial u_2}\psi_b(u_1, u_2)\right\}e^{-G_0(C)\psi_b(u_1, u_2)}e^{-G_0(C^c)\psi_b(u_1, u_2)}\, \d u_1 \d u_2 +\\
			&+ \int_{\R_+^2} \left\{-G_0(C)\frac{\partial}{\partial u_1 \partial u_2}\psi_b(u_1, u_2)\right\}e^{-G_0(C)\psi_b(u_1, u_2)}e^{-G_0(C^c)\psi_b(u_1, u_2)}\, \d u_1 \d u_2\\
			& = \int_{\R_+^2} \left\{G_0(C)^2\frac{\partial}{\partial u_1}\psi_b(u_1, u_2)\frac{\partial}{\partial u_2}\psi_b(u_1, u_2)\right\}e^{- \psi_b(u_1, u_2)}\, \d u_1 \d u_2 +\\
			& + \int_{\R_+^2} \left\{-G_0(C)\frac{\partial}{\partial u_1 \partial u_2}\psi_b(u_1, u_2)\right\}e^{ \psi_b(u_1, u_2)}\, \d u_1 \d u_2
		\end{aligned}
	\end{equation*}
	By Lemma \ref{tech1} we then obtain
	\begin{equation*}
		\int_{\R_+^2}\E \left(e^{-u_1\mu_1(\X  \times \X)-u_2\mu_2(\X \times \X)}\mu_1(C)\mu_2(C)\right)\, \d u_1 \d u_2 = G_0(C)^2\left( 1-\delta\right) + G_0(C)\delta,
	\end{equation*}
	as desired.
\end{proof}

\begin{lemma}\label{tech3}
	Let $C,D$ be such that $C \cap D = \emptyset$. Then
	\begin{equation*}
		\int_{\R_+^2}\E \left(e^{-u_1\mu_1(\X  \times \X)-u_2\mu_2(\X \times \X)}\mu_1(C)\mu_2(D)\right)\, \d u_1 \d u_2   
		= G_0(C)G_0(D)\left( 1-\delta \right)
	\end{equation*}
\end{lemma}

\begin{proof}[Proof of Lemma S3]
	Let $Y = \left(C \cup D\right)^c$. Since $C$ and $D$ are disjoint, by independence of evaluations on disjoint sets it holds
	\begin{equation*}
		\begin{aligned}
			\int_{\R_+^2}&\E \left(e^{-u_1\mu_1(\X  \times \X)-u_2\mu_2(\X \times \X)}\mu_1(C)\mu_2(D)\right)\, \d u_1 \d u_2  \\
			= \int_{\R_+^2}&\E \left(e^{-u_1\mu_1(C \cup D)-u_2\mu_2(C \cup D)}\mu_1(C)\mu_2(D)\right\}\E \left\{ e^{-u_1\mu_1(Y)-u_2\mu_2(Y)} \right) \, \d u_1 \d u_2 \\
			=\int_{\R_+^2}&\E \left(e^{-u_1\mu_1(C)-u_2\mu_2(C)}\mu_1(C)\right)\E \left(e^{-u_1\mu_1(D)-u_2\mu_2(D)}\mu_2(D)\right)\times\\
			&\times\E \left(e^{-u_1\mu_1(Y)-u_2\mu_2(Y)} \right) \, \d u_1 \d u_2\\
			=\int_{\R_+^2}& \frac{\partial}{\partial u_1}\left\{ e^{-G_0(C) \psi_b(u_1, u_2)}\right\}\frac{\partial}{\partial u_2}\left\{ e^{-G_0(D) \psi_b(u_1, u_2)}\right\}
			e^{-G_0(Y)\psi_b(u_1, u_2)} \, \d u_1 \d u_2\\
			=G_0&(C)G_0(D)\int_{\R_+^2} \left\{\frac{\partial}{\partial u_1}\psi_b(u_1, u_2)\frac{\partial}{\partial u_2}\psi_b(u_1, u_2)\right\}e^{- \psi_b(u_1, u_2)}\, \d u_1 \d u_2
		\end{aligned}
	\end{equation*}
	The result follows by applying Lemma \ref{tech1}.
\end{proof}

\bigskip

\begin{proof}[Proof of Proposition 7]
	We have
	\begin{equation*}
		\begin{aligned}
			&\P \left( X \in A, Y \in B\right) = \E \left( \frac{\tilde{\mu}_1(A)}{\tilde{\mu}_1(\X)}\frac{\tilde{\mu}_2(B)}{\tilde{\mu}_2(\X)} \right) = \E \left( \frac{\mu_1(A \times \X)}{\mu_1(\X \times \X)}\frac{\mu_2(\X \times B)}{\mu_2(\X \times \X)} \right)=\\
			&= \int_{\R_+^2}\E \left(e^{-u_1\mu_1(\X \times \X)-u_2\mu_2(\X \times \X)}\mu_1(A \times \X)\mu_2(\X \times B)\right)\, \d u_1 \d u_2 = \\
			&=  \int_{\R_+^2}\E \biggl( e^{-u_1\mu_1(\X  \times \X)-u_2\mu_2(\X \times \X)}\bigl\{\mu_1(A \times B)\mu_2(A \times B)+\mu_1(A \times B)\mu_2(A^c \times B)+\\
			&+\mu_1(A \times B^c)\mu_2(A \times B)+\mu_1(A \times B^c)\mu_2(A^c \times B)\bigr\} \biggr) \, \d u_1 \d u_2
		\end{aligned}
	\end{equation*}
	We compute each integral separately applying Lemmas \ref{tech2} and \ref{tech3} and obtain 
	\begin{equation}\label{simple_pred}
		\begin{aligned}
			\P \left( X \in A, Y \in B\right) &= G_0(A \times \X)G_0(\X \times B)\left( 1-\delta \right) + G_0(A \times B)\delta\\
			& = P_0(A)P_0(B)\left( 1-\delta \right) + G_0(A \times B)\delta,
		\end{aligned}
	\end{equation}
	as desired. Then the probability of a tie in the product space is given exactly by $\delta$, denoted $\gamma$ in the 
	manuscript. The probability of a tie is given by the particular case $\psi_b(u_1, u_2) = \psi(u_1+u_2)$, since
	\[
	-\int_{\R_+^2} \left\{ \frac{\partial^2}{\partial u_1 \partial u_2} \psi_b(u_1 + u_2) \right\}e^{- \psi_b(u_1 + u_2)}\, \d u_1 \d u_2 = -\int_0^{\infty}\int_0^u \, \d v\left\{ \frac{\partial^2}{\partial u^2} \psi_b(u) \right\}e^{- \psi_b(u)} \, \d u,
	\]
	with the change of variables $u = u_1+u_2$ and $v = u_1$.
\end{proof}

\bigskip

\begin{proof}[Proof of Proposition 8]
	Since
	\[
	\E \left( \tilde{p}_1(A)\tilde{p}_2(B) \right) = \P \left(X \in A, Y \in B \right),
	\]
	by \eqref{simple_pred} we have
	\begin{equation*}
		\E \left( \tilde{p}_1(A)\tilde{p}_2(B) \right) = G_0(A \times \X)G_0(\X \times B)\left( 1-\gamma \right) + G_0(A \times B)\gamma.
	\end{equation*}
	Finally,
	\begin{equation*}
		\begin{aligned}
			\text{cov} \left( \tilde{p}_1(A),\tilde{p}_2(B) \right) &=  G_0(A \times \X)G_0(\X \times B)\left( 1-\gamma \right) + G_0(A \times B)\gamma - G_0(A \times \X)G_0(\X \times B)\\
			&=  \gamma \left\{G_0(A \times B)- G_0(A \times \X)G_0(\X \times B)\right\}.
		\end{aligned}
	\end{equation*}
	From this one also obtains
	\[
	\begin{aligned}
		\text{\mbox{var}} \left( \tilde{p}_1(A) \right) &= \text{cov} \left( \tilde{p}_1(A),\tilde{p}_1(A) \right) = \beta \left\{P_0(A)- P_0(A)^2\right\}\\
		& = \beta P_0(A)\left\{1-P_0(A) \right\},
	\end{aligned}
	\]
	and the desired result follows.
\end{proof}

\bigskip

\subsection{Proofs of Section 5}


\begin{proof}[Proof of Theorem 1]
	We need to compute the conditional Laplace functional of $(\mu_1, \mu_2)$, i.e.
	\[
	\E \left(e^{-\int_{\X^2}h_1(x)\, \mu_1(\d x)-\int_{\X^2}h_2(x)\, \mu_2(\d x)} \mid (X_{i})_{i=1}^{n}, (Y_{j})_{j=1}^{m}  \right),
	\]
	with $h_i \, : \X^2 \, \to  \R^+$ measurable functions. Define $A_j = A_{j,\epsilon} = \left\{x \in \X \mid d(x, X_i^*) < \epsilon \right\}$ and  $B_j = B_{j,\epsilon} = \left\{x \in \X \mid d(x, Y_j^*) < \epsilon \right\}$, with $1 \leq i \leq k$ and $1 \leq j \leq c$, such that $A_i \cap A_j = \emptyset$ and $B_i \cap B_j = \emptyset$ for any $i \neq j$. Moreover, denote
	\[
	A_{k+1} = \left( \cup_{i = 1}^k A_i \right)^c, \quad B_{c+1} = \left( \cup_{i = 1}^c B_i \right)^c.
	\]
	Thus our goal becomes to compute
	\begin{equation}\label{to_compute}
		\begin{aligned}
			&\E \left(e^{-\int_{\X^2}h_1(x)\, \mu_1(\d x)-\int_{\X^2}h_2(x)\, \mu_2(\d x)} \mid (X_{i})_{i=1}^{n}, (Y_{j})_{j=1}^{m} \right)\\
			& = \lim_{\epsilon \to 0} \E \left(e^{-\int_{\X^2}h_1(x)\, \mu_1(\d x)-\int_{\X^2}h_2(x)\, \mu_2(\d x)} \mid \underline{X}_{n}^* \in \times_{j = 1}^kA_j, \underline{Y}_{m}^* \in \times_{j = 1}^cB_j  \right)\\
			& = \lim_{\epsilon \to 0}\frac{\E\left(e^{-\int_{\X^2}h_1(x)\, \mu_1(\d x)-\int_{\X^2}h_2(x)\, \mu_2(\d x)} \prod_{j = 1}^k\tilde{p}_1(A_j)^{n_j}\prod_{j = 1}^c\tilde{p}_2(B_j)^{m_j} \right)}{\E \left( \prod_{j = 1}^k\tilde{p}_1(A_j)^{n_j}\prod_{j = 1}^c\tilde{p}_2(B_j)^{m_j} \right)}.
		\end{aligned}
	\end{equation}
	We start to evaluate
	\[
	\begin{aligned}
		&\E \left(\tilde{p}_1(A_1)^{n_1} \dots \tilde{p}_1(A_k)^{n_k}\tilde{p}_2(B_1)^{m_1}\tilde{p}_2(B_c)^{m_c}\right) =\\
		&= \E \left( \frac{\tilde{\mu}_1(A_1)^{n_1} \dots \tilde{\mu}_1(A_k)^{n_k}\tilde{\mu}_2(B_1)^{m_1}\tilde{\mu}_2(B_c)^{m_c}}{\tilde{\mu}_1(\X)^n\tilde{\mu}_2(\X)^m} \right)\\
		&= \E \left( \frac{\mu_1(A_1 \times \X)^{n_1} \dots \mu_1(A_k \times \X)^{n_k}\mu_2(\X \times B_1)^{m_1}\mu_2(\X \times B_c)^{m_c}}{\mu_1(\X \times \X)^n\mu_2(\X \times \X)^m} \right) = \mathcal{I}.
	\end{aligned}
	\]
	By Netwon's binomial
	\[
	\begin{aligned}
		&\mu_1(A_h \times \X) = \sum_{i_1^h+ \dots i_{c+1}^h = n_h}\binom{n_h}{i_1^h, \dots, i_{c+1}^h} \prod_{r = 1}^{c+1}\mu_1^{i_r^h}(A_h \times B_r), \quad h = 1, \dots, k,\\
		&\mu_2(\X \times B_r) = \sum_{j_1^r+ \dots j_{k+1}^r = m_r}\binom{m_r}{j_1^r, \dots, j_{k+1}^r} \prod_{h = 1}^{k+1}\mu_2^{j_h^r}(A_h \times B_r), \quad r = 1, \dots, c.
	\end{aligned}
	\]
	For ease of notation denote
	\[
	\begin{aligned}
		\sum_{\bm{i}, \bm{j}} \binom{\bm{n}}{\bm{i}}\binom{\bm{m}}{\bm{j}} =& \sum_{i_1^1+ \dots i_{c+1}^1 = n_1}\binom{n_1}{i_1^1, \dots, i_{c+1}^1}\dots \sum_{i_1^{c+1}+ \dots i_{c+1}^{k+1} = n_{k+1}}\binom{n_{k+1}}{i_1^{k+1}, \dots, i_{c+1}^{k+1}}\times\\
		&\times \sum_{j_1^1+ \dots j_{k+1}^1 = m_1}\binom{m_1}{j_1^1, \dots, j_{k+1}^1} \dots \sum_{j_1^{k+1}+ \dots j_{k+1}^{k+1} = m_{k+1}}\binom{m_{k+1}}{j_1^{k+1}, \dots, j_{k+1}^{k+1}}.
	\end{aligned}
	\]
	Thus
	\[
	\mathcal{I} = \sum_{\bm{i}, \bm{j}} \binom{\bm{n}}{\bm{i}}\binom{\bm{m}}{\bm{j}} \mathcal{I}_{\bm{i}, \bm{j}},
	\]
	with
	\[
	\begin{aligned}
		\mathcal{I}_{\bm{i}, \bm{j}} = \E \Biggl( &\frac{\prod_{h = 1}^{k}\prod_{r = 1}^{c}\mu_1^{i_r^h}(A_h \times B_r)\mu_2^{j_h^r}(A_h \times B_r)}{\mu_1(\X \times \X)^n}\times\\
		&\times\frac{\prod_{h = 1}^k\mu_1^{i_{c+1}^h}(A_h\times B_{c+1})\prod_{r = 1}^c\mu_2^{j_{k+1}^r}(A_{k+1}\times B_{r})}{\mu_2(\X \times \X)^m} \Biggr)
	\end{aligned}
	\]
	Letting $\mu_1 := \mu_1(\X \times \X)$ and $\mu_2 := \mu_2(\X \times \X)$, we have
	\[
	\frac{1}{\mu_1(\X \times \X)^n\mu_2(\X \times \X)^m} = \frac{1}{\Gamma(n)\Gamma(m)}\int_{\R_+^2}u_1^{n-1}u_2^{m-1}e^{-u_1\mu_1-u_2\mu_2}\, \d \underline{u},
	\]
	with $\underline{u} = (u_1, u_2)$. Thus, by Fubini's Theorem
	\[
	\begin{aligned}
		\mathcal{I}_{\i, \j} = \int_{\R_+^2}\frac{u_1^{n-1}u_2^{m-1}}{\Gamma(n)\Gamma(m)}&\E \biggl(e^{-u_1\mu_1-u_2\mu_2}\left\{\prod_{h = 1}^{k}\prod_{r = 1}^{c}\mu_1^{i_r^h}(A_h \times B_r)\mu_2^{j_h^r}(A_h \times B_r)\right\}\times\\
		&\times\prod_{h = 1}^k\mu_1^{i_{c+1}^h}(A_h\times B_{c+1})\prod_{r = 1}^c\mu_2^{j_{k+1}^r}(A_{k+1}\times B_{r}) \biggr)\, \d \underline{u} = \\
		& = \int_{\R_+^2}\frac{u_1^{n-1}u_2^{m-1}}{\Gamma(n)\Gamma(m)} \rho_{\i, \j}(\underline{u}) \, \d \underline{u}.
	\end{aligned}
	\]
	By independence of evaluations on disjoint sets we have
	\[
	\begin{aligned}
		\rho_{\i, \j}(\underline{u}) &= \E \Biggl(\left\{\prod_{h = 1}^{k}\prod_{r = 1}^{c}e^{-u_1\mu_1(A_h \times B_r)-u_2\mu_2(A_h \times B_r)}\mu_1^{i_r^h}(A_h \times B_r)\mu_2^{j_h^r}(A_h \times B_r)\right\}\times\\
		&\times\left\{\prod_{h = 1}^ke^{-u_1\mu_1(A_h \times B_{c+1})-u_2\mu_2(A_h \times B_{c+1})}\mu_1^{i_{c+1}^h}(A_h\times B_{c+1})\right\}\times\\
		&\times\left\{\prod_{r = 1}^ce^{-u_1\mu_1(A_{k+1} \times B_r)-u_2\mu_2(A_{k+1} \times B_r)}\mu_2^{j_{k+1}^r}(A_{k+1}\times B_{r})\right\} \Biggr)
	\end{aligned}
	\]
	This can be equivalently written as
	\[
	\begin{aligned}
		\prod_{h = 1}^{k}\prod_{r = 1}^{c}&\E \left( e^{-u_1\mu_1(A_h \times B_r)-u_2\mu_2(A_h \times B_r)}\mu_1^{i_r^h}(A_h \times B_r)\mu_2^{j_h^r}(A_h \times B_r)\right)\times\\
		&\times\prod_{h = 1}^k\E\left(e^{-u_1\mu_1(A_h \times B_{c+1})-u_2\mu_2(A_h \times B_{c+1})}\mu_1^{i_{c+1}^h}(A_h\times B_{c+1})\right)\times\\
		&\times\prod_{r = 1}^c\E\left(e^{-u_1\mu_1(A_{k+1} \times B_r)-u_2\mu_2(A_{k+1} \times B_r)}\mu_2^{j_{k+1}^r}(A_{k+1}\times B_{r})\right).
	\end{aligned}
	\]
	Considering each element separately we have
	\[
	\begin{aligned}
		\E \biggl( e^{-u_1\mu_1(A_h \times B_r)-u_2\mu_2(A_h \times B_r)}&\mu_1^{i}(A_h \times B_r)\mu_2^{j}(A_h \times B_r)\biggr) \\
		&= \E \left((-1)^{i+j}\frac{\partial^{i+j}}{\partial u_1^i\partial u_2^j} e^{-u_1\mu_1(A_h \times B_r)-u_2\mu_2(A_h \times B_r)}\right)\\
		&= (-1)^{i+j}\frac{\partial^{i+j}}{\partial u_1^i\partial u_2^j}\E \left( e^{-u_1\mu_1(A_h \times B_r)-u_2\mu_2(A_h \times B_r)}\right)\\
		& =(-1)^{i+j} \frac{\partial^{i+j}}{\partial u_1^i\partial u_2^j}\left\{e^{-\int_{A_h \times B_r}\int_{\R_+^2}(1-e^{-u_1s_1-u_2s_2})\, \rho(ds)G_0(x)}\right\}.
	\end{aligned}
	\]
	Recall that we are interested in the limit as $\epsilon \to 0$, so that
	\begin{equation}\label{order_infinitesim}
		\begin{aligned}
			\frac{\partial^{i+j}}{\partial u_1^i\partial u_2^j}\biggl\{&e^{-\int_{A_h \times B_r}\int_{\R_+^2}(1-e^{-u_1s_1-u_2s_2})\, \rho(\d s)G_0(\d x)}\biggr\} \sim e^{-\int_{A_h \times B_r}\int_{\R_+^2}(1-e^{-u_1s_1-u_2s_2})\, \rho(\d s)G_0(\d x)}\times\\
			&\times\frac{\partial^{i+j}}{\partial u_1^i\partial u_2^j}\left\{\int_{A_h \times B_r} \int_{\R_+^2}(1-e^{-u_1s_1-u_2s_2})\, \rho(\d s)G_0(\d x) \right\},
		\end{aligned}
	\end{equation}
	where we say $f \sim g$ if $\lim_{\epsilon \to 0}f(x)/g(x) = 1$. By simple algebra we get
	\[
	\begin{aligned}
		\frac{\partial^{i+j}}{\partial u_1^i\partial u_2^j}&\left\{e^{-\int_{A_h \times B_r}\int_{\R_+^2}(1-e^{-u_1s_1-u_2s_2})\, \rho(\d s)G_0(\d x)}\right\}=\frac{\partial^{i+j-1}}{\partial u_1^{i-1}\partial u_2^j}\biggl\{-\int_{A_h \times B_r}\int_{\R_+^2}e^{-u_1s_1-u_2s_2}\times \\
		&\times s_1\, \rho(\d s)G_0(\d x)e^{-\int_{A_h \times B_r}\int_{\R_+^2}(1-e^{-u_1s_1-u_2s_2})\, \rho(\d s)G_0(\d x)} \biggr\}\\
		&= \frac{\partial^{i+j-2}}{\partial u_1^{i-2}\partial u_2^j}\biggl\{\int_{A_h \times B_r}\int_{\R_+^2}e^{-u_1s_1-u_2s_2}s_1^2\, \rho(\d s)G_0(\d x)e^{-\int_{A_h \times B_r}\int_{\R_+^2}(1-e^{-u_1s_1-u_2s_2})\, \rho(\d s)G_0(\d x)}\\
		&+\left(\int_{A_h \times B_r}\int_{\R_+^2}e^{-u_1s_1-u_2s_2}s_1\, \rho(\d s)G_0(\d x) \right)^2 e^{-\int_{A_h \times B_r}\int_{\R_+^2}(1-e^{-u_1s_1-u_2s_2})\, \rho(\d s)G_0(\d x)}\biggr\},
	\end{aligned}
	\]
	and
	\[
	\lim_{\epsilon \to 0} \frac{\left(\int_{A_h \times B_r}\int_{\R_+^2}e^{-u_1s_1-u_2s_2}s_1\, \rho(\d s)G_0(\d x) \right)^2}{\int_{A_h \times B_r}\int_{\R_+^2}e^{-u_1s_1-u_2s_2}s_1^2\, \rho(\d s)G_0(\d x)} = 0.
	\]
	By applying this argument repeatedly we obtain \eqref{order_infinitesim}. Thus, letting $\rho({\underline{u}}) = \sum_{\bm{i}, \bm{j}} \binom{\bm{n}}{\bm{i}}\binom{\bm{m}}{\bm{j}}\rho_{\bm{i}, \bm{j}}({\underline{u}})$, by aggregating the terms we have
	\[
	\begin{aligned}
		\rho({\underline{u}}) \sim& \sum_{\bm{i}, \bm{j}} \binom{\bm{n}}{\bm{i}}\binom{\bm{m}}{\bm{j}} (-1)^{n+m}e^{-\psi_b(u)}\times\\
		&\times\prod_{h = 1}^k\prod_{r = 1}^c \left\{\frac{\partial^{i_r^h+j_h^r}}{\partial u_1^{i_r^h}\partial u_2^{j_h^r}}\int_{A_h \times B_r} \int_{\R_+^2}(1-e^{-u_1s_1-u_2s_2})\, \rho(\d s)G_0(\d x) \right\}\times\\
		&\times\prod_{h = 1}^k\left\{\frac{\partial^{i_{c+1}^h}}{\partial u_1^{i^h_{c+1}}} \int_{A_h \times B_{c+1}} \int_{\R_+^2}(1-e^{-u_1s_1-u_2s_2})\, \rho(\d s)G_0(\d x) \right\}\times\\
		&\times\prod_{r = 1}^c\left\{\frac{\partial^{j_{k+1}^r}}{\partial u_2^{i^r_{k+1}}} \int_{A_{k+1} \times B_r} \int_{\R_+^2}(1-e^{-u_1s_1-u_2s_2})\, \rho(\d s)G_0(\d x) \right\}\\
		& = \sum_{\bm{i}, \bm{j}} \binom{\bm{n}}{\bm{i}}\binom{\bm{m}}{\bm{j}}(-1)^{n+m}V(\bm{i}, \bm{j}).
	\end{aligned}
	\]
	The following three Lemmas characterize the set of indices $(\bm{i}, \bm{j})$ that are relevant once the limit is taken.
	
	\bigskip
	
	\begin{lemma}\label{lemma1}
		Consider $(\bm{i}, \bm{j})$ such that $0 < i_r^h, i_l^h < n_h$, with $r > l$ and $1 \leq h \leq k$. Then $\exists (\tilde{\bm{i}}, \tilde{\bm{j}})$ such that $\lim_{\epsilon \to 0} V(\bm{i}, \bm{j})/V(\tilde{\bm{i}}, \tilde{\bm{j}}) \to 0$.
	\end{lemma}
	
	\begin{proof}[Proof of Lemma S2.4]
		For ease of notation set $\bm{i}^h = (i^h_1, \dots, i^h_{c+1})$. Then
		\begin{itemize}
			\item If $r = c+1$, set $\tilde{\bm{i}}^h = (i_1^h, \dots, i_l^h+i^h_{c+1}, \dots, 0)$.
			\item If $j_h^r = 0$, set $\tilde{i}^h = (i_1^h, \dots, i_l^h+i_r^h, \dots, 0, \dots)$.
			\item If $j_h^l = 0$, set $\tilde{\bm{i}}^h = (i_1^h, \dots, 0, \dots, i_r^h+i_l^h, \dots)$.
			\item If $j_h^l > 0$ and $j_h^r > 0$, set $\tilde{\bm{j}}^r = (j_1^r,  \dots, 0,\dots, j^r_{k+1}+j_h^r)$ and $\tilde{\bm{i}}^h = (i_1^h, \dots, i_l^h+i_r^h, \dots, 0, \dots)$.
		\end{itemize}
		For example in the last case we have
		\[
		\lim_{\epsilon \to 0} \frac{\mbox{var}(\bm{i}, \bm{j})}{\mbox{var}(\tilde{\bm{i}}, \tilde{\bm{j}})} = \lim_{\epsilon \to 0} \frac{\int_{A_h \times B_r}\int_{\R_+^2}e^{-u_1s_1-u_2s_2}s_1^{i^h_r}s_2^{j^r_h}\, \rho(\d s)G_0(\d x)}{\int_{A_{c+1} \times B_r}\int_{\R_+^2}e^{-u_1s_1-u_2s_2}s_2^{j^r_h+j^r_{c+1}}\, \rho(\d s)G_0(\d x)} = 0,
		\]
		as desired.
	\end{proof}
	Thus, Lemma \ref{lemma1} guarantees that $\bm{i}^h$ has exactly one element different from $0$, that is equal to $n_h$.
	
	\medskip
	
	\begin{lemma}\label{lemma2}
		Consider $(\i, \j)$ such that $i_r^h = n_h$ and $j_h^r = 0$. Then there exists $ (\tilde{\i}, \tilde{\j})$ such that
		\[
		\lim_{\epsilon \to 0} V(\i, \j)/V(\tilde{\i}, \tilde{\j}) \to 0.
		\] 
	\end{lemma}
	\begin{proof}[Proof of Lemma S5]
		Set $(\tilde{\bm{i}}, \tilde{\bm{j}})$ equal to $(\i, \j)$, apart from $\tilde{i}^h_r = 0$ and $\tilde{i}^h_{c+1} = n_h$.
	\end{proof}
	
	\bigskip
	
	\begin{lemma}\label{lemma3}
		Consider $(\i, \j)$ such that $i^h_{c+1} = n_h$ and $j_h^r > 0$. Then there exists $(\tilde{\i}, \tilde{\j})$ such that
		\[
		\lim_{\epsilon \to 0} V(\i, \j)/V(\tilde{\i}, \tilde{\j}) \to 0.
		\] 
	\end{lemma}
	\begin{proof}[Proof of Lemma S2.6]
		Set $(\tilde{\bm{i}}, \tilde{\bm{j}})$ equal to $(\bm{i}, \bm{j})$, apart from $\tilde{j}_h^r = 0$ and $\tilde{j}_{k+1}^r = m_r$.
	\end{proof}
	
	\bigskip
	
	\noindent The three lemmas imply that each relevant $(\bm{i}, \bm{j})$ corresponds to an admissible latent structure, i.e.
	\[
	\begin{aligned}
		\rho({u}) \sim \sum_{\p \in \mathcal{P}}(-1)^{n+m}e^{-\psi_b(\underline{u})}&\prod_{(i,j) \in \Delta_\p} \left\{\frac{\partial^{n_i+m_j}}{\partial u_1^{n_i}\partial u_2^{m_j}}\int_{A_i \times B_j} \int_{\R_+^2}(1-e^{-u_1s_1-u_2s_2})\, \rho(\d s)G_0(\d x)\right\}\times\\
		&\times\prod_{(i,j) \in \Delta^1_\p} \left\{\frac{\partial^{n_i}}{\partial u_1^{n_i}}\int_{A_i \times B_{c+1}} \int_{\R_+^2}(1-e^{-u_1s_1-u_2s_2})\, \rho(\d s)G_0(\d x)\right\}\times\\
		&\times\prod_{(i,j) \in \Delta^2_\p} \left\{\frac{\partial^{m_j}}{\partial u_2^{m_j}}\int_{A_{k+1} \times B_{j}} \int_{\R_+^2}(1-e^{-u_1s_1-u_2s_2})\, \rho(\d s)G_0(\d x)\right\}.
	\end{aligned}
	\]
	Evaluating the derivatives we have
	\[
	\begin{aligned}
		\rho({\underline{u}}) \sim \sum_{\p \in \mathcal{P}}e^{-\psi_b({\underline{u}})}&\prod_{(i,j) \in \Delta_\p} \left\{\int_{A_i \times B_j}\int_{\R_+^2}e^{-u_1s_1-u_2s_2}s_1^{n_i}s_2^{m_j}\, \rho(\d s)G_0(\d x)\right\}\times\\
		&\times\prod_{(i,j) \in \Delta^1_\p} \left\{\int_{A_i \times B_{c+1}}\int_{\R_+^2}e^{-u_1s_1-u_2s_2}s_1^{n_i}\, \rho(\d s)G_0(\d x)\right\}\times\\
		&\times\prod_{(i,j) \in \Delta^2_\p} \left\{\int_{A_{k+1} \times B_{j}}\int_{\R_+^2}e^{-u_1s_1-u_2s_2}s_2^{m_j}\, \rho(\d s)G_0(\d x)\right\}.
	\end{aligned}
	\]
	Finally, we get
	\[
	\begin{aligned}
		\mathcal{I} \sim \sum_{\p \in \mathcal{P}}\int_{\R_+^2}\frac{u_1^{n-1}u_2^{m-1}}{\Gamma(n)\Gamma(m)}e^{-\psi_b({\underline{u}})}&\prod_{(i,j) \in \Delta_\p} \left\{\int_{A_i \times B_j}\int_{\R_+^2}e^{-u_1s_1-u_2s_2}s_1^{n_i}s_2^{m_j}\, \rho(\d s)G_0(\d x)\right\}\times\\
		&\times\prod_{(i,j) \in \Delta^1_\p} \left\{\int_{A_i \times B_{c+1}}\int_{\R_+^2}e^{-u_1s_1-u_2s_2}s_1^{n_i}\, \rho(\d s)G_0(\d x)\right\}\times\\
		&\times\prod_{(i,j) \in \Delta^2_\p} \left\{\int_{A_{k+1} \times B_{j}}\int_{\R_+^2}e^{-u_1s_1-u_2s_2}s_2^{m_j}\, \rho(\d s)G_0(\d x)\right\}\, \d {u}.
	\end{aligned}
	\]
	Evaluating the numerator of \eqref{to_compute} the same reasoning yields a formula asymptotic to
	\[
	\begin{aligned}
		\sum_{\p \in \mathcal{P}}\int_{\R_+^2}\frac{u_1^{n-1}u_2^{m-1}}{\Gamma(n)\Gamma(m)}e^{-\psi_h({\underline{u}})}&\prod_{(i,j) \in \Delta_\p} \left\{\int_{A_i \times B_j}\int_{\R_+^2}e^{-(h_1(x)+u_1)s_1-(h_2(x)+u_2)s_2}s_1^{n_i}s_2^{m_j}\, \rho(\d s)G_0(\d x)\right\}\\
		&\prod_{(i,j) \in \Delta^1_\p} \left\{\int_{A_i \times B_{c+1}}\int_{\R_+^2}e^{-(h_1(x)+u_1)s_1-(h_2(x)+u_2)s_2}s_1^{n_i}\, \rho(\d s)G_0(\d x)\right\}\\
		&\prod_{(i,j) \in \Delta^2_\p} \left\{\int_{A_{k+1} \times B_{j}}\int_{\R_+^2}e^{-(h_1(x)+u_1)s_1-(h_2(x)+u_2)s_2}s_2^{m_j}\, \rho(\d s)G_0(\d x)\right\}\, \d {u}.
	\end{aligned}
	\]
	where $\psi_h({\underline{u}}) = \int_{\X^2}\int_{\R_+^2}\left(1-e^{-(h_1(x)+u_1)s_1-(h_2(x)+u_2)s_2}\right)\, \rho(\d s) G_0(\d x)$. Note that
	\[
	\begin{aligned}
		1-e^{-(h_1(x)+u_1)s_1-(h_2(x)+u_2)s_2} &= e^{-u_1s_1-u_2s_2}\left[e^{u_1s_1+u_2s_2}-1+1-e^{-h_1(x)s_1-h_2(x)s_2} \right]\\
		&=\left[1-e^{-u_1s_1-u_2s_2} \right]+\left[1-e^{-h_1(x)s_1-h_2(x)s_2} \right],
	\end{aligned}
	\]
	so that
	\[
	\begin{aligned}
		e^{-\psi_h({\underline{u}})} &= e^{-\psi_b({\underline{u}})}e^{-\int_{\X^2}\int_{\R_+^2}\left[1-e^{-h_1(x)s_1-h_2(x)s_2} \right]\rho(\d s) G_0(\d x)}\\
		& = e^{-\psi_b({\underline{u}})}\E\left[e^{-\int_{\X^2}h_1(x) \, \hat{\mu}_1(\d x) - \int_{\X^2}h_2(x) \, \hat{\mu}_2(\d x)} \right].
	\end{aligned}
	\]
	Furthermore
	\[
	G_0(A_h \times B_r) = \epsilon \frac{G_0(A_h \times B_r)}{\epsilon} \sim \epsilon g_{h,r}, \quad 1 \leq i \leq c, 1 \leq j \leq k,
	\]
	and
	\[
	G_0(A_h \times \d x) \sim \epsilon g_{h,c+1}Q_{X_h^*}( \d x), \quad G_0(\d x \times  B_r) \sim \epsilon g_{k+1, r}P_{Y_r^*}( \d x).
	\]
	Thus, evaluating the limit in \eqref{to_compute} we get
	\[
	\begin{aligned}
		\E \biggl[&e^{-\int_{\X^2}h_1(x)\, \mu_1(\d x)-\int_{\X^2}h_2(x)\, \mu_2(\d x)} \mid (X_i)_{i\geq1}^{n}, (Y_j)_{j\geq1}^{m}  \biggr]  =\\
		&\times \sum_{p \in \mathcal{P}}\int_{\R_+^2}\E\left[e^{-\int_{\X^2}h_1(x) \, \hat{\mu}_1(\d x) - \int_{\X^2}h_2(x) \, \hat{\mu}_2(\d x)} \right]\times\\
		&\times \prod_{(i,j) \in \Delta_p}\int_{\R_+^2}e^{-h_1(X_i^*, Y_j^*)s_1-h_2(X_i^*, Y_j^*)s_2}\frac{s_1^{n_i}s_2^{m_j}e^{-u_1s_1-u_2s_2}\rho(\d s)}{\tau_{n_i, m_j(\underline{u})}}\times\\
		&\times \prod_{(i,j) \in \Delta^1_p}\int_{\X}\int_{\R_+^2}e^{-h_1(X_i^*, x_2)s_1-h_2(X_i^*,x_2)s_2}\frac{s_1^{n_i}e^{-u_1s_1-u_2s_2}\rho(\d s)}{\tau_{n_i,0(\underline{u})}}Q_{X_i^*}( \d x_2)\times\\
		&\times \prod_{(i,j) \in \Delta^2_p}\int_{\X}\int_{\R_+^2}e^{-h_1(x_1, Y_2^*)s_1-h_2(x_1,Y_2^*)s_2}\frac{s_2^{m_j}e^{-u_1s_1-u_2s_2}\rho(\d s)}{\tau_{0, m_j(\underline{u})}}P_{Y_j^*}(\d x_1)\times\\
		&\times \left( \frac{\int_{\R_+^2}u_1^{n-1}u_2^{m-1}\prod_{(i,j) \in p}g_{i,j}\tau_{n_i, m_j}(\underline{u})e^{-\psi_b(\underline{u})}\, \d \underline{u}}{\sum_{\q \in \mathcal{P}}\int_{\R_+^2}u_1^{n-1}u_2^{m-1}\prod_{(i,j) \in \q}g_{i,j}\tau_{n_i, m_j}(\underline{u})e^{-\psi_b(\underline{u})}\d \underline{u}}\right) \times\\
		&\times \frac{u_1^{n-1}u_2^{m-1}\prod_{(i,j) \in p}\tau_{n_i, m_j}(\underline{u})e^{-\psi_b(\underline{u})}\, \d u}{\int_{\R_+^2}u_1^{n-1}u_2^{m-1}\prod_{(i,j) \in p}\tau_{n_i, m_j}(\underline{u})e^{-\psi_b(\underline{u})} \, \d \underline{u}},
	\end{aligned}
	\]
	as desired.
\end{proof}

\bigskip

\begin{proof}[Proof of Corollary 2]
	We use the shorthand notation $\mu_1(f) = \int_{\X}f(x) \, \mu_1(\d x)$ for any measurable function $f \, : \, \X \, \to \, \R$ such that $\mu_1(|f|) < \infty$. Letting ${U}$ be the set of latent variables of Theorem 1, i.e. ${U} = \left(p, U_1, U_2, Z^x, Z^y \right)$ for any $y_1, \dots, y_n \in (0,1)$ and $A_1, \dots, A_n \in \mathcal{X}^2$ we get
	\[
	\begin{aligned}
		&\P \left[p_1(A_1) \leq y_1, \dots, p_n(A_n) \leq y_n \mid {U},  (X_{i})_{i=1}^{n},  (Y_{j})_{j=1}^{m} \right] \\
		=&\P \left[\mu_1(\mathbbm{1}_{A_1}-y_1) \leq 0, \dots, \mu_n(\mathbbm{1}_{A_n}-y_n) \leq 0 \mid {U},(X_{i})_{i=1}^{n},  (Y_{j})_{j=1}^{m} \right].
	\end{aligned}
	\]
	The result follows since the finite dimensional distributions of $p_1$ given ${U}$, $(X_{i})_{i=1}^{n}$, and $(Y_{j})_{j=1}^{m}$ coincide with the ones of the normalized posterior distribution of $\mu_1$, given ${U}$, $(X_{i})_{i=1}^{n}$, and $(Y_{j})_{j=1}^{m}$.
\end{proof}

\bigskip

\begin{proof}[Proof of Theorem 2]
	Set ${\underline{H}} = \left(p, U_1, U_2 \right)$ with domain $D$. Then
	\[
	\begin{aligned}
		\P(X_{n+1} \in \d x \mid  (X_{i})_{i=1}^{n},  (Y_{j})_{j=1}^{m} ) 
		&=
		\E[\tilde{p}_1(\d x) \mid  (X_{i})_{i=1}^{n},  (Y_{j})_{j=1}^{m} ] \\
		&=  \int_D\E[\tilde{p}_1(\d x) \mid \underline{H} = \underline{h},  (X_{i})_{i=1}^{n}, (Y_{j})_{j=1}^{m}] \, F(\d \underline{v}),
	\end{aligned}
	\]
	where $F(\cdot)$ is the posterior distribution of $\underline{H}$, with $\underline{h} = (p, u_1, u_2)$. Recalling the notation in Corollary $2$ we have
	\[
	\begin{aligned}
		\E[\tilde{p}_1(\d x) \mid& \underline{H} = \underline{h}, (X_{i})_{i=1}^{n}, (Y_{j})_{j=1}^{m} ] =\E \left[\frac{\hat{\mu}_1(\d x \times \X)}{R} \right] + \E \left[\frac{\sum_{(i,j) \in \Delta_{p}}J^1_{i,j}\delta_{X_i^*}}{R} \right]+\\
		&+ \E \left[\frac{\sum_{(i,j) \in \Delta^1_{p}}J^1_{i,c+1}\delta_{X_i^*}}{R} \right]+\E \left[\frac{\sum_{(i,j) \in \Delta^2_{p}}J^1_{k+1,j}\delta_{Z_j^y}}{R} \right]= \sum_{k = 1}^4I_k,
	\end{aligned}
	\]
	where $R = T_1+\sum_{(i,j) \in \Delta_{p}}J_{i,j}^1+\sum_{(i,j) \in \Delta^1_{p}}J_{i,c+1}^1+\sum_{(i,j) \in \Delta^2_{p}}J_{k+1,j}^1$. \\ Set $S = \sum_{(i,j) \in \Delta_{p}}J_{i,j}^1+\sum_{(i,j) \in \Delta^1_{p}}J_{i,c+1}^1+\sum_{(i,j) \in \Delta^2_{p}}J_{k+1,j}^1$ and exploit the conditional independence between $J_{ij}^1$ and $\hat{\mu}_1$ to obtain
	\[
	\begin{aligned}
		I_1 &= \int_{\R_+}\E \left[e^{-vS} \right]\E\left[\hat{\mu}_1(\d x \times \X)e^{-vT_1}\right] \, \d v \\
		&=\theta P_0(\d x) \int_{\R_+} \left(\prod_{(i,j) \in p}\frac{\tau_{n_i, m_j}(u_1+v,u_2)}{\tau_{n_i, m_j}(u_1,u_2)}\right)\tau_{1,0}(u_1+v, u_2)e^{-\psi^{\underline{u}}_b(v,0)} \, \d v,
	\end{aligned}
	\]
	where $\psi_b^{\underline{u}}(\lambda_1,\lambda_2)$ is the Laplace exponent of $(\hat{\mu}_1, \hat{\mu}_2)$ in Theorem 1. Observing that $\psi_b^{\underline{u}}(v,0)+\psi(u_1,u_2) = \psi(u_1+v,u_2)$ and denoting with $L(\cdot)$ the distribution of $\p$, we obtain
	\[
	\begin{aligned}
		\xi_0 &= \int_D I_1 \, F(\d \underline{u})\\
		& = \theta P_0(\d x)\int\int_{\R_+^3}\biggl\{ u_1^{n-1}u_2^{m-1}\left(\prod_{(i,j) \in p}\tau_{n_i, m_j}(u_1+v,u_2)\right)\tau_{1,0}(u_1+v, u_2)\times\\
		&\times e^{-\psi(u_1+v, u_2)}\, \d u_1 \d u_2 \d v L(\d p)\biggr\}\\
		& = \frac{\theta P_0(\d x)}{n}\int\int_{\R_+^2} u_1^{n}u_2^{m-1}\left(\prod_{(i,j) \in p}\tau_{n_i, m_j}(u_1,u_2)\right)\tau_{1,0}(u_1, u_2)e^{-\psi(u_1, u_2)}\, \d u_1 \d u_2  L(\d p)\\
		& = \frac{\theta P_0(\d x)}{n}\int_Du_1\tau_{1,0}(u_1, u_2) \, F(\d \underline{u}),
	\end{aligned}
	\]
	where the second equality follows from the change of variables $(w,z) = (u_1+v,u_1)$. The proof for the remaining weights follows along the same lines and leads to
	\[
	\xi_i^x = \frac{1}{n}\int_Du_1\left[ \frac{\tau_{n_i+1, m_j}(u_1, u_2)}{\tau_{n_i, m_j}(u_1, u_2)} +  \frac{\tau_{n_i+1, 0}(u_1, u_2)}{\tau_{n_i, 0}(u_1, u_2)}\right] \, F(\d \underline{u})
	\]
	and
	\[
	\xi_i^y = \frac{1}{n}\int_Du_1 \frac{\tau_{1, m_j}(u_1, u_2)}{\tau_{0, m_j}(u_1, u_2)} \, F(\d \underline{u}).
	\]
	The weights for $Y_{m+1}$ can be computed in an analogous fashion.
\end{proof}

\section{A toy example of borrowing of information}
Classical borrowing of information across samples is typically associated to positive correlation across observations in different populations and, as a consequence, it induces shrinkage of the predictions. Let us consider the toy situation in which observations coming from two different populations have been collected and a normal model is assumed
\begin{align*}
	X_i\mid \mu_x&\overset{iid}{\sim}\text{N}(\mu_x,\,1)\qquad \text{for } i=1,\ldots,n\\
	Y_j\mid \mu_y&\overset{iid}{\sim}\text{N}(\mu_y,\,1)\qquad \text{for } j=1,\ldots,m
\end{align*}
To obtain a working model, one has to specify a certain prior over $\mu_x$ and $\mu_y$. The main typical strategies one may employ are the following:
\begin{itemize}
	\item Modeling $\mu_x$ and $\mu_y$ as independent, which ultimately means that we do not consider the information coming from one population to be relevant for inference on the other. 
	
	\item Modeling $\mu_x$ and $\mu_y$ as dependent, which induces borrowing of information. This typically reflects the idea that, if the observed values of $Y_1,\ldots,Y_m$ are on average higher than our prior guess on $\mu_y$, then we should upwards revise our belief on $\mu_x$ and our prediction for $X_1$. 
\end{itemize} 
To clarify this last point, we compare a typical strategy used to perform borrowing of information, which is provided by the following hierarchy
\begin{equation}\label{te:borrow}
	\begin{split}
		\mu_x\mid\mu_0&\sim \text{N}(\mu_0,\,1)\\
		\mu_y\mid\mu_0&\sim \text{N}(\mu_0,\,1)\\
		\mu_0&\sim \text{N}(\nu,\,1)
	\end{split}
\end{equation}
with the case of independent priors, namely
\begin{equation}\label{te:indip}
	\begin{split}
		\mu_x\sim \text{N}(\nu,\,2)&\quad\quad\quad
		\mu_y\sim \text{N}(\nu,\,2)\\
		&\mu_x\perp \mu_y
	\end{split}
\end{equation}
where the variance is chosen to match the marginal distributions of the hierarchical specification. We assume that only the sample $(Y_1,\ldots,Y_m)$ has been observed and we discuss its impact on the posterior distribution of $\mu_x$ and on the predictive distribution of $X_1$ under the two specifications. 
Under independence in \eqref{te:indip}, one obviously has 
\[
p(\mu_x\mid(Y_{j})_{j=1}^{m})= 
\text{N}(\nu,\,2)
\]
while under model~\eqref{te:borrow} the new distribution of $\mu_x$ is 
\[
\begin{aligned}
	p(\mu_x\mid (Y_{j})_{j=1}^{m})&\propto\int_{\mathbb{R}} p(\mu_x\mid\mu_0) \, p(\mu_0\mid (Y_{j})_{j=1}^{m}) \d \mu_0 \\ 
	&\:=\:\text{N}\left(\frac{1}{2m+1}\nu + \frac{2m}{2m+1} \frac{\nu+\bar y }{2},\,1 + \frac{m+1}{2m+1}\right),
\end{aligned}
\]
where $\bar y$ denotes the empirical average of $Y_1,\ldots,Y_m$, and 
\[
\mathbb{E}[X_1\mid (Y_{j})_{j=1}^{m}] = \mathbb{E}[\mu_x\mid (Y_{j})_{j=1}^{m}] = \nu + \frac{m}{2m + 1}(\bar y - \nu)
\]
Therefore, when $\bar y > \nu$ the borrowing results in an increase of the estimate for $\mu_x$ and of the prediction for $X_1$, while if $\bar y < \nu$ the borrowing of information induces the opposite effect. The shrinking behaviour is ultimately a consequence of the fact that the hierarchical prior in \eqref{te:borrow} induces positive correlation across $X_i$ and $Y_j$. However, what we show in the main paper is that classical shrinkage of the estimates is not the only way to borrow information within partially exchangeable populations, neither necessarily the best one. 

\section{Example of correlation between FuRBI priors on Borel set}

\textcolor{black}{Consider a pair of n-FuRBI priors with equal jumps (see Example $4$ in the main document), where the baseline distribution $G_0$ is given by a bivariate normal with zero mean, unit variances and correlation $\rho \in \{-0.99, -0.5, 0, 0.5, 0.99\}$. In Figure \ref{corr_FURBI} we depict the correlations on sets of the form $(-\infty, x]$, with $x \in [-5, 5]$ and for each value of the correlation. Notice that such correlation may be of particular interest in survival settings, where the distribution function is often the main focus.}

\textcolor{black}{When $\rho = 0$, the correlation is equal to $0$ as expected, since $G_0(A \times A)= P_0(A)^2$ and the numerator of the formula in Proposition $8$ vanishes. For values of $\rho$ different from $0$, the correlation is symmetric around $0$, due to the symmetry of the Gaussian distribution, and different signs indicate opposite behaviours: therefore, $\rho < 0$ implies negative correlation on such Borel sets. }

\textcolor{black}{However, note that a different sign does not mean a completely specular behaviour: for instance the correlation with $\rho = 0.99$ is higher in absolute value than the one with $\rho = -0.99$. This is due to the fact that it is somewhat impossible to have strictly negative correlation on all Borel sets. Intuitively, if the two priors have high negative correlation on $(-\infty, 0]$, it means that one of them has much larger mass on $(-\infty, 0]$ and the other on $(0, +\infty)$: therefore, both priors will have a high mass on $(-\infty, a]$, with $a$ large positive number, so that the correlation can not attain again large negative values.}

\textcolor{black}{Finally, if $\rho \to 1$, then the correlation converges to the constant function $1$, that is the value obtained with equal atoms: indeed, the two priors will have equal jumps and linearly dependent atoms (see Corollary 1).}

\begin{figure}\label{fig_cor}
	\vspace{-0.5\baselineskip}
	\floatbox[{\capbeside\thisfloatsetup{capbesideposition={right,center},capbesidewidth=6cm}}]{figure}[\FBwidth]
	{\caption{\textcolor{black}{Correlation on Borel sets of the form $(-\infty, x]$, with $x \in [-5,5]$. The four lines, from bottom to top, correspond to $\rho \in \{-0.99, -0.5, 0, 0.5, 0.99\}$.}}\label{corr_FURBI}}
	{\includegraphics[width=7cm]{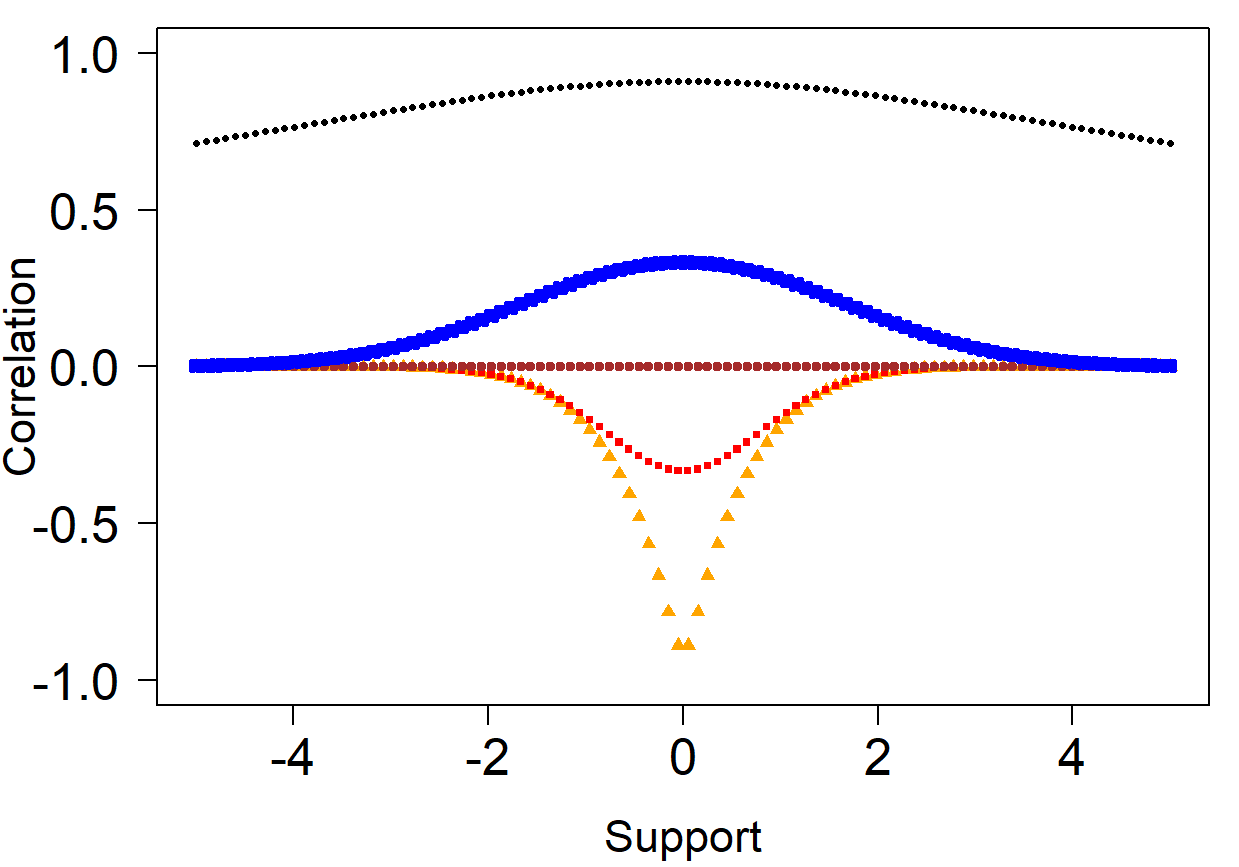}}
\end{figure}

\section{\textcolor{black}{Algorithms for posterior inference}}

\textcolor{black}{In this section we address the issue of sampling from the posterior distribution. In discrete nonparametric models, we need to distinguish whether the random probability measures are directly applied to the data or rather convoluted with a suitable kernel (known as \emph{mixture} model, see Section $6$ in the paper). }

\textcolor{black}{Nevertheless, from a computational perspective, if the first problem is solved the second one can be tackled in a similar way: it is indeed easy to propose a Gibbs sampler that alternates sampling of suitable latent variables and the posterior distribution given data originated by the random probability measure (see Section S4.4 below for how to extend algorithms to mixture models). }

\textcolor{black}{Therefore, in the following three sections, we assume to collect observations from 
	\begin{equation}\label{partial_model}
		\begin{aligned}
			(X_i,Y_j)\mid(\tilde p_1,\tilde p_2)  \overset{iid}{\sim} \tilde p_1\times\tilde p_2\qquad
			(\tilde p_1,\,\tilde p_2) \sim Q
		\end{aligned}
	\end{equation}
}
\subsection{\textcolor{black}{Marginal posterior samplers}}
\textcolor{black}{The first approach is to  directly simulate the trajectories of $(\tilde p_1,\tilde p_2)$ from its posterior, giving rise to so--called conditional algorithms. See, e.g, \cite{ishwaran2001gibbs,walker07,papa08,arbel2017moment}. Conditional samplers for the n-FuRBI priors can be found in Sections S3.2-3 below.}

\textcolor{black}{Alternatively, and this is the route followed in this section, one can use marginal algorithms, that integrate out the random measures and sample sequentially from the predictive distributions \citep[see, for instance,][]{neal2000markov}.}

\textcolor{black}{Given $(X_{i})_{i=1}^{n}$ and $(Y_{j})_{j=1}^{m}$ and using the results in Theorem~2, we can sample iteratively new observations from $\tilde p_1$ as follows}

\textcolor{black}{\begin{algorithm}
		Marginal algorithm - 1
		\begin{tabbing}
			\qquad (a) \enspace Compute weights $\xi_0$, $\{\xi_i^x\}$ and $\{\xi_j^y\}$ from $(X_{i})_{i=1}^{n}$ and $(Y_{j})_{j=1}^{m}$\\
			\qquad (b) \enspace Draw $X_{n+1}$ from
			$m(\d x) = \xi_0P_0(\d x) + \sum_{i= 1}^k\xi_i^x \delta_{X_i^{*}}(\d x)+  \sum_{j = 1}^c\xi_j^yP_{Y_j^*}\left(\d x\right)$
		\end{tabbing}
\end{algorithm}}

\textcolor{black}{The algorithm is straightforward, but relies on the computation of the weights at point (a): this is not optimal, since in general the explicit evaluation can be demanding. Nonetheless, Theorem~1 and Corollary~2 show that, conditionally on a suitable set of latent variables, the posterior representation simplifies greatly. Indeed, given $((X_{i})_{i=1}^{n}, (Y_{j})_{j=1}^{m}, U_1, U_2, {p} )$, the predictive distribution of the first sample is
	\begin{equation}\label{marginal}
		\begin{aligned}
			m(\d x) &\propto  \theta \tau_{1,0}(U_1, U_2) P_0(\d x) +\sum_{(i,j) \in \Delta_p}\frac{\tau_{n_i+1,m_j}(U_1, U_2)}{\tau_{n_i,m_j}(U_1, U_2)}\delta_{X_i^*}(\d x) \\
			&+ \sum_{(i,j) \in \Delta^1_p}\frac{\tau_{n_i+1, 0}(U_1, U_2)}{\tau_{n_i,0}(U_1, U_2)}\delta_{X_i^*}(\d x) + \sum_{(i,j) \in \Delta^2_p}\frac{\tau_{1, m_j}(U_1, U_2)}{\tau_{0,m_j}(U_1, U_2)}P_{Y_j^*}(\d x).
		\end{aligned}
	\end{equation}
	Those new weights, whose derivation can be found in Section S1.4, are 
	easier to compute, as the next example shows.}

\textcolor{black}{\begin{example}[Inverse Gaussian n-FuRBI with equal jumps]{\rm 
			For this case we obtain $\tau_{n,m}(u_1, u_2) = \int_{\R}s^{n+m}e^{-(u_1+u_2)s} \rho(\d s):=\tau_{n+m}(u_1+u_2)$, where $\rho(\d s)$ is the common marginal jump intensity. If the L\'evy intensity is
			$v(\d s, \d x) = {e^{-{s}/{2}}}/({s^{{3}/{2}}}\sqrt{2\pi})\d s \, \alpha(\d x)$
			the resulting normalized CRM corresponds to the normalized inverse Gaussian process introduced in
			\cite{lijoi2005hierarchical}. We then obtain
			$\tau_j(u) = {2^{j-1}\Gamma \left(j-{1}/{2} \right)}/({\sqrt{\pi}(2u+1)^{j-{1}/{2}}})$,
			where $u = u_1+u_2$. Thus, conditionally on the 
			latent variables, we have
			\[
			\begin{aligned}
				m(\d x)& \propto  \theta  P_0(\d x)+\frac{2}{\sqrt{2U+1}}\sum_{(i,j) \in \Delta_p}\left(n_i+m_j-\frac{1}{2}\right)\delta_{X_i^*}(\d x) \\
				&+ \frac{2}{\sqrt{2U+1}}\sum_{(i,j) \in \Delta^1_p}\left(n_i-\frac{1}{2}\right)\delta_{X_i^*}(\d x) + \frac{2}{\sqrt{2U+1}}\sum_{(i,j) \in \Delta^2_p}\left(m_j-\frac{1}{2}\right)P_{Y_j^*}(\d x),
			\end{aligned}
			\]
			where $U = U_1+U_2$. Sampling from this mixture is straightforward.}
\end{example}}
\textcolor{black}{Thus we can derive a second marginal algorithm.}

\textcolor{black}{\begin{algorithm}
		Marginal algorithm - 2
		\begin{tabbing}
			\qquad (a)\enspace Draw $\left(U_1, U_2, {p} \right)$ from their conditional distributions specified in Section~$5$\\
			\qquad (b)\enspace Draw $X_{n+1}$ from $m(\d x)$ in \eqref{marginal}
		\end{tabbing}
\end{algorithm}}

\textcolor{black}{However, even the full conditional distribution of $p$ may not always be available in closed form, and it may be computationally intensive to evaluate, since it may have a very large support. When this is the case, we may encode the latent clustering structure in a more convenient way introducing two arrays of latent variables $\mathcal{C}_x=(c_{i,x})_{i\geq1}$ and $\mathcal{C}_y=(c_{j,y})_{j\geq1}$ such that $c_{i,x}=c_{i',x}$ denotes a tie between $X_{i}$ and $X_{i'}$, $c_{j,y}=c_{j',y}$ denotes a tie between $Y_{j}$ and $Y_{j'}$, while $c_{i,x}=c_{j,y}$ denotes a hyper-tie between $X_{i}$ and $Y_{j}$. Moreover, we reorder the unique values in $\underline{X}^*_{n}$ and $\underline{Y}^*_{m}$, so that
	$X^*_{c}=X_{i}$ if and only if $c_{i,x}=c$ and $Y^*_{c}=Y_{j}$ if and only if $c_{j,y}=c$.
	Therefore, $\mathbbm{P}[c_{n+1,x}=c \mid\mathcal{C}_x,\mathcal{C}_y,\underline{X}^{*}_{n}, \underline{Y}^{*}_{m}]$ is
	\[
	\begin{cases}
		\mathbb{P}[X_{n+1}=X^*_{c}\mid\mathcal{C}_x,\mathcal{C}_y,\underline{X}_n^{*}, \underline{Y}_m^{*}] , & \text{for } c\in{\mathcal{C}_x}\\
		\displaystyle\int\mathbb{P}[X_{n+1}=x\mid \mathcal{C}_y,\underline{Y}_m^{*}]\,p_{Y^*_{c}}(x) \d x, & \text{for }  c\in{\mathcal{C}_{y}} \setminus {\mathcal{C}_x}\\
		\displaystyle\int\mathbb{P}[X_{n+1}=x]\, p_{0}(x) \d x, & \text{otherwise}\\
	\end{cases}
	\]
	Finally, the distribution of ${p}$, given $\mathcal{C}_x$ and $\mathcal{C}_y$, is degenerate. Moreover, the posterior distribution of $(U_1,U_2)$ given ${p}$ is equal to the posterior distribution of $(U_1,U_2)$ given $\mathcal{C}_x$ and $\mathcal{C}_y$. 
	Therefore, we may build a marginal algorithm sampling $\mathcal{C}_x$ and $\mathcal{C}_y$ instead of  ${p}$, without modifying the full conditional distribution for $U_1$ and $U_2$.
	The final marginal algorithm boils down to}

\textcolor{black}{\begin{algorithm}
		Marginal algorithm - 3
		\begin{tabbing}
			\qquad (a)\enspace Draw $(U_1, U_2)$ and $c_{n+1,x}$\\
			\qquad (b)\enspace Sample $X_{n+1}$ from 
			$m(\d x) = \begin{cases}
				\delta_{X_{c_{n+1,x}}^*}(\d x), & \text{if } c_{n+1,x}\in{\mathcal{C}_x}\\
				P_{Y_{c_{n+1,x}}^*}(\d x), & \text{if }  c_{n+1,x}\in{\mathcal{C}_{y}} \setminus {\mathcal{C}_x}\\
				P_{0}(\d x), & \text{otherwise}\\
			\end{cases}$
		\end{tabbing}
\end{algorithm}}

\textcolor{black}{The advantage of such approach is twofold. First, we do not need to sample directly the full conditional distribution of $p$. Second, when the algorithm is applied to mixture models, as in section $6$, sampling the unique values, instead of single observations, improves the mixing of the algorithm (cfr. Neal, 2000).}

\subsection{\textcolor{black}{Conditional posterior sampler based on the law of the CRV}}
\textcolor{black}{To develop a conditional algorithm, we can sample from the distribution of $(\tilde\mu_1,\tilde\mu_2)$ and then normalize each draw to get an approximate realization of the random probabilities. Here we develop a general conditional sampler based on this approach that can be tailored to specific choices of the intensity in the prior.}

\textcolor{black}{By Theorem 1, we know that a posteriori  ${\mu}=(\mu_1,\mu_2)$ is the sum of two components, that we call ${\mu}_{obs}$ and $\hat{{\mu}}$ and are such that
	\[
	{\mu}_{obs} = \sum_{(i,j) \in \Delta_p}{J}_{i,j}\delta_{\left( X_i^*, Y_j^* \right)}+ \sum_{(i,j) \in \Delta^1_p}{J}_{i,c+1}\delta_{\left( X_i^*, Z_i^x \right)}+\sum_{(i,j) \in \Delta^2_p}{J}_{k+1,j}\delta_{\left( Z_j^y, Y_j^* \right)}.
	\]
	where ${J}_{i,j}=(J^1_{i,j},J^2_{i,j})$, and
	\[
	\hat{{\mu}}=\left(\sum\limits_{h=1}^{+\infty} S^1_h \delta_{(V_{h},W_h)},
	\sum\limits_{h=1}^{+\infty} S^2_h \delta_{(V_{h},W_h)}\right)
	\]
	is a CRV with L\'evy intensity $e^{-U_1s_1-U_2s_2}\rho(\d s_1, \d s_2)G_0(\d x)$. 
	Denote the marginal and joint tail integrals of $\hat{{\mu}}$ as 
	\[
	N_1(s)=\int\limits_s^{+\infty}\int\limits_0^{+\infty}e^{-U_1s_1-U_2s_2}\rho(\d u_1, \d u_2), \quad N_2(s)=\int\limits_0^{+\infty}\int\limits_s^{+\infty}e^{-U_1s_1-U_2s_2}\rho(\d u_1, \d u_2)
	\]
	and 
	\[
	N(s_1,s_2)=\int\limits_{s_1}^{+\infty}\int\limits_{s_2}^{+\infty}e^{-U_1s_1-U_2s_2}\rho(\d u_1, \d u_2).
	\]
	Lastly, define the correspondent L\'evy copula as $ F(x,y)=N(N_1^{-1}(x),N_2^{-1}(y))$. If $F(x,y)$ is continuous on $[0,+\infty]^2$, the iterative conditional sampler based on the Ferguson and Klass algorithm \citep{ferguson1972representation} reads
	\begin{enumerate}
		\item[(a)] Generate $\mu_{obs}$ as follows
		\begin{enumerate}
			\item[(a1)] Generate $\left(U_1, U_2, \p\right)$ from the distributions specified in Section $5$;
			\item[(a2)] Generate ${J}_{i,j}=(J^1_{i,j},J^2_{i,j})$ from the distributions specified in Theorem 1;
			\item[(a3)] Generate $Z^x_i$ and $Z^y_j$ from the distributions specified in Section $5$.
		\end{enumerate}
		\item[(b)] Generate an approximation of $\hat\mu$, given by $\left(\sum\limits_{h=1}^{M} S^1_h \delta_{(V_{h},W_h)},
		\sum\limits_{h=1}^{M} S^2_h \delta_{(V_{h},W_h)}\right)$ as follows
		\begin{enumerate}
			\item[(b1)] Generate $\xi^x_1,\ldots,\xi^x_M$ from a Poisson Process with unit rate;
			\item[(b2)] Generate $\xi^y_1,\ldots,\xi^y_M$ from $\xi^y_h\sim\frac{\partial}{\partial x}F(x,\xi)\Biggr|_{x=\xi^x_h}$
			\item[(b3)] Determine $(S^1_{h},S^2_h)$ solving 
			\[
			\xi^x_h=N_1(S^1_{h}) \qquad \xi^y_h=N_2(S^2_{h})
			\]
			\item[(b4)] Generate $(V_{h},W_h)$ from $G_0$.
		\end{enumerate}
		\item[(c)] Obtain a draw from $\tilde p_1$ as follows
		\[
		\tilde p_1\approx\frac{\sum\limits_{h=1}^{M} S^1_{h} \delta_{V_{h}}+
			\sum_{(i,j) \in \Delta_p}J^1_{i,j}\delta_{X_i^*}+ 
			\sum_{(i,j) \in \Delta^1_p}J^1_{i,c+1}\delta_{X_i^*}+
			\sum_{(i,j) \in \Delta^2_p}J^1_{k+1,j}\delta_{Z_j^y}}
		{\sum\limits_{h=1}^{M} S^1_{h}+
			\sum_{(i,j) \in \Delta_p}J^1_{i,j}+ 
			\sum_{(i,j) \in \Delta^1_p}J^1_{i,c+1}+
			\sum_{(i,j) \in \Delta^2_p}J^1_{k+1,j}}.
		\]
		An analogous approximation can be computed for $\tilde p_2$. 
\end{enumerate}}
\subsection{\textcolor{black}{Conditional posterior sampler for gamma process with equal jumps}}\label{Gamma_MCMC}
\textcolor{black}{Alternatively, a second strategy for conditional algorithms is to sample approximate draws from the posterior distribution of the random probabilities $(\tilde p_1,\tilde p_2)$. We provide an example for gamma FuRBI CRMs with equal jumps.}

\textcolor{black}{In the case of a process with equal jumps, we know from the definition that the measures in the product space are $p_1=p_2=p$. Therefore, posterior inference can be conducted without loss of generality on
	\[
	p = \sum_{k \geq 1}\bar{W}_k \delta_{(\theta_k, \phi_k)}, \quad \text{with } (\theta_k, \phi_k) \simiid G_0(\cdot),
	\]
	where $\{\bar{W}_k\}_k$ are the weights of a Dirichlet process, which can defined through the popular stick-breaking construction \citep{sethuraman1994constructive}. In this context, \cite{ishwaran2001gibbs} developed a conditional algorithm for hierarchical mixture models, called \emph{blocked Gibbs} sampler, based on the approximation
	\[
	p \approx \sum_{k = 1}^N\bar{W}_k \delta_{(\theta_k, \phi_k)}, \quad \text{for large $N$}.
	\]
	Exploiting the appealing analytical properties of the Dirichlet process, it is possible to devise simple formulae for the posterior distribution of the $N$ jumps and $N$ locations: see Section $5$ of \cite{ishwaran2001gibbs} for more details.}

\subsection{\textcolor{black}{Sampling from mixture models using marginal algorithms}}
\textcolor{black}{Consider the mixture model defined in Section $6.1$. Starting from Algorithm $2$ in Section $S4.1$, we devise a Gibbs sampler for drawing from the posterior distribution of $(X_{i})_{i=1}^{n}$ and $(Y_{j})_{j=1}^{m}$.}

\textcolor{black}{Denoting by $\bm{X}^t = (X_1^t, \dots, X_n^t)$ and $\bm{Y}^t = (Y_1^t, \dots, Y_n^t)$ the vectors sampled at step $t$, the algorithm reads
	\begin{enumerate}
		\item Initialize at random $\bm{X}^0$ and $\bm{Y}^0$.
		\item For any $t \geq 1$ do:
		\begin{enumerate}
			\item[(b.1)] Draw $\left(U_1, U_2, \p \right)$ given $\bm{X}^{t-1}$ and $\bm{Y}^{t-1}$, from the distributions specified in Theorem $1$.
			\item[(b.2)]  Draw $\bm{X}_{n}$, given $\left(U_1, U_2, \p\right)$ as follows: for any $i$ sample $X_i^t$ from
			\begin{multline*}
				q(\d x \mid \bm{X}_{-i}^t) = q_{i,0}(U_1, U_2)P_0(\d x) + \sum_{(i,j) \in \Delta_{\bm{p}}} q_{i,j}(U_1, U_2) \delta_{X_i^{*}}\\
				+ \sum_{(i,j) \in \Delta^1_{\bm{p}}} q^1_{i,j}(U_1, U_2) \delta_{X_i^{*}}(\d x)+ \sum_{(i,j) \in \Delta_{\bm{p}}} q^2_{i,j}(U_1, U_2)P_{Y_j^*}(\d x),
			\end{multline*}
			where $\bm{X}_{-i}^t = \left(X_1^t, \dots, X_{i-1}^t, X_{i+1}^{t-1}, \dots X_n^{t-1}\right)$, with unique values $\left(X_1^*,\dots, X_k^*\right)$ and multiplicities $(n_1, \dots, n_k)$. Analogously, $\left(Y_1^*, \dots, Y_c^*\right)$ denotes the unique values in $\bm{Y}^{t-1}$ with multiplicities $(m_1, \dots, m_c)$. The mixing proportions are given by		
			\[
			\begin{aligned}
				&q_{i,0}(U_1, U_2) \propto \theta \tau_{1,0}(U_1, U_2)\int_{\X}f(W_i \mid x)P_0(\d x), \\
				& q_{i,j}(U_1, U_2) \propto \frac{\tau_{n_i+1, m_j}(U_1, U_2)}{\tau_{n_i, m_j}(U_1, U_2)}f(W_i \mid X_i^*),\\
				&q^1_{i,j}(U_1, U_2) \propto \frac{\tau_{n_i+1, 0}(U_1, U_2)}{\tau_{n_i, 0}(U_1, U_2)}f(W_i \mid X_i^*), \\
				& q^2_{i,j}(U_1, U_2) \propto \frac{\tau_{1, m_j}(U_1, U_2)}{\tau_{0, m_j}(U_1, U_2)}\int_{\X}f(W_i \mid x)P_{Y_j^*}(\d x)
			\end{aligned}
			\]	
		\end{enumerate}
		\item[(c)] Sample $\bm{Y}^t$ similarly to point (b).
	\end{enumerate}
	Once a sample of $(X_{i})_{i=1}^{n}$ and $(Y_{j})_{j=1}^{m}$ is available, sampling new observations $X_{n+1}$ and $Y_{n+1}$ proceeds as explained in Section $S3.1$.}

\section{\textcolor{black}{Additional simulation studies}}
\subsection{\textcolor{black}{Additional simulation scenarios}}
\textcolor{black}{We consider the same setting of Section $6.2$ in the main manuscript, with different data generating distributions. Formally we have
	\[
	W_i \simiid p(\cdot-10), \quad \mbox{var}_j \simiid p(\cdot-v),
	\]
	where $v \in [-16, 16]$ and $p(\cdot)$ is the density function of a zero mean random variable. In the main manuscript we let $p(\cdot) = N(\cdot \mid 0,1)$, while here we consider three different choices
	\[
	p_1(\cdot) = \text{Exp}(\cdot \mid 1), \quad p_2(\cdot) = 0.5N(\cdot \mid 5,1)+0.5N(\cdot \mid -5, 1), \quad p_3(\cdot) = t(\cdot \mid 3),
	\]
	where $t(\cdot \mid q)$ denotes the density of a Student's t distribution with $q$ degrees of freedom. We let $i = 1, \dots, 20$, $j = 1, \dots, 100$ and consider the same nonparametric models of Section $6.2$, with Gaussian kernel. Therefore, the prior specification is misspecified in the first and third case, with different tail behaviours of the kernel with respect to the true data generating mechanism. This implies a more complex behaviour of the latent clustering structure: indeed the posterior distribution places positive mass to more than one clusters, in order to accommodate for the misspecification.}
\begin{figure}[h!]
	\centering
	\captionsetup{width=\linewidth}
	\includegraphics[width=.49\textwidth]{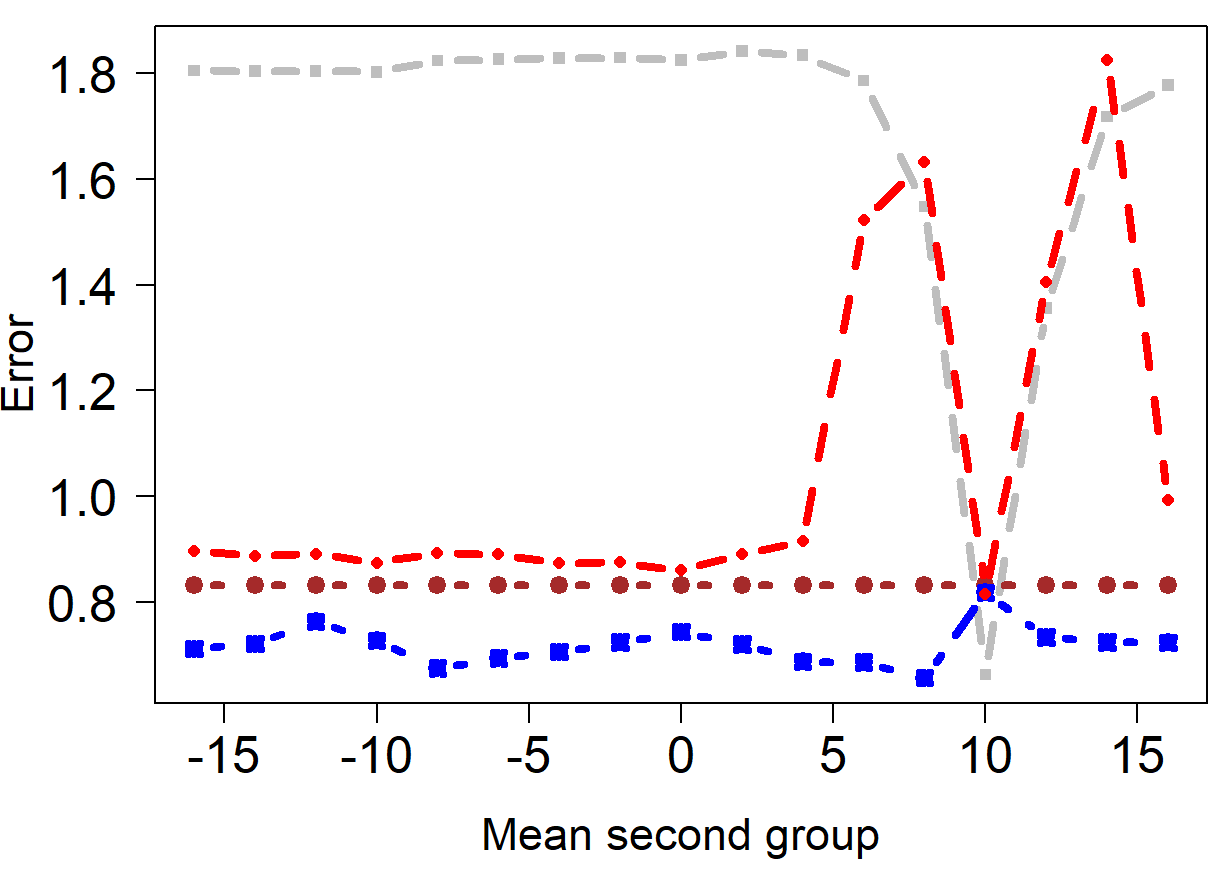} \,
	\includegraphics[width=.49\textwidth]{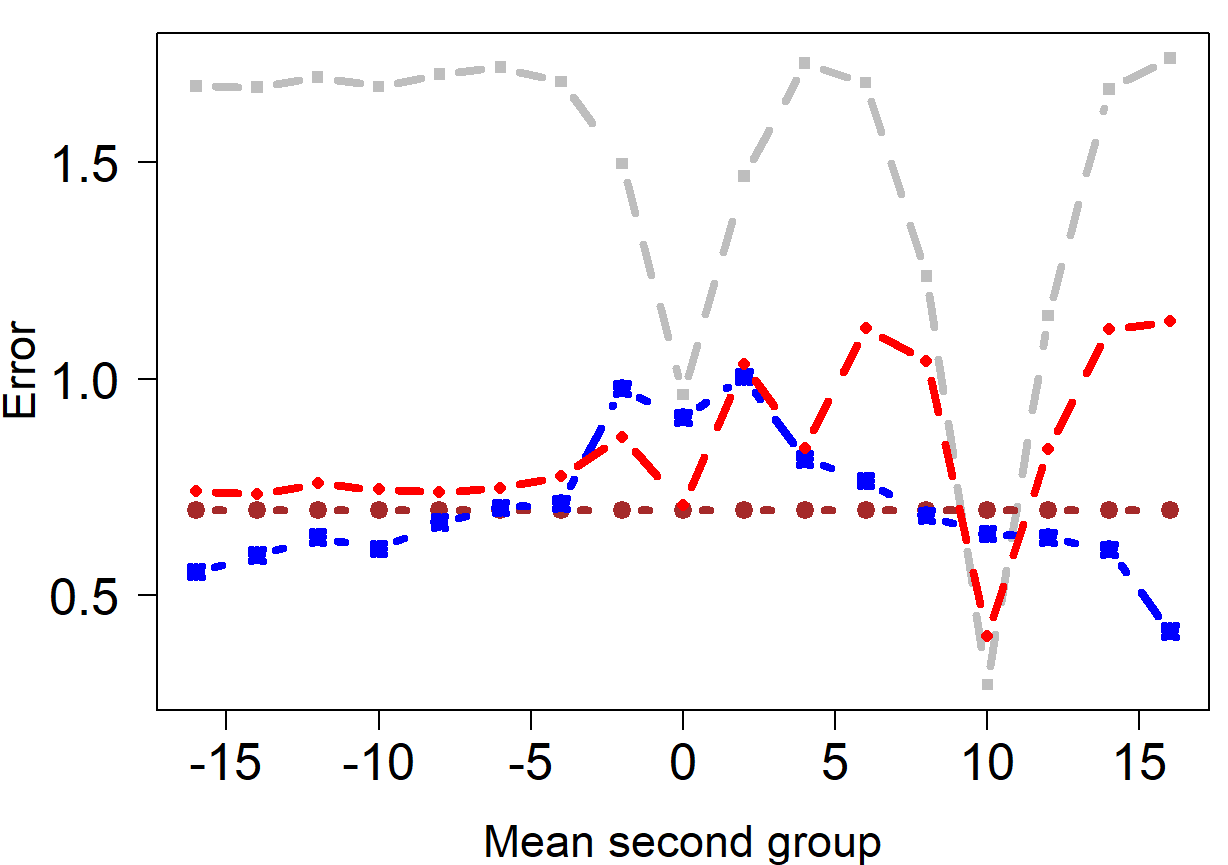}\\
	\includegraphics[width=.49\textwidth]{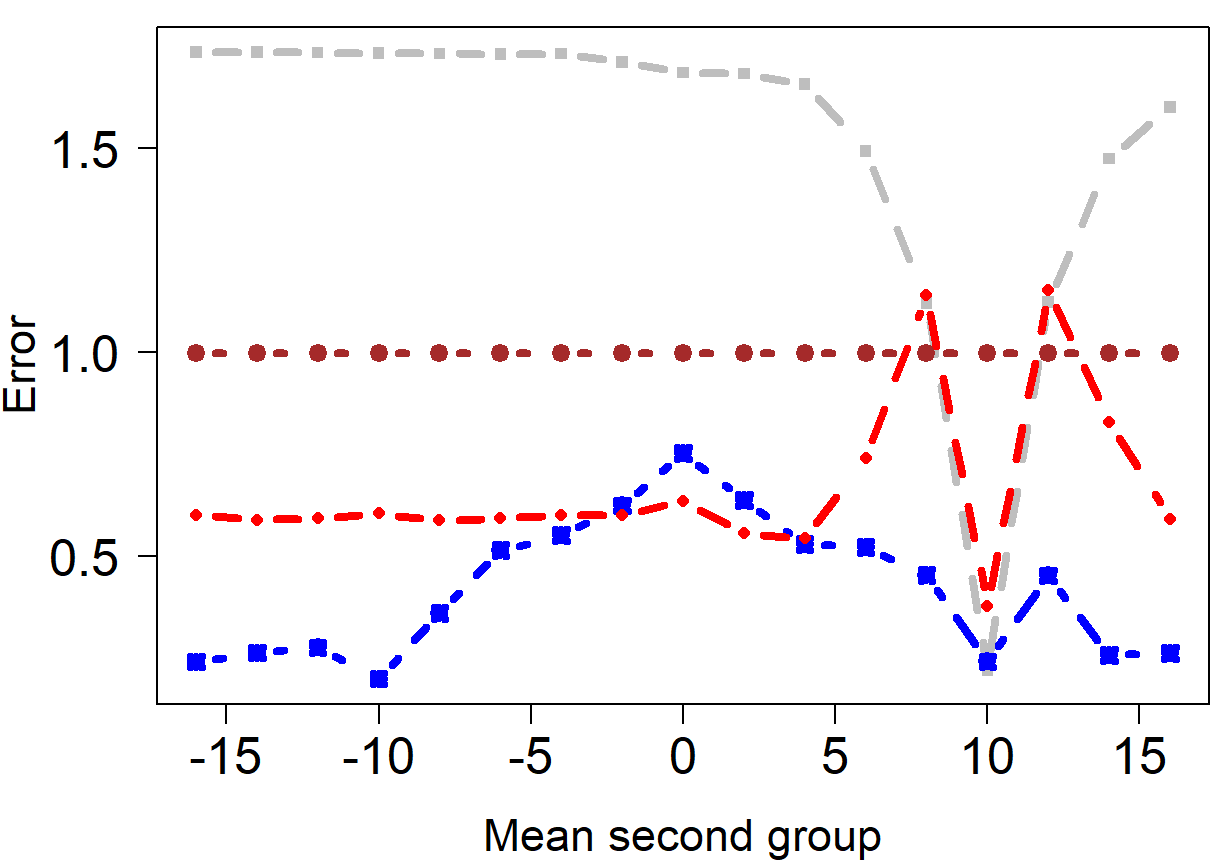} \,
	\includegraphics[width=.35\textwidth]{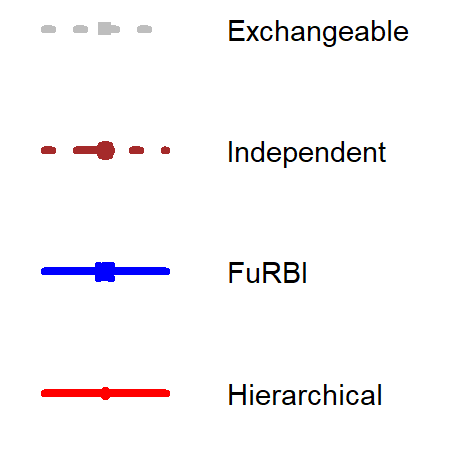}
	
	\caption{ \textcolor{black}{Mean integrated error (computed on a grid and as the median over $50$ different samples) for the four models, as the true mean of the second group varies. Rotating clockwise from the top left panel: data generated from shifted exponential, mixtures of two Gaussians and shifted Student's t distributions.}}
	\label{fig_sim1}
\end{figure}
\textcolor{black}{The mean integrated error for the three cases is depicted in Figure \ref{fig_sim1}, for different values of $v$. The interpretation is similar to the one discussed in Section $6.2$: the FuRBI specification yields an advantage especially when $v$ is far from $0$, corresponding to the prior mean, and from $10$, when the means of the two groups coincide. Indeed, in the first case the borrowing provides little information, while in the second one exchangeability holds.}

\textcolor{black}{The second setting, corresponding to the two-components mixture, apparently seems more problematic for the FuRBI model, which yields a less distinct advantage. Clearly, when $v$ is close to zero the exchangeable and hierarchical models are favoured, since the two true distributions share one of the modes. Moreover, the availability of only $20$ observations for the first group makes it more difficult to both detect the presence of two clusters and tune appropriately the correlation. Indeed, the left part of figure \ref{fig_sim2} depicts the error when $50$ observations for the first group are collected: as expected, the performances of the FuRBI approach significantly improve.}

\begin{figure}[h!]
	\centering
	\captionsetup{width=\linewidth}
	\includegraphics[width=0.49\textwidth]{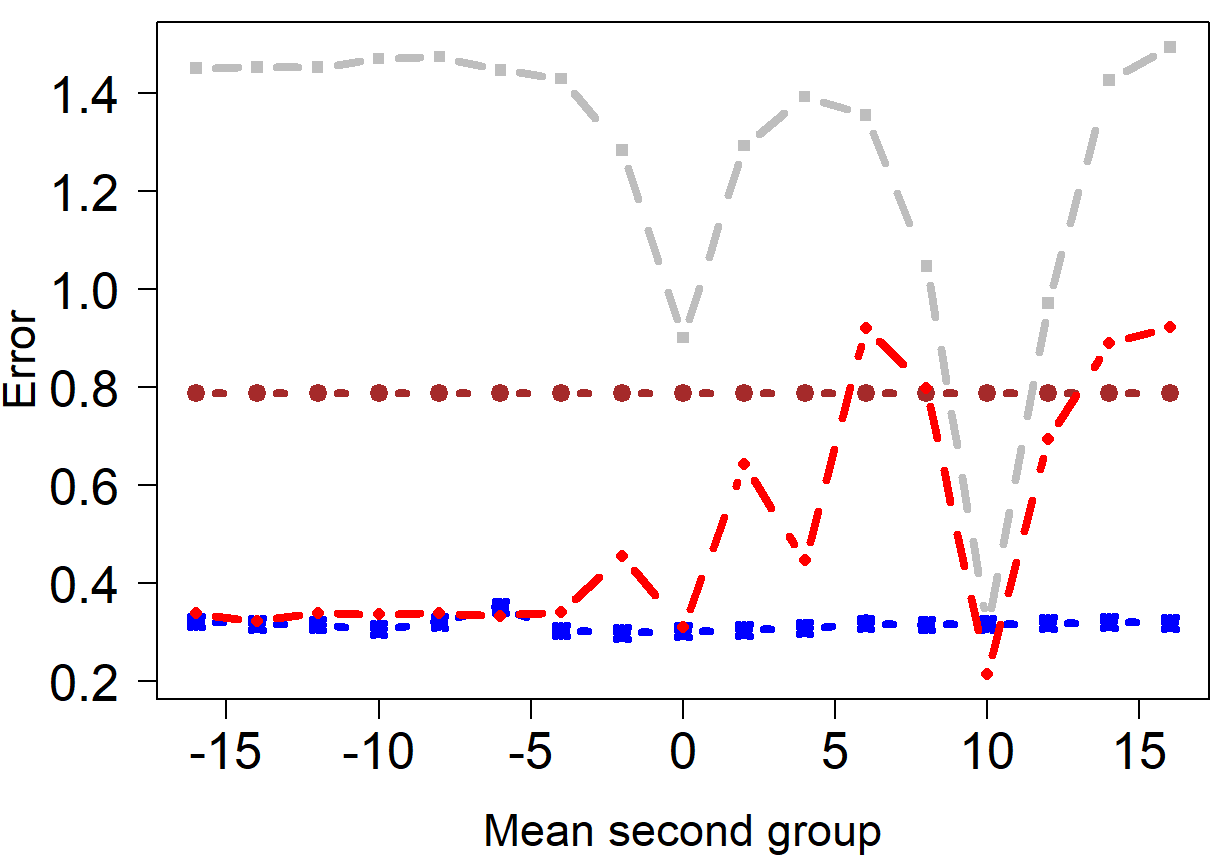}
	\includegraphics[width=0.49\textwidth]{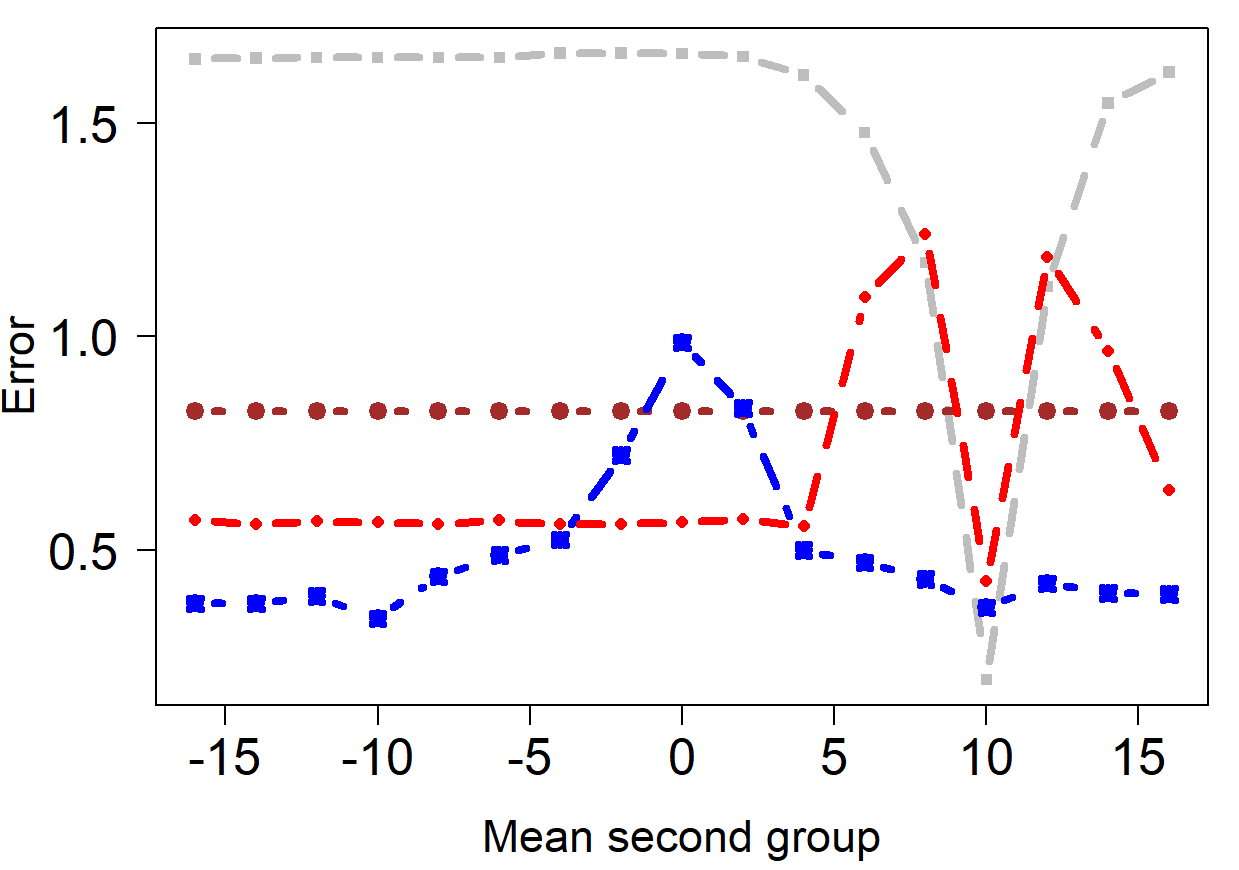}
	\caption{\textcolor{black}{Mean integrated error (computed on a grid and as the median over $50$ different samples) for the four models, as the true mean of the second group varies. Left: data generated from mixtures of two Gaussians ($50$ observations for the first group). Right: data generated from  shifted Student's t (first group) and shifted exponential (second group) distributions.}}
	\label{fig_sim2}
\end{figure}
\textcolor{black}{Finally, the right part of figure \ref{fig_sim2} shows the error when the two distributions are different: the first group is endowed with a Student's t density, while the second one is exponentially distributed. Notice that the two groups are now very far in distributional sense, especially in terms of tail behaviour. The plot indicates an interesting trade-off: when $v$ is far from the prior mean (i.e. $0$) the FuRBI approach allows to alleviate the prior misspecification, otherwise borrowing information from very different distributions may be detrimental.}

\subsection{\textcolor{black}{Logit stick-breaking prior and borrowing of information}}

\begin{figure}[H]
	\centering
	\captionsetup{width=\linewidth}
	\includegraphics[width=0.45\textwidth]{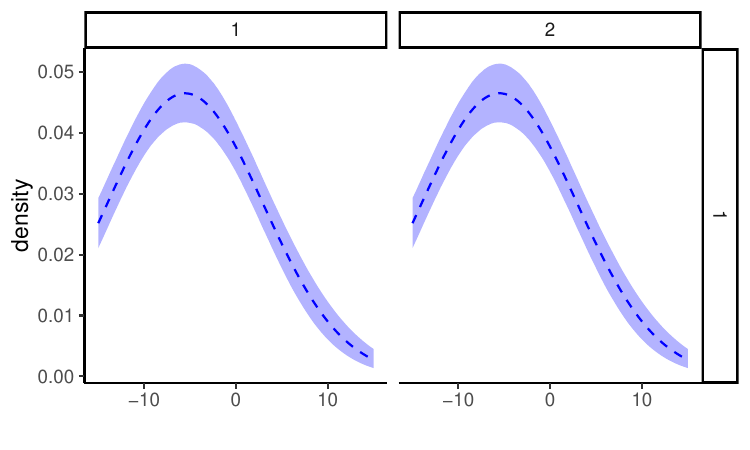}
	\includegraphics[width=0.45\textwidth]{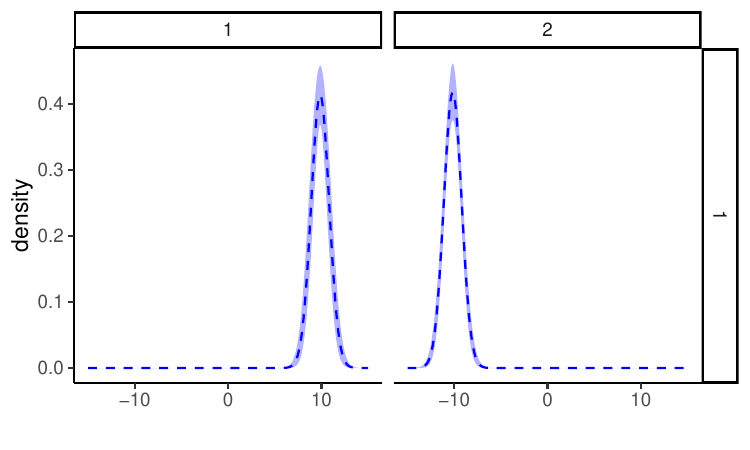}
	\caption{\textcolor{black}{Left panel: density estimates for the logit stick-breaking model with only dependent weights, and thus, $\rho_0=1$. Right panel: density estimates for the logit stick-breaking model with dependent weights and atoms. Shaded areas denote 95\% credible intervals.
			Data are simulated according to $ W_i \simiid N(\cdot \mid 10,1)$, for $i = 1, \dots, 20$ (for sample n.1), and
			$\mbox{var}_j \simiid N(\cdot \mid -10, 1)$, for $j = 1, \dots, 100$ (for sample n.2).}}
	\label{fig:LSBP}
\end{figure}
\textcolor{black}{Figure \ref{fig:LSBP} is based on the same data of Section $6.2$. See \cite{rigon2021ogit} for the model and the associated algorithm. }
\section{Predicting stocks and bonds returns: additional results}

\subsection{Density estimation for bond returns}
\begin{figure}[H]
	\centering
	\begin{subfigure}[b]{0.3\textwidth}
		\centering
		\includegraphics[width=\textwidth]{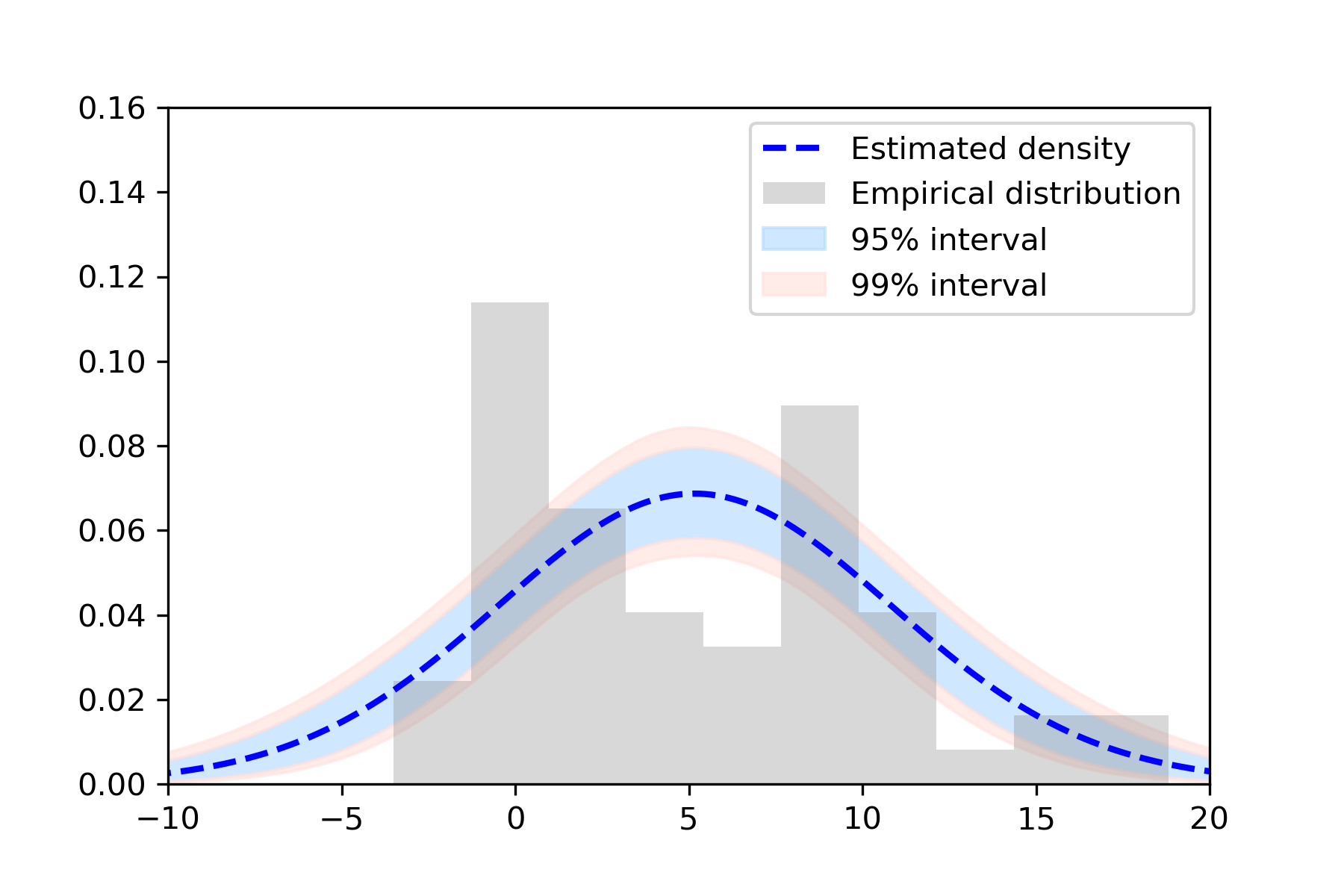}
		\caption{FuRBI full}
	\end{subfigure}
	\hfill
	\begin{subfigure}[b]{0.3\textwidth}
		\centering
		\includegraphics[width=\textwidth]{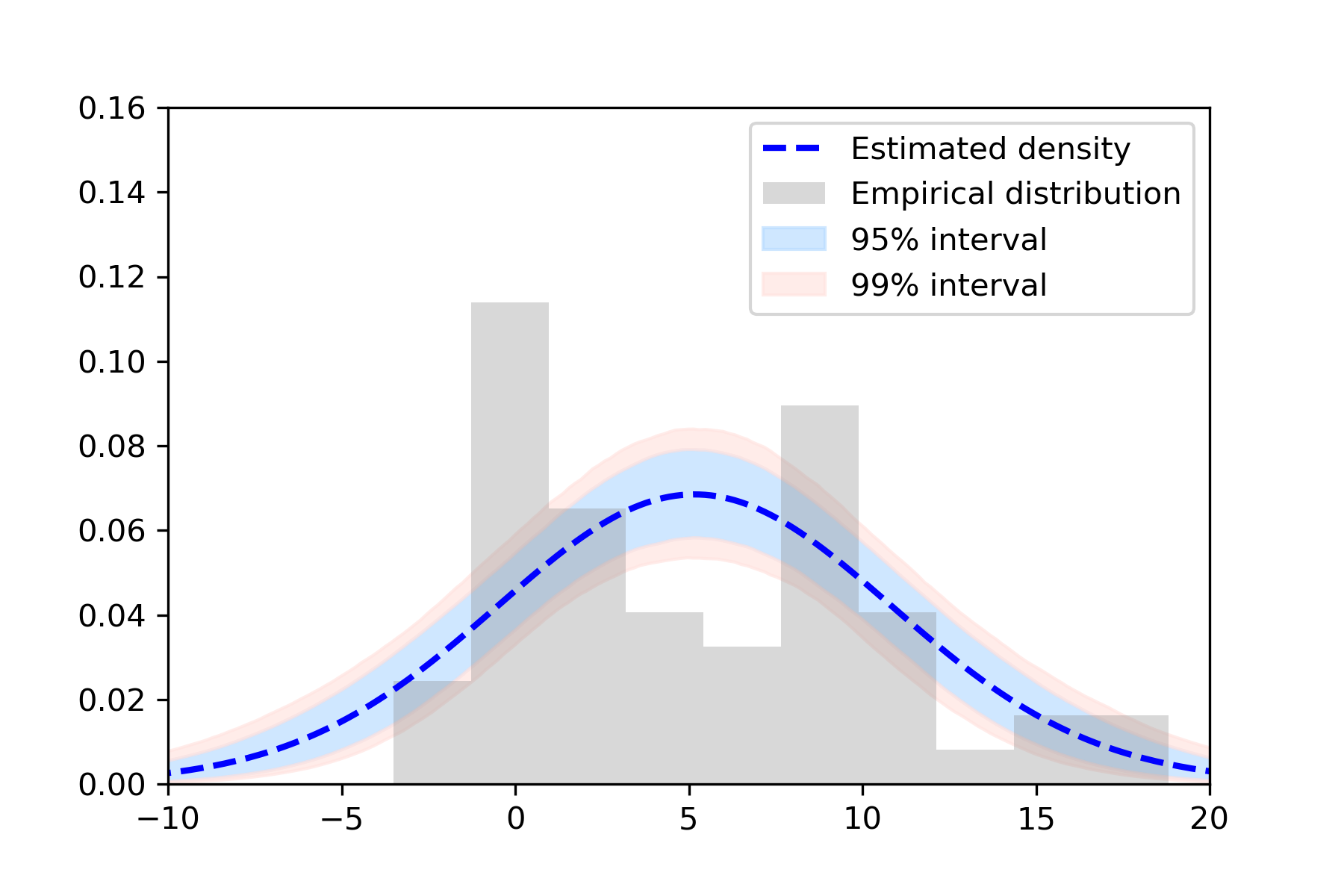}
		\caption{FuRBI $-0.95$}
	\end{subfigure}
	\hfill
	\begin{subfigure}[b]{0.3\textwidth}
		\centering
		\includegraphics[width=\textwidth]{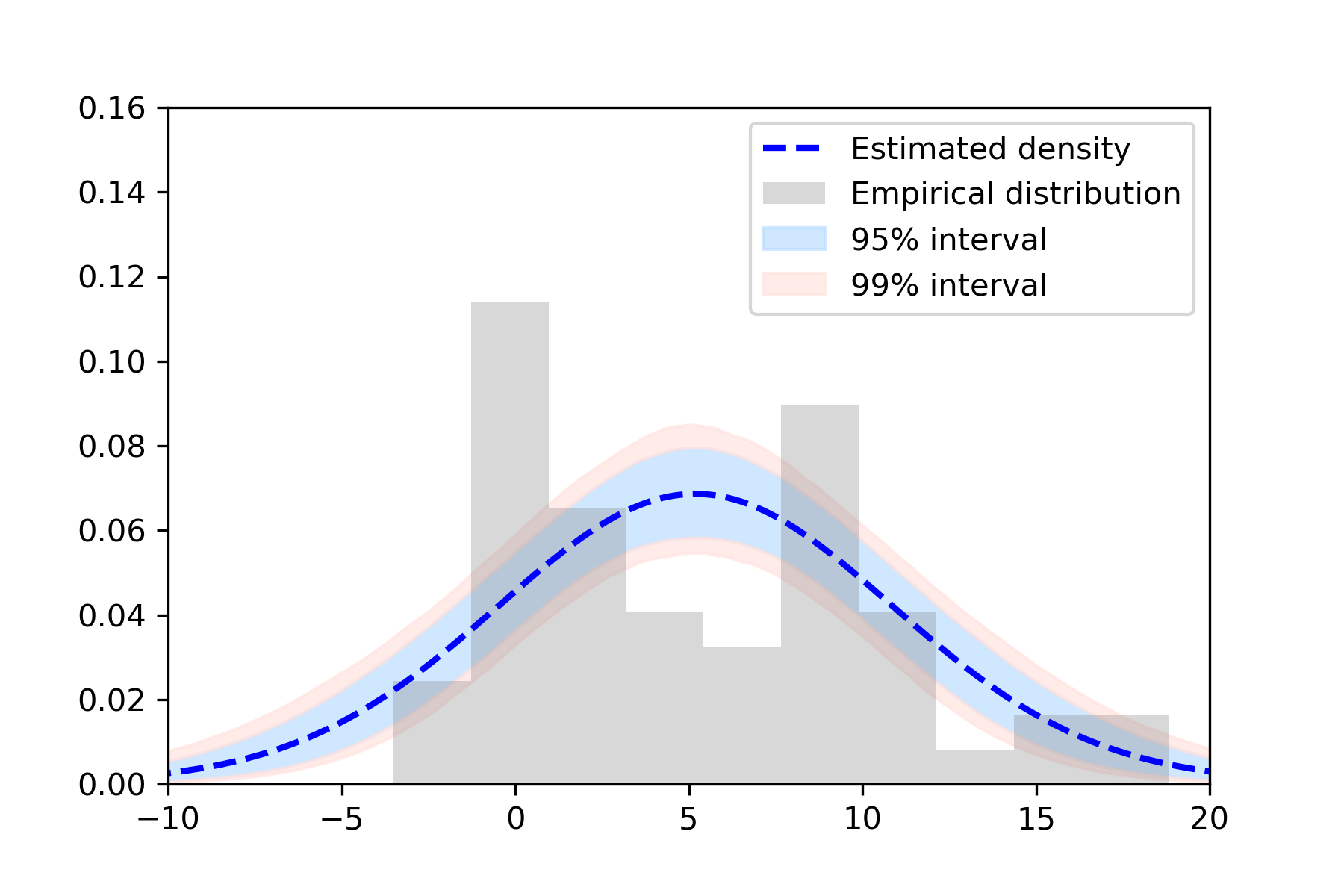}
		\caption{FuRBI $0.95$}
	\end{subfigure}
	\begin{subfigure}[b]{0.3\textwidth}
		\centering
		\includegraphics[width=\textwidth]{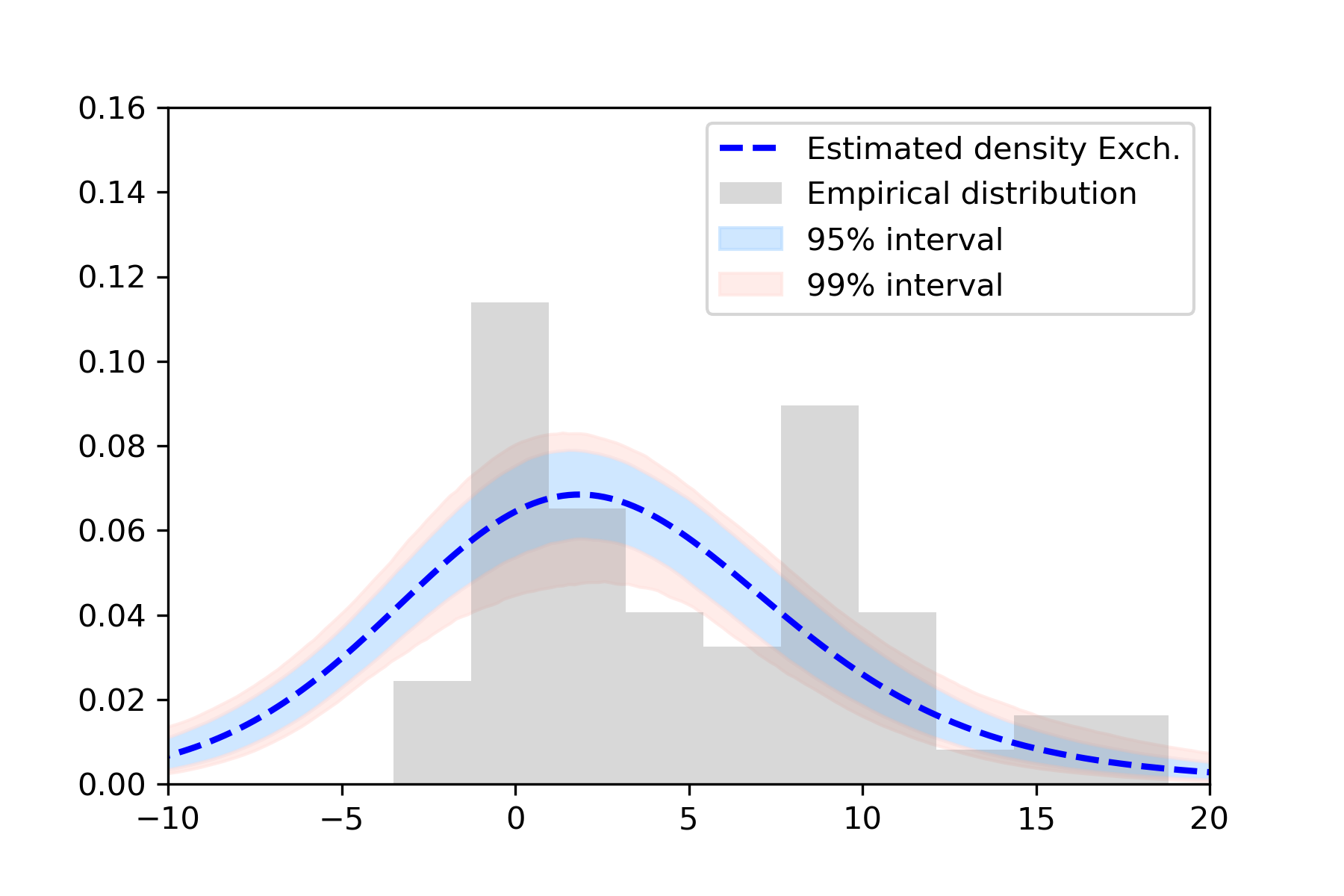}
		\caption{Exchangeable model}
	\end{subfigure}
	\hfill
	\begin{subfigure}[b]{0.3\textwidth}
		\centering
		\includegraphics[width=\textwidth]{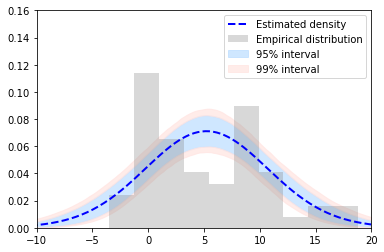}
		\caption{GM-dependent model}
	\end{subfigure}
	\hfill
	\begin{subfigure}[b]{0.3\textwidth}
		\centering
		\includegraphics[width=\textwidth]{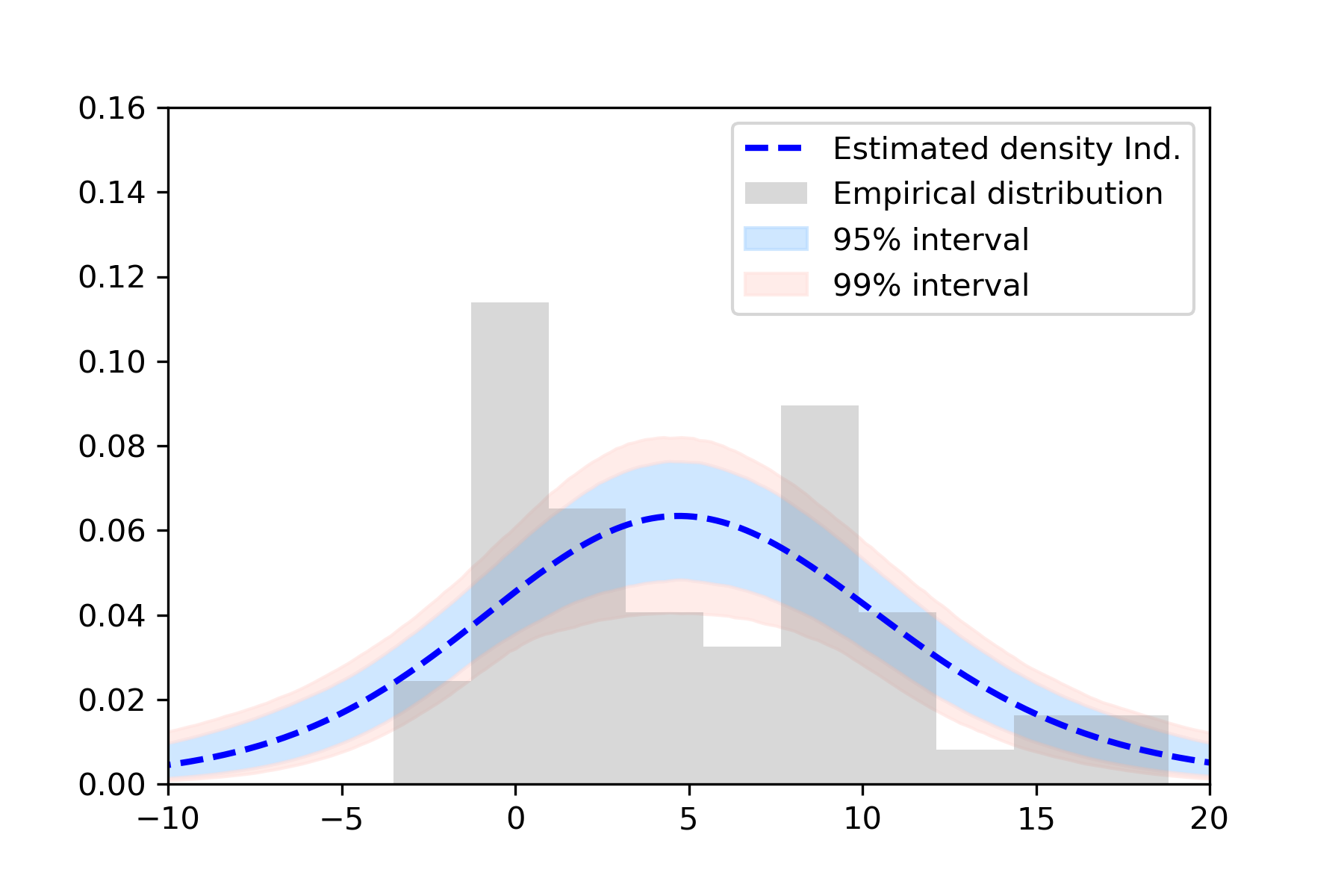}
		\caption{Independent model}
	\end{subfigure}
	\caption{\label{fig:densest} Density estimates for bonds returns.}
\end{figure}

\subsection{Sensitivity analysis}
Figure~\ref{fig:densest_sens} shows the results obtained with different specifications of the hyperparameters, which are
\begin{itemize}
	\item Specification n.1: $\lambda_j=0.1$, $\alpha_j=3$, and $\beta_j=3$, $j=1,2$,
	\item Specification n.2: $\lambda_j=0.1$, $\alpha_j=1.5$, and $\beta_j=4.5$, $j=1,2$,
	\item Specification n.3: $\lambda_j=0.01$, $\alpha_j=0.1$, and $\beta_j=0.2$, $j=1,2$.
\end{itemize} 

\begin{figure}[H]
	\centering
	\begin{subfigure}[b]{0.30\textwidth}
		\centering
		\includegraphics[width=\textwidth]{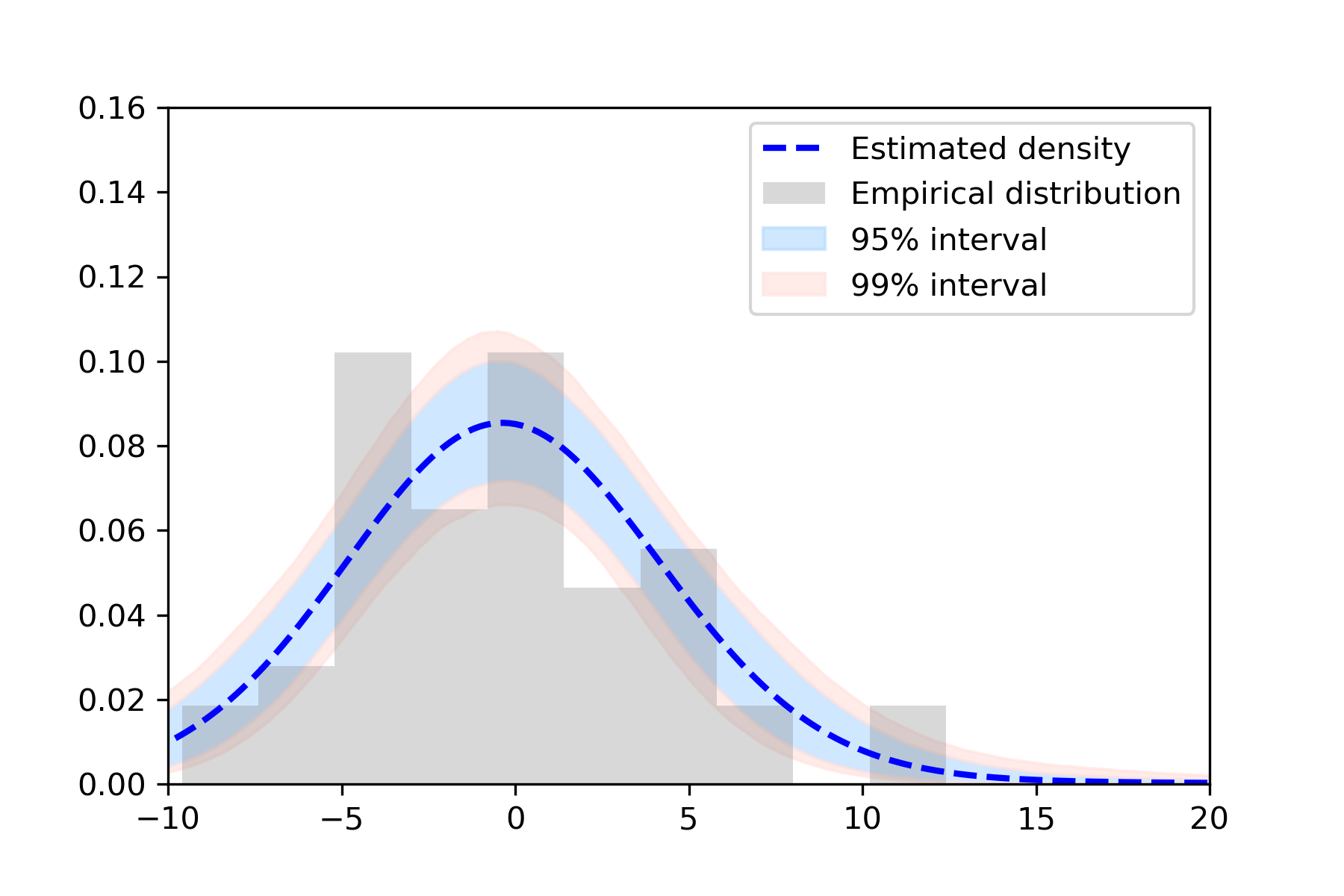}
		\caption{Spec. 1: stocks}
	\end{subfigure}
	\hfill
	\begin{subfigure}[b]{0.3\textwidth}
		\centering
		\includegraphics[width=\textwidth]{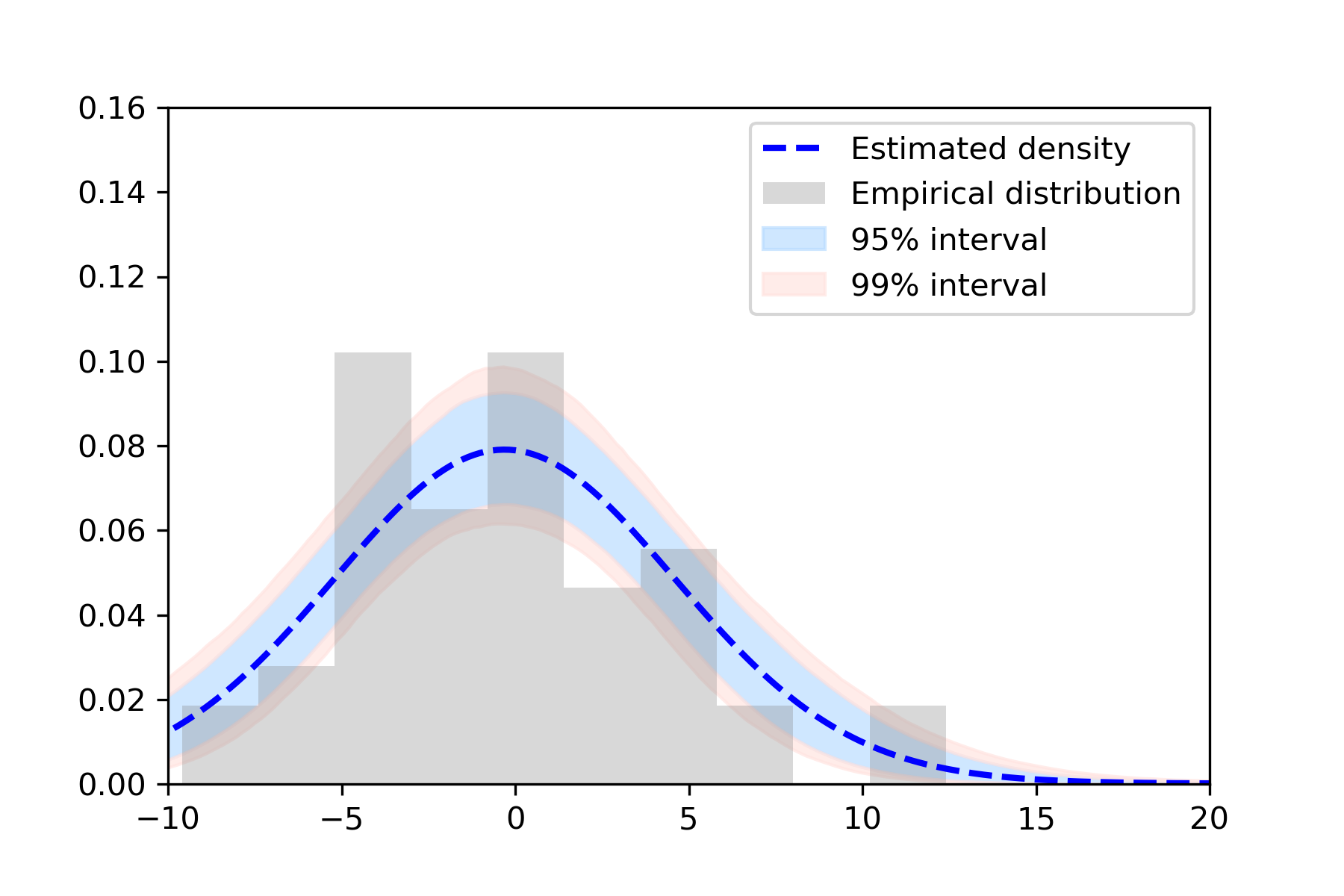}
		\caption{Spec. 2: stocks}
	\end{subfigure}
	\hfill
	\begin{subfigure}[b]{0.3\textwidth}
		\centering
		\includegraphics[width=\textwidth]{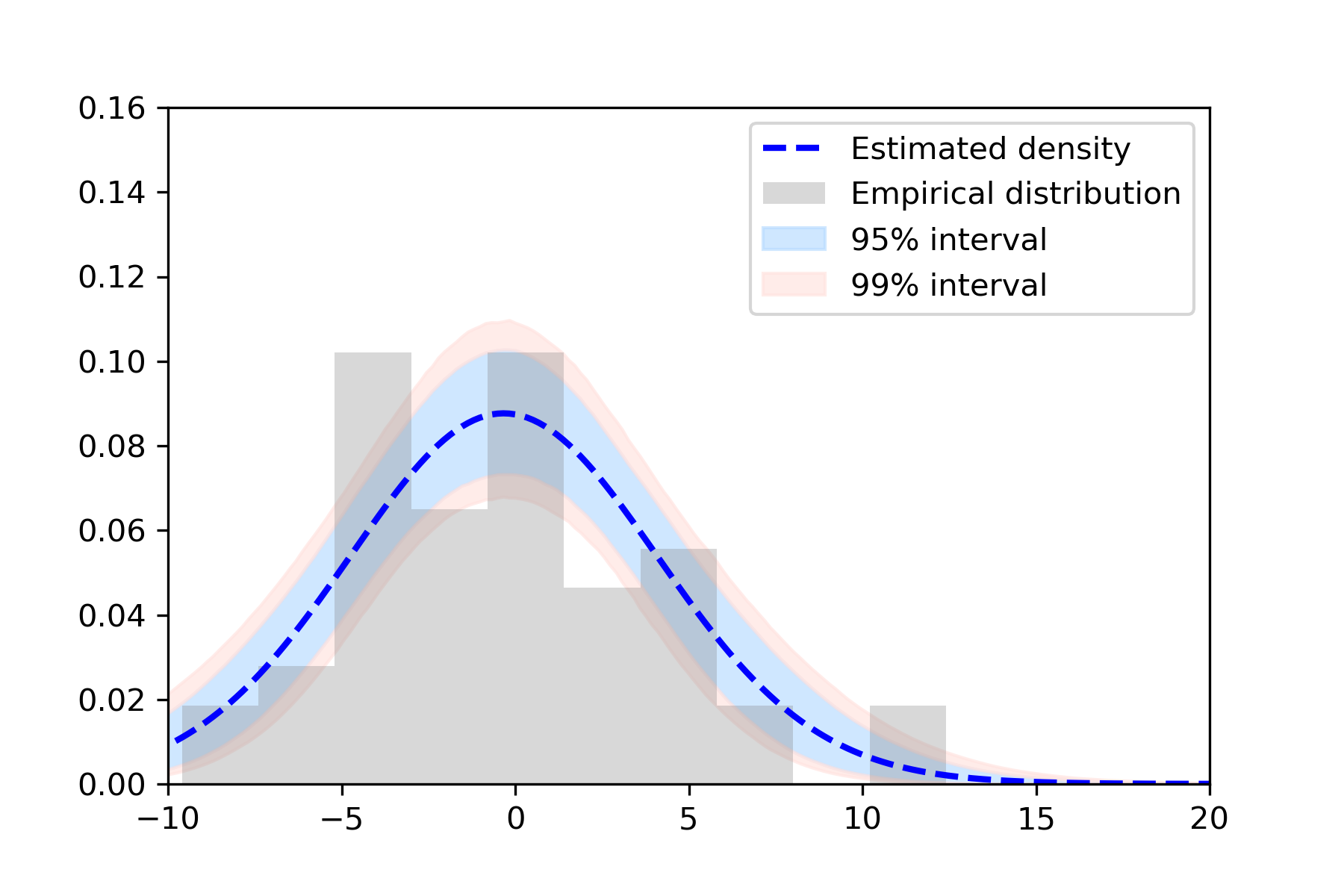}
		\caption{Spec. 3: stocks}
	\end{subfigure}\\
	\begin{subfigure}[b]{0.3\textwidth}
		\centering
		\includegraphics[width=\textwidth]{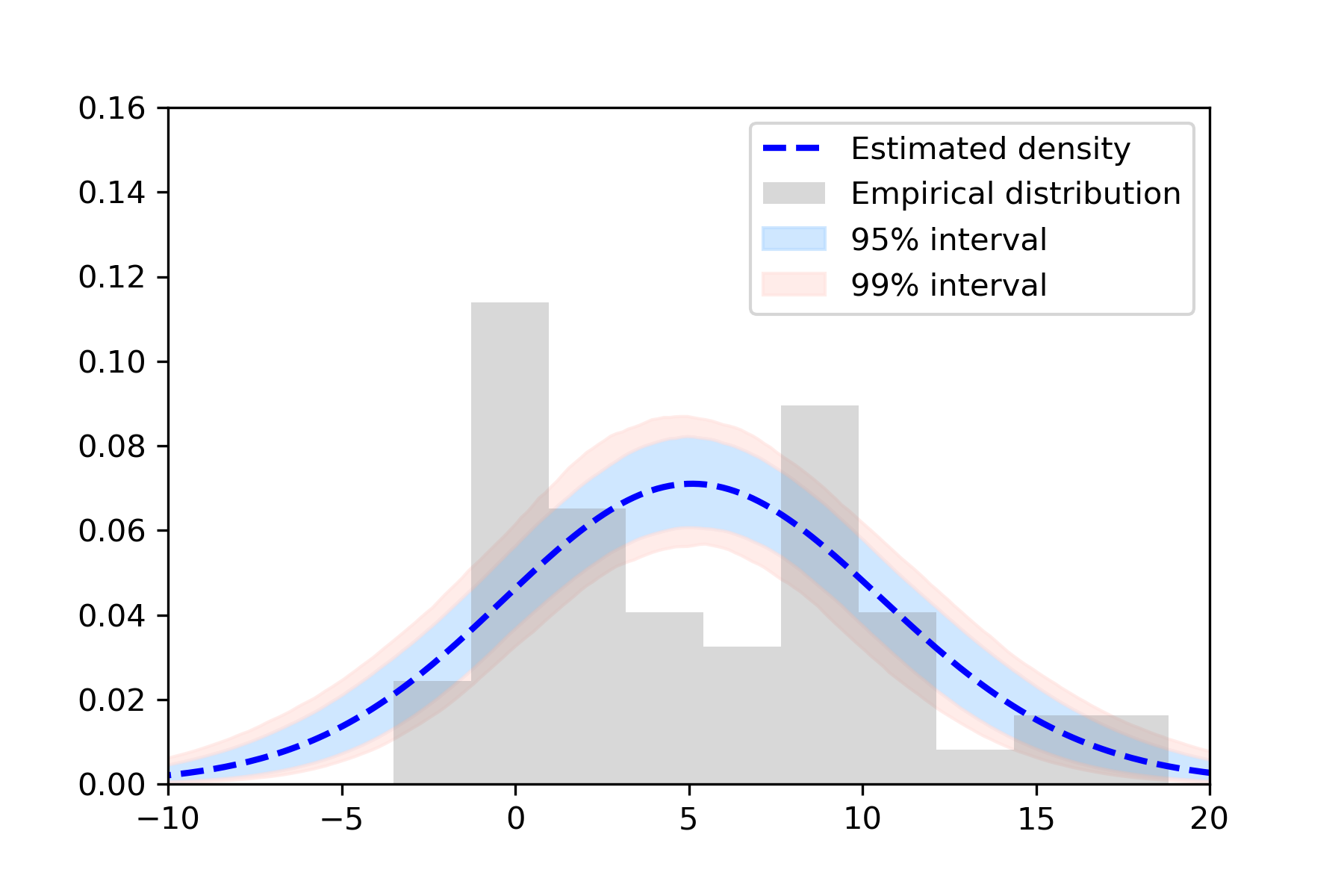}
		\caption{Spec. 1: bonds}
	\end{subfigure}
	\hfill
	\begin{subfigure}[b]{0.3\textwidth}
		\centering
		\includegraphics[width=\textwidth]{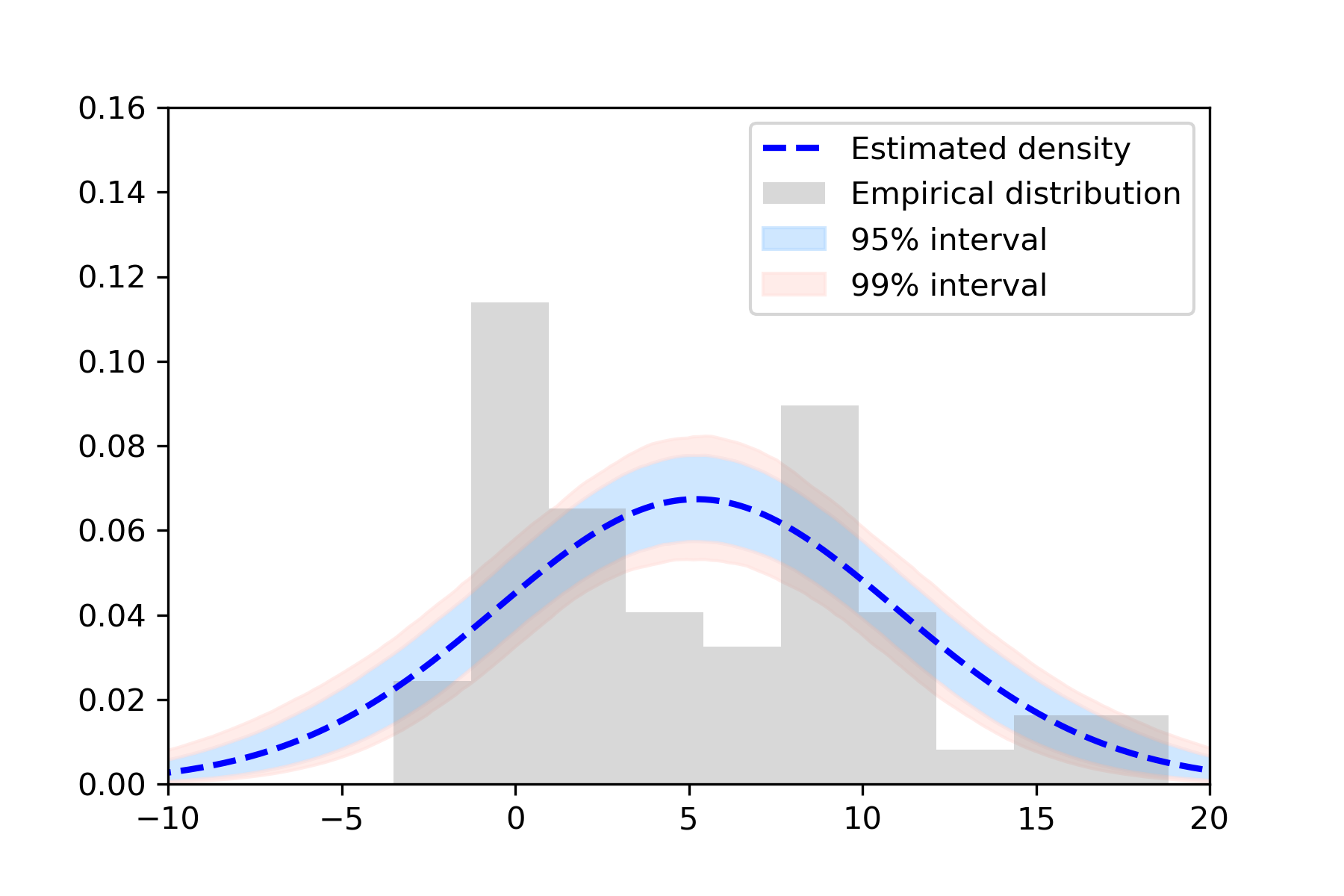}
		\caption{Spec. 2: bonds}
	\end{subfigure}
	\hfill
	\begin{subfigure}[b]{0.3\textwidth}
		\centering
		\includegraphics[width=\textwidth]{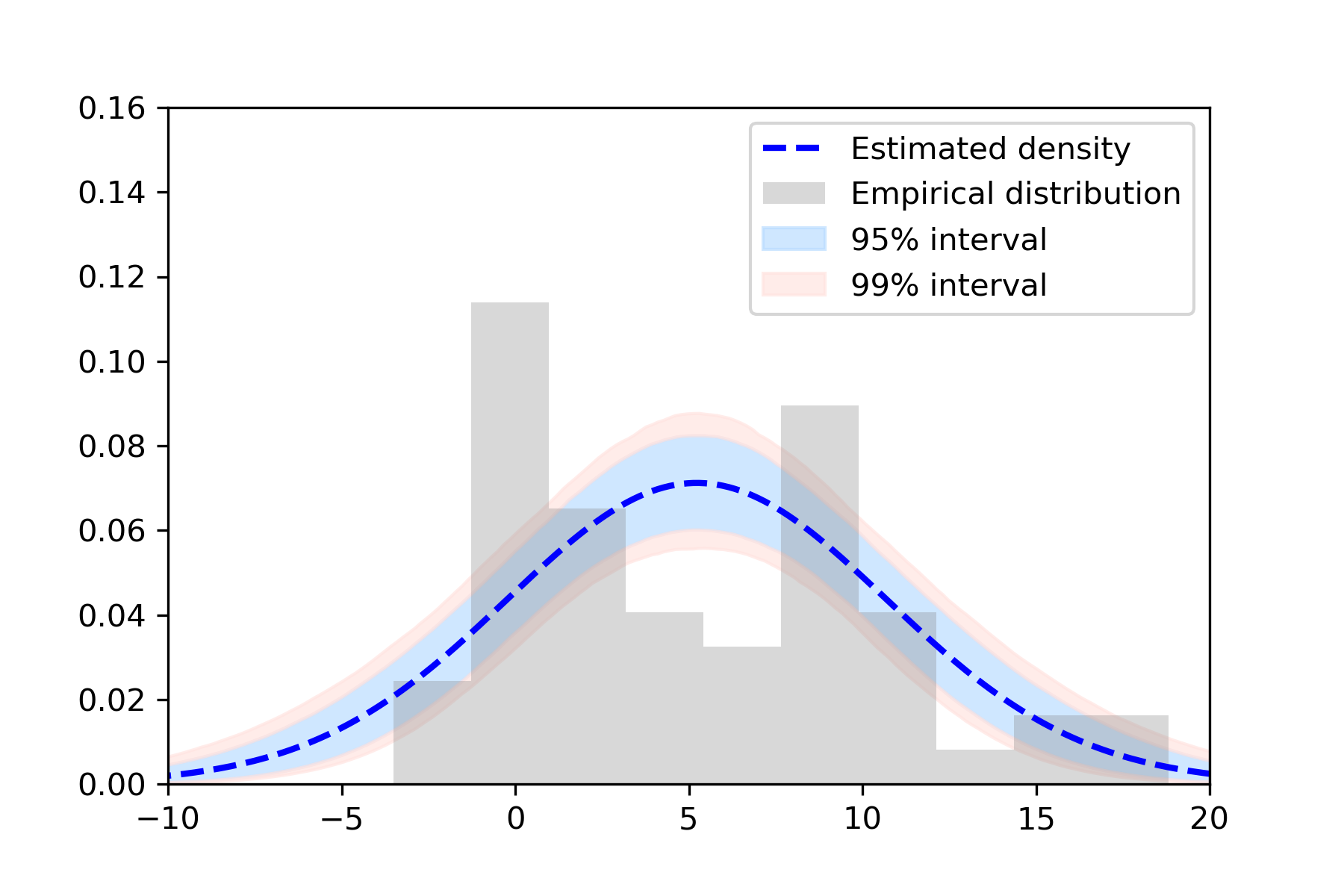}
		\caption{Spec. 3: bonds}
	\end{subfigure}
	\caption{\label{fig:densest_sens} \textcolor{black}{Sensitivity analysis: density estimates for bonds returns.}}
\end{figure}

\subsection{Posterior distribution of $\rho_0$}
\begin{figure}[H]
	\centering
	\captionsetup{width=\linewidth}
	\includegraphics[width=0.5\linewidth]{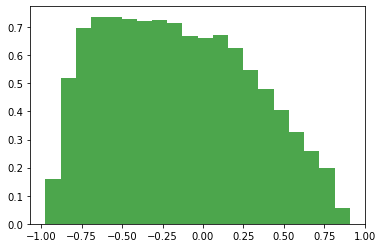}
	\caption{Posterior distribution of $\rho_0$ for the analysis in Section 6.3.}
	\label{fig:rho0posterior}
\end{figure}

\section{\textcolor{black}{Clustering multivariate data with missing entries: additional details}}
\subsection{\textcolor{black}{Choosing the hyperparameters}}
\textcolor{black}{Assume $P=3$, as in the simulation study of Section $6.4$: the general case follows accordingly. In this case $I = \{\emptyset, (1), (2), (3), (1,2), (2,3), (1,3), (1,2,3)\}$. In order to specify the prior, assumptions on the missing generating mechanism should be made. The missing completely at random (MCAR) assumption implies that each observation $W^{(x)}_{i}$, for $x \in I$, is the result of randomly eliminating entries from an (unobserved) complete observation $W_i$. For instance, $W^{(1)}_{i} = (w_{2,i}, w_{3,i})$ is obtained from a latent $W_{i} = (w_{1,i}, w_{2,i}, w_{3,i})$ after eliminating the first entry. Under this assumption the latent complete observations $W_{i}$ are exchangeable, because the original value of $W_{i}$ is independent from the mechanism that generates the missing values. Thus, there exists $\tilde q$ such that $W_i \mid \tilde q \overset{iid}{\sim} \tilde q$
	and $\tilde q_x$ is the projection of $\tilde q$ onto coordinates different than $x$, e.g.
	$\tilde q_{(1)}(\cdot, \cdot) \overset{a.s}{=} \int \tilde q(\d x_1, \cdot, \cdot)$.
	This implies that the weights of $\tilde q_x$ should be almost surely the same for every $x$.
	Instead, if the missing mechanism is not completely at random, $\tilde q_x$ can not be described as the projection of a unique $\tilde q$. Indeed the missing mechanism may be informative, leading to sample-specific features. Therefore, the choice of an additive n-FuRBIs allows $\tilde q_x$ to have sample-specific components when needed.}

\textcolor{black}{As for the baseline distribution $G_0$ on $\bm{\mu}$, suppose that an hyper-tie is sampled between an observation $(w_{2,i}, w_{3,i})$ from sample $``(1)"$ and one observation $(w_{1,i}, w_{3,i})$ from sample $``(2)"$, thus assigning the two observations to the same cluster. $G_0$ is then used to sample the corresponding locations: $(X^*_2, X^*_3)$ and $(Y^*_1, Y^*_3)$. Since we want to interpret the hyper-tie between incomplete observations as a tie between complete observations, we must have $X^*_3 = Y^*_3$, while $X^*_2$ and $Y^*_1$ are sampled jointly with a certain correlation $\rho_{1,2}$ and depending on $X^*_3$ through correlations $\rho_{1,3}$ and $\rho_{2,3}$. Therefore, since coordinates corresponding to the same original variable should be assigned the same value, $G_0$ is actually degenerate on a $P=3$ dimensional space. In the simulation and real data application $G_0$ is a $3$-variate normal, whose correlation matrix $\rho_0$ depends on correlation parameters $\rho_{12}$, $\rho_{23}$, $\rho_{13}$ on which a truncated uniform hyperprior is used, where the truncation ensures that the matrix is almost-surely positive-definite. Since the data are centered, the mean of $G_0$ is instead fixed equal to a vector of all $0$. Moreover, an independent $\text{Gamma}(3,3)$ prior is assigned to the three variances $(\sigma^2_1, \sigma^2_2, \sigma^2_3)$. Finally, the concentration parameter $\theta$ is set equal to $0.1$ in order to favor sparsity, i.e., a lower number of clusters.}
\subsection{\textcolor{black}{Simulating scenarios: missing data distribution}}
\begin{figure}[H]
	\centering
	\begin{subfigure}[b]{0.41\textwidth}
		\includegraphics[width=\textwidth]{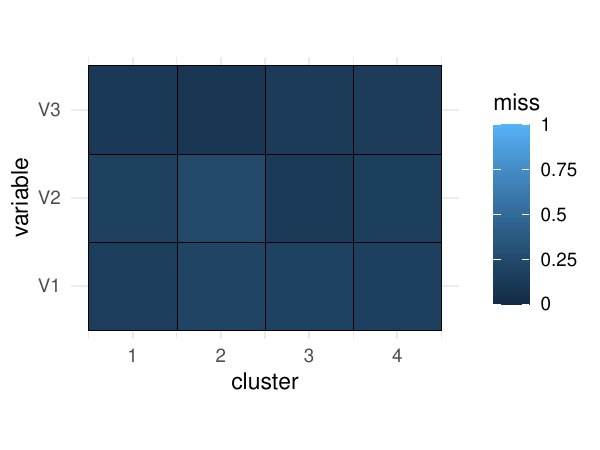}
		\caption{MCAR 16.1\% missing entries}
	\end{subfigure}
	\hfill
	\begin{subfigure}[b]{0.41\textwidth}
		\includegraphics[width=\textwidth]{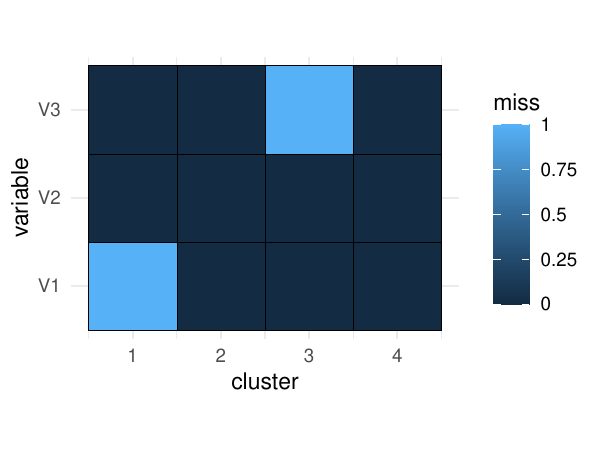}
		\caption{MNAR 17.7\% missing entries}
	\end{subfigure}
	\begin{subfigure}[b]{0.41\textwidth}
		\includegraphics[width=\textwidth]{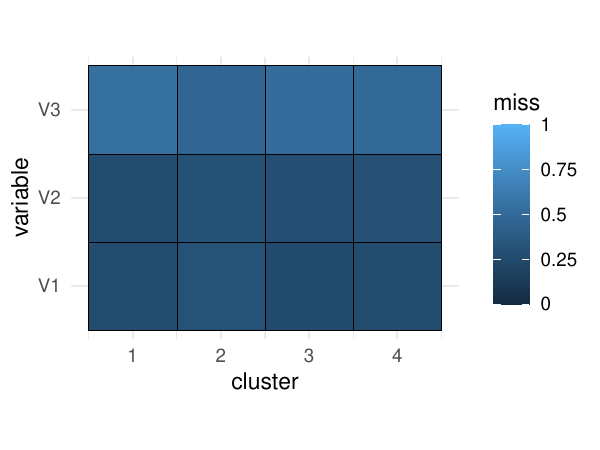}
		\caption{MCAR 35.9\% missing entries}
	\end{subfigure}
	\hfill
	\begin{subfigure}[b]{0.41\textwidth}
		\includegraphics[width=\textwidth]{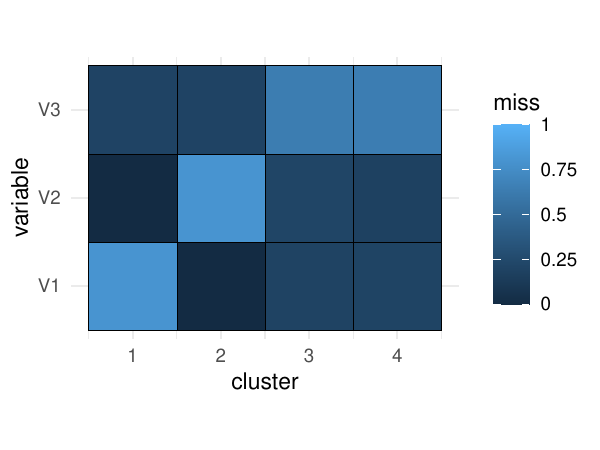}
		\caption{MNAR 34\% complete observations}
	\end{subfigure}
	\caption{\textcolor{black}{Percentages of missing entries of each variable-cluster pair.}}
	\label{fig:my_label}
\end{figure}

\section{Mixing performance of the MCMC chains}

\begin{figure}[H]
	\begin{subfigure}[b]{0.3\textwidth}
		\includegraphics[width=\textwidth]{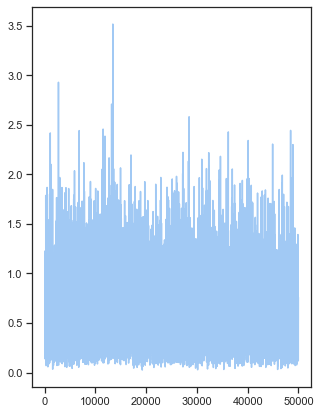}
		\subcaption{Concentration parameter $\theta$.}
	\end{subfigure}
	\hfill
	\begin{subfigure}[b]{0.3\textwidth}
		\includegraphics[width=\textwidth]{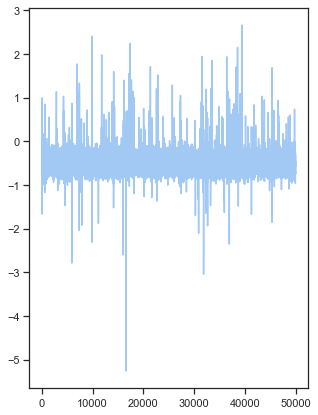}
		\subcaption{Stock n.1 cluster location.}
	\end{subfigure}
	\hfill
	\begin{subfigure}[b]{0.3\textwidth}
		\includegraphics[width=\textwidth]{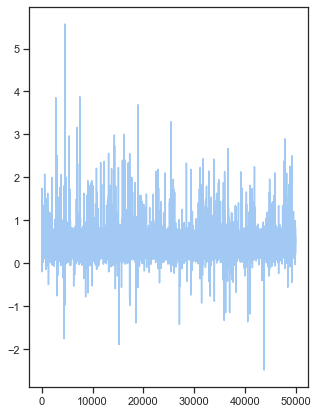}
		\subcaption{Bond n.1 cluster location.}
	\end{subfigure}
	\caption{\label{fig:traceplots_finance} \textcolor{black}{Trace plots of the MCMC chain used in the real data analysis of Section $6.3$. }} 
\end{figure}

\begin{figure}[H]
	\centering
	\begin{subfigure}[b]{0.45\textwidth}
		\includegraphics[width=\textwidth]{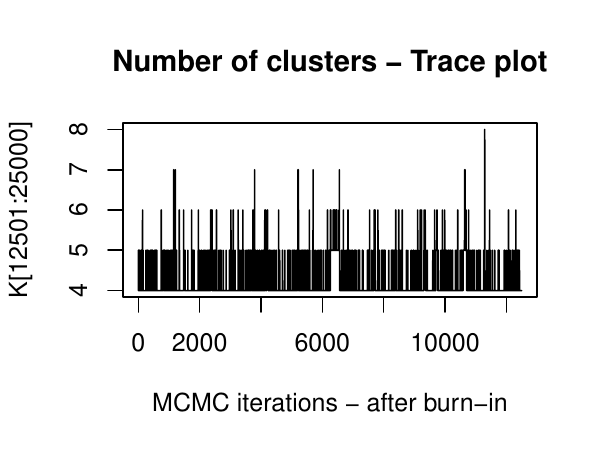}
	\end{subfigure}
	\hfill
	\begin{subfigure}[b]{0.45\textwidth}
		\includegraphics[width=\textwidth]{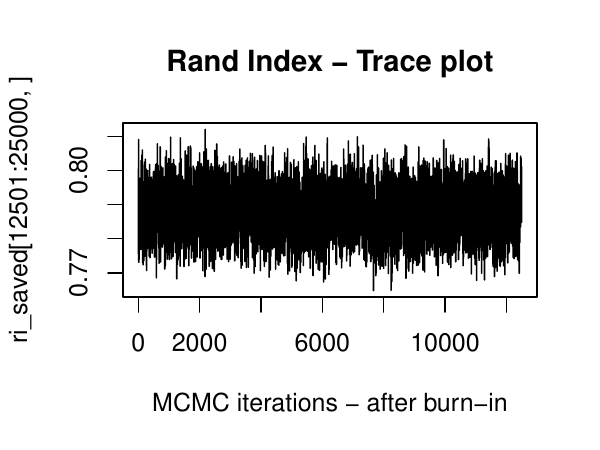}
	\end{subfigure}
	\caption{\label{fig:traceplots} \textcolor{black}{Trace plots of the MCMC chain of simulation study n.1 in Section $6.4$, for the additive n-FuRBI model with $z=0.5$. Left: number of clusters. Right: Rand index.}} 
\end{figure}

\begin{table}[H]
	\begin{tabular}{c|c|c}
		&ESS / N & ESS / N \\
		Model &   Rand index &  num. clusters\\
		\hline
		Additive n-FuRBIs, z = 0.2 &0.1957&0.0518\\
		Additive n-FuRBIs, z = 0.5 &0.1994&0.0413\\
		Additive n-FuRBIs, z = 0.8 &0.1253&0.0596\\
		DPM &0.1623&0.0227\\
	\end{tabular}
	\caption{\label{tab:K} \textcolor{black}{Effective Sample Size (ESS) per iteration in simulation study n.1 of Section $6.4$ with $1,000$ observations.}}
\end{table}

\begin{table}[H]
	\resizebox{\linewidth}{!}{
		\begin{tabular}{c|c|c|c|c|c|c}
			& total& dimension &Type of & code &average time \\
			&sample size& of data point&algorithm& language& per iter (in sec)\\
			\hline
			Financial data - \footnotesize{Sec. 6.3} & $n=104$ & $1$ & marginal & Python&0.12\\
			Simulation studies - \footnotesize{Sec. 6.4} & $n=1000$ & $3$& marginal & \textsc{R}& 2.41\\
			Brandsma data - \footnotesize{Sec. 6.4}& $n=4106$&$4$ & marginal &\textsc{R}& 8.75\\
		\end{tabular}
		\caption{\label{comp_time} \textcolor{black}{Computational time in second per one iteration of the MCMC chain with n-FuRBIs. Codes are run on an Intel Xeon W-1250 processor. Note that the in the last two lines not only the sample size is higher but also the data are multivariate.}}}
\end{table}

\end{document}